\documentclass[11pt,letterpaper]{article}

\usepackage[top=1in, bottom=1in, left=1in, right=1in]{geometry}

\makeatletter

\usepackage{algorithm,algorithmicx}
\usepackage[noend]{algpseudocode}
\renewcommand{\algorithmiccomment}[1]{\bgroup\hfill\tiny//~#1\egroup}

\usepackage[american]{babel}

\usepackage{dsfont}
\usepackage{amsthm}
\usepackage{amsmath}
\usepackage{amssymb} %,bm,bbm}
\usepackage{xspace}
\usepackage{url}
\usepackage[usenames,dvipsnames]{xcolor}
\usepackage{enumerate}

\usepackage{thmtools, thm-restate}
\declaretheorem[numberwithin=section]{theorem}

\declaretheorem[sibling=theorem]{lemma}
\declaretheorem[sibling=theorem]{corollary}

\declaretheorem[sibling=theorem]{remark}

\declaretheorem[sibling=theorem]{definition}

\usepackage{multirow}

\usepackage{longtable}           %for long table - more than a single page
\usepackage[latin1]{inputenc}
\usepackage[markup=nocolor]{changes}

\usepackage[nottoc]{tocbibind}

\usepackage[textsize=tiny]{todonotes}

\newcommand{\para}[1]{\paragraph{#1}} 
\newcommand{\paras}[1]{\textit{#1}}

\newcommand{\R}{\mathbb{R}}
\newcommand{\Z}{\mathbb{Z}}
\newcommand{\N}{\mathbb{N}}

\newcommand{\Vol}{\textsc{vol}}

\newcommand{\Ex}{\mathbb{E}}

\newcommand{\x}[1]{\ensuremath{\mathsf{x}_{#1}}}
\newcommand{\w}[1]{\ensuremath{\mathsf{w}_{#1}}}

\newcommand{\wmin}{\ensuremath{\mathsf{w}_{\min}}}

\def\argmax{\operatornamewithlimits{argmax}}

\newcommand{\Space}{\ensuremath{\mathds{T}^d}}
\newcommand{\Pois}{\textsc{Pois}}
\newcommand{\probSpace}{\ensuremath{V}}

\renewcommand{\epsilon}{\ensuremath{\varepsilon}}
\newcommand{\eps}{\ensuremath{\varepsilon}}

\newcommand{\temporary}[1]{}

\makeatletter
\newcommand{\pushright}[1]{\ifmeasuring@#1\else\omit\hfill$\displaystyle#1$\fi\ignorespaces}
\newcommand{\pushleft}[1]{\ifmeasuring@#1\else\omit$\displaystyle#1$\hfill\fi\ignorespaces}
\makeatother

\begin{document}

\title{Greedy Routing and the Algorithmic Small-World Phenomenon}
\author{Karl Bringmann\thanks{Max-Planck-Institute for Informatics, Saarbr\"ucken, Germany, \texttt{kbringma@mpi-inf.mpg.de}} \and Ralph Keusch\thanks{Department of Computer Science, ETH Zurich, Switzerland, \texttt{rkeusch@inf.ethz.ch}} \and Johannes Lengler\thanks{Department of Computer Science, ETH Zurich, Switzerland, \texttt{lenglerj@inf.ethz.ch}} \and Yannic Maus\thanks{Department of Computer Science, University of Freiburg, Germany, \texttt{yannic.maus@cs.uni-freiburg.de}} \and Anisur R. Molla\thanks{Department of Computer Science, University of Freiburg, Germany, \texttt{armolla@cs.uni-freiburg.de}}}
\maketitle

\medskip
\begin{abstract}

The algorithmic small-world phenomenon, empirically established by Milgram's letter forwarding experiments from the 60s~\cite{milgram1967small}, was theoretically explained by Kleinberg in 2000~\cite{KleinbergModel}. However, from today's perspective his model has several severe 
shortcomings 
%drawbacks
that limit the applicability to real-world networks.
In order to give a more convincing explanation of the algorithmic small-world phenomenon, we study decentralized greedy routing in a more flexible random graph model (geometric inhomogeneous random graphs) which overcomes all previous shortcomings. 
%has not only theoretically good properties, 
Apart from exhibiting good properties in theory, it has also been extensively experimentally validated that this model reasonably captures real-world networks. 

In this model, the greedy routing protocol is purely distributed as each vertex only needs to know information about its direct neighbors. We prove that it succeeds with constant probability, and in case of success almost surely finds an almost shortest path of length $\Theta(\log\log n)$, where our bound is tight including the leading constant.
Moreover, we study natural local patching methods which augment greedy routing by backtracking and which do not require any global knowledge. We show that such methods can ensure success probability 1 in an asymptotically tight number of steps.

%We show that in this model, greedy routing succeeds with constant probability, and in case of success almost surely finds a path of length $\Theta(\log\log n)$, where our bound is tight including the leading constant as it matches the shortest path length.
%Our results are robust to variations of \ym{I don't like the word changes here \jl{better?}} the model parameters and the routing objective. 
%Moreover, since constant success probability is too low for technical applications, we study natural local patching methods augmenting greedy routing by backtracking. We show that such methods can ensure success probability 1 in a number of steps that is still asymptotically tight.
%
%Crucially, the greedy protocol on GIRGs is purely distributed as each vertex only needs to know information about its direct neighbors and, more surprisingly, also the patching protocols do not require any global knowledge.
% Moreover, the protocols are highly energy efficient as only one node needs to be awake at a time.

These results also address the question of Krioukov et al.~\cite{krioukov2007compact} whether there are efficient local routing protocols for the internet graph. There were promising experimental studies, but the question remained unsolved theoretically. Our results give for the first time a rigorous and analytical affirmative answer.

\end{abstract}
\pagenumbering{gobble}
\clearpage

%\clearpage
%\tableofcontents
%\clearpage
%\thispagestyle{empty}
%\newpage
%\setcounter{page}{1}\subsection
\setcounter{page}{1}
\pagenumbering{arabic}

\newpage

%\clearpage
\tableofcontents\label{tableofcontent}
\clearpage

\section{Introduction}\label{sec:introduction}

The idea that everyone is connected to everybody else through six degrees of separation, also known as the \emph{small world} phenomenon, was empirically established by the classic experiments of Milgram~\cite{milgram1967small,travers1969experimental}. Several random people were given a target person's name and address, and asked to send a letter to the target person by forwarding it to a personal acquaintance who is more likely to know the target person, iterating this process until the letter reached the target (or was lost). Among the more than 20\% successful trials, Milgram reported an average distance of about~6. This showed not only that very short paths exist between two typical persons, but also that people are able to find them without global knowledge of the network, by greedily routing the letter.
% to a suitable acquaintance.

In the last 50 years, this phenomenon inspired researchers in various disciplines such as sociology, physics, and experimental and theoretical computer science to study networks with small-world properties. The first theoretical explanation of the algorithmic small-world phenomenon was the seminal work of Kleinberg~\cite{KleinbergModel} who showed that greedy routing succeeds in $O(\log^2 n)$ steps on certain random graphs. However, as we will discuss in Section~\ref{sec:choiceofmodel}, his model has several severe shortcomings, that have not been resolved simultaneously in subsequent work~\cite{franceschetti2006navigation,martel2004analyzing,
fraigniaud2006eclecticism,manku2004know,chaintreau2008networks}: routing succeeds only if a perfect lattice structure is contained in the network, the result is fragile with respect to small changes in the model, and the graph model is unrealistically homogeneous. 

In order to give a more cogent theoretical explanation for Milgram's result, we study greedy routing in a more flexible and more convincing graph model: In \emph{geometric inhomogeneous random graphs (GIRGs)} \cite{bringmann2015euclideanGIRG}, every vertex draws independently a random position in a geometric space and a positive weight according to a power law which represents the node's influence. Two nodes are then more likely to connect if they have larger weights and if they are geometrically close.\footnote{Note that the geometric coordinates may capture more than just real world coordinates, like occupation and interests~\cite{LibenNowellPNAS05}.}
We picked this model for the following reasons. 

\begin{enumerate}[(1)] \label{enum:reasonsgirg}
\item GIRGs are a \emph{natural extension of Kleinberg's model} that overcome its three main shortcomings. 

\item Among the random graph models that have been proposed to model real-world networks, the GIRG model stands out because it not only \emph{theoretically} reproduces structural properties of social, economical, and technological networks, but a special case of GIRGs (hyperbolic random graphs~\cite{BogunaPK10}) has also been extensively \emph{experimentally validated}~\cite{shavitt2008hyperbolic,BogunaPK10,zhao2011efficient,lehman2016experimental,papadopoulos2015network,garcia2016hidden}. E.g., Bogu\~n\'a et al.~\cite{BogunaPK10} computed and studied a (heuristic) maximum likelihood fit of the internet graph into the hyperbolic random graph model, based on ideas from~\cite{Kleinberg07}.

\item In general it is unclear how greedy routing generalizes to \emph{inhomogeneous} geometric graphs. An ideal neighbor should optimize two objectives: being close to the target, and having large weight. Many possibilities have been suggested to resolve this conflict (e.g., in the physics community~\cite{thadakamalla2007search,huang2014navigation}). Since GIRGs have a specific connection probability for each pair of vertices, there is a \emph{natural interpretation of Milgram's instruction} to route to an ``acquaintance who is more likely than you to know the target person''~\cite{travers1969experimental}: pick the neighbor most likely being adjacent to the target. Thus, Milgram's experiment has a natural analogon on GIRGs.

\item In the digital age, the area of routing protocols plays a more and more important role. Krioukov et al.~\cite{krioukov2007compact} identified as a major open problem in this area ``whether we can devise routing protocols for the internet that, having no full view of the network topology, can still efficiently route messages''.  
This question stimulated further research on the embeddings from~\cite{BogunaPK10,Kleinberg07}, and by now routing on such graphs is experimentally well-studied~\cite{papadopoulos2015network,blasius2016efficient,BogunaK09, BogunaKC09, CvetkovskiC09, KrioukovPBV09, PapadopoulosKBV10,KrioukovPKVB}. It turned out that greedy routing \emph{works experimentally surprisingly well} on hyperbolic random graphs and GIRGs~\cite{KrioukovPBV09,KrioukovPKVB, PapadopoulosKBV10}.  
For example, the resulting greedy paths are only marginally longer than shortest paths in the network. 
Explaining these experimental findings theoretically was a driving question for many researchers in the theory community~\cite{uelithesis,bringmann2015euclideanGIRG,bringmann2015generalGIRG,
FriedrichKrohmer15,kiwi2015bound,fountoulakis1,
DBLP:conf/waw/CandelleroF14,gugelmann2012random}, but it remained open. 
\end{enumerate}
\bigskip

We thus believe that studying greedy routing in GIRGs, in each step maximizing the connection probability to the target, gives a convincing theoretical explanation for Milgram's experiments, and addresses \emph{at the same time} the question of Krioukov et al. We obtain the following results for any fixed source $s$ and target $t$ (for exact statements see Section~\ref{sec:results}).

\begin{itemize}
\item Greedy routing \emph{succeeds} with probability $\Omega(1)$. 
\end{itemize}

This is already surprising, as incorporating noisy positions in Kleinberg's model results in a tiny success probability, as we discuss in Section~\ref{sec:choiceofmodel} below. The result is also best possible, as in GIRGs the target node is isolated with constant probability.

\begin{itemize}
\item A.a.s.\footnote{We say that an event holds a.a.s.\ if it holds with probability $1-o(1)$ with $n \to \infty$.}, the \emph{stretch}\footnote{The stretch is the ratio of the lengths of the greedy path and the shortest path between source and target.} of greedy routing in case of success is $1+o(1)$. Hence, greedy routing is asymptotically optimal. In particular, since the average distance in the giant component of the GIRG model is \emph{ultra-small}, i.e., $\Theta(\log \log n)$, so are the lengths of the greedy paths. 
\end{itemize}

Arguably, $\Theta(\log \log n)$ is closer to six degrees of separation than Kleinberg's $O(\log^2 n)$, and average distances in modern networks are even smaller~\cite{backstrom2012four,facebook2016three}. All previous analyses of greedy routing in small-world models yielded at least $\Omega(\log n)$ steps, e.g.~\cite{KleinbergModel,FraigniaudSTOC10}.

 %, which might be due to the fact that (almost) all previous works were on homogeneous graphs, cf.\! Section~\ref{sec:choiceofmodel}.
\begin{itemize}
\item Our results are \emph{robust} in the model parameters, i.e., they hold for all parameter choices of the model, in particular for any power law exponents in the natural range $(2,3)$, cf.\ Section~\ref{sec:choiceofmodel}. Moreover, we show that nodes do not need to know exact weights and positions of their neighbors, but rather \emph{rough approximations suffice} for all results.
\end{itemize}

Since constant success probability is too low for technical application, it is natural to ask whether there are patching methods that use backtracking to enforce success (if source and target are in the same component). We give an affirmative answer in a very strong sense:

\begin{itemize}
\item \emph{Every patching method} that satisfies three natural criteria has \emph{success probability 1} (conditional on $s$ and $t$ being in the same component), and still \emph{has a.a.s.\ stretch $1+o(1)$}. 
\end{itemize}

This result gives (in our model) a strongly affirmative answer to Krioukov et al.'s question. Crucially, each node only needs to know the positions and weights of its direct neighbors and the geometric position of $t$, which we assume to be part of the message. Note that the greedy routing protocol is distributed and highly energy efficient since only one node needs to be awake at a time. 
%Each node only needs to know the positions and weights of its direct neighbors %(nodes could learn this information easily before the protocol starts, using only small messages),
%and the geometric position of t (which we assume to be part of the message). 
Perhaps surprisingly, the same can be achieved for the patching protocols, 
%do not require global knowledge of the vertices, 
so we can at the same time guarantee success and maintain the benefits of local protocols, with negligible additional costs. 
%Moreover, all protocols are highly energy efficient as only one node needs to be awake at a time.
%Crucially, such patching does not require global knowledge of the vertices, so we can at the same time guarantee success and maintain the benefits of local protocols, with negligible additional costs.

Finally, our results explain many findings of the extensive experimental work which has been done on greedy routing in GIRGs. We postpone the discussion to Section~\ref{sec:experiments}.

%Before we come to the formal definitions and results, we discuss related work. In particular, in Section~\ref{sec:choiceofmodel} we discuss Kleinberg's model and its shortcoming. In this section we also give a more detailed account of the GIRG model and discuss why it is more realistic than previous models. %We will give a detailled discussion of related experimental work in Section~\ref{sec:experiments}.
%We postpone the discussion of the results of Bogu\~n\'a et al.~\cite{BogunaPK10} and related papers to Section~\ref{sec:experiments}.

\subsection{Related Work and Choice of Model}\label{sec:choiceofmodel}

\paragraph{Kleinberg.} Kleinberg's model starts from an $n \times n$ lattice graph, i.e., we embed node $v = (i,j) \in [n]\times [n]$ at position $\x v = (i/n,j/n)$ in the unit square and connect any two nodes $u,v$ in Manhattan distance $\|\x u - \x v\| \le 1/n$ by an edge. Then we add for each node $u$ a constant number of additional long-range edges, where the other endpoint $v$ is chosen proportional to $\|\x u - \x v\|^{-\alpha\cdot 2}$ with decay parameter $\alpha$. This model has several shortcomings that limit the applicability to real-world networks and that are also present in subsequent work~\cite{franceschetti2006navigation,martel2004analyzing,
fraigniaud2006eclecticism,manku2004know,chaintreau2008networks,
kleinberg2002small}.

%In the theoretical study of networks, the term ``small world'' is often used equivalently with a diameter growing at most polylogarithmically with the number of nodes; the first such bound dates back to Watts and Strogatz~\cite{WattsStrogatz}. This notion ignores, however, the algorithmic aspect of the small world phenomenon. As mentioned before, the seminal work of Kleinberg~\cite{KleinbergModel} was the first to show that greedy routing succeeds in $O(\log^2 n)$ steps on certain random graphs and thus explain the \emph{algorithmic small world} phenomenon. \kb{if we want to mention Watts,Strogatz at all, then we should leave it on the first page}

\begin{itemize}
\item \textbf{Fragile Exponent:} As shown by Kleinberg~\cite{KleinbergModel}, his model is very fragile with respect to the decay parameter in  exponent of the long-range probability distribution, since changing this distribution to $\|\x u - \x v\|^{-\alpha\cdot 2}$ for any $\alpha \ne 1$ increases the expected number of steps to $n^{\Omega(1)}$. This is reasonable for $\alpha < 1$, since in this regime most neighbors of most vertices lie very far 
%, in fact, all but a polynomial fraction of vertices have no neighbor in distance at most $n^{-0.1}$, 
and thus we cannot expect greedy routing to utilize locality. However, for $\alpha > 1$ there is no intuitive argument ruling out greedy routing in general. 
In particular, investigations of trajectories of bank notes in main cites of the USA suggest a value of $\alpha = 1.59 \pm 0.02$~\cite{brockmann2006scaling}.
%, which is well in the regime $\alpha > 1$ where Kleinberg's result would show a polynomial lower bound. 
%\emph{Is there a more robust variant of Kleinberg's model on which greedy routing works for all $\alpha > 1$?}

\item \textbf{Perfect Lattice:} Kleinberg's model assumes an underlying perfect lattice substrate. In this way, every vertex knows a priori a path to the target, which is an unrealistically strong assumption. In a more realistic model each vertex might choose a random position $\x v$ in some geometric space\footnote{The positions represent real-world coordinates, but also interests and occupation.}, followed by the same edge sampling procedure as in Kleinberg's model. %, i.e., we connect every vertex to all nodes in distance at most $1/n$ and then sample neighbors from the remaining nodes proportional to $\|\x u - \x v\|^{-2}$. 
It can be shown that in this adapted model with high probability greedy routing does not reach the target, essentially since in each step the current vertex has constant probability not to have any neighbor with smaller distance to the target. Thus, the unnatural assumption of a perfect lattice (or a similar globally known structure) is crucial for Kleinberg's result. %; in a more natural noisy version his positive result fails on a fundamental level. 
%\emph{Does greedy routing work on a variant of Kleinberg's model with noisy positions?}

\item \textbf{Homogeneity:} Kleinberg's graphs are homogeneous, i.e., all vertices have the same degree. 
%\footnote{Some variants of Kleinberg's model have Poisson degree distribution~\cite{franceschetti2006navigation}, which is still almost homogeneous, as it is highly concentrated around its mean.}
In contrast, most real-world networks have been found to be \emph{scale-free}, meaning that their degree distribution follows a power law $p(k) \sim k^{-\beta}$, typically with power law exponent $\beta$ in the range $(2,3)$~\cite{dorogovtsev2002evolution,albert2002statistical} (or at least a closely related distribution~\cite{mitzenmacher2004brief}). 
\end{itemize}

Despite some effort to remove these shortcomings (see related work below) there is still no model which avoids all of them. Moreover, while the corresponding extensions of Kleinberg's model were developed in the context of routing protocols, GIRGs have been developed as a model for real-world networks. Thus there is a much larger body of evidence (theoretical and experimental) that GIRGs resemble real-world networks than there is for any other graph model in this area.

\paragraph{Geometric Inhomogeneous Random Graphs (GIRGs).} 

We now substantiate claims (1) and (2) on page~\pageref{enum:reasonsgirg}. Thereby, we further explain our reasons for picking the GIRG model.

%There are four reasons for our choice of the GIRG model, which we substantiate in the remainder of the section.
%\begin{enumerate}[(1)]
%\item As we will argue, GIRGs are a \emph{natural extension of Kleinberg's model} if we remove the three shortcomings from above. 
%\item Since we have a concrete connection probability for each pair of vertices, there is a \emph{natural interpretation of Milgram's process} on GIRGs which goes beyond purely geometric routing. 
%\item Among the random graph models that have been proposed to model social networks, the class of \emph{hyperbolic random graphs} stands out because it has not only \emph{theoretically good properties}, but it has also been \emph{experimentally validated}. Since hyperbolic random graphs are a (somewhat technical) special case of the GIRG model, all our results apply for hyperbolic random graphs.
%\item \emph{Milgram's process has been studied experimentally on hyperbolic random graphs and GIRGs} since it is connected to the routing scheme proposed in~\cite{BogunaPK10}. 
%\end{enumerate}
\begin{enumerate}[(1)]
\item Let us first describe the GIRG model, which is inspired by the classic Chung-Lu random graphs~\cite{chung2002avg,Chung02}, and compare it with Kleinberg's model. In a GIRG, each node $v$ is equipped with a random position $\x v$ in an underlying geometry (for technical simplicity, we fix the $d$-dimensional torus) and with a \emph{weight} $\w v$ which is drawn according to a power law $p(w) \sim w^{-\beta}$ and which represents the (expected) degree of vertex $v$ and thus its ``connectedness''. After drawing the positions $\x v$ and weights $\w v$, all edges are picked independently. Consider vertices $u,v$ for which we already sampled the weights $\w u,\w v$ but not yet the positions $\x u, \x v$. If the marginal probability of $u,v$ forming an edge is $\Theta(\min\{\w u \w v / n, 1\})$ then the expected degree of vertex $v$ is $\Theta(\w v)$ (as in the classic Chung-Lu random graphs). Reintroducing the positions $\x u,\x v$ into this picture, a short calculation shows that if the probability of $u,v$ forming an edge is 
$$ p_{uv} = \Theta\bigg( \min\bigg\{ \bigg( \frac{\w u \w v}{n \cdot \|\x u - \x v\|^{d}} \bigg)^{\alpha} ,1 \bigg\} \bigg),$$
for some $\alpha >1$. then the marginal probability is indeed $\Theta(\min\{\w u \w v / n, 1\})$, implying that the expected degree of node $v$ is $\Theta(\w v)$. Moreover, note that in terms of distances, $p_{uv}$ is proportional to Kleinberg's distribution $\|\x u - \x v\|^{-\alpha \cdot d}$, where we want to choose $\alpha > 1$. In this sense, the GIRG model is a natural scale-free variant of Kleinberg's model that removes the implausible assumption of a perfect grid. As our results do not depend on the concrete choice of $\alpha$, we also overcome the problem of fragile exponents.

\item A large portion of the theoretical random graph literature of the last 10 years (e.g. \cite{spatialpreferred,jacob2013spatial,barabasi2012network, ferretti2014duality,BogunaPK10,fountoulakis1,
deijfen2013scale,bringmann2015generalGIRG,bringmann2015euclideanGIRG,
DBLP:conf/waw/CandelleroF14,gugelmann2012random,
FriedrichKrohmer15,kiwi2015bound,
heydenreich2016structures,deprez2015percolation,BollobasSR07,
bonato2010geometric,bradonjic2008structure}) deals with the design and analysis of random graphs that have more and more of the experimentally observed properties of social, technological, and other real-world networks, e.g., sparsity, scale-freeness, small diameter, and high clustering. 
Most of these modern models come equipped with an underlying geometry, since locality naturally yields high clustering and, more generally, local community structure. This makes them candidates for analyzing geometric routing or other variants of greedy routing.

%Interestingly, some of these models have an \emph{ultra-small} average distance $O(\log \log n)$ between all pairs of vertices, which poses the question: \emph{Can greedy routing be performed in an ultra-small number of steps?} All previous analyses of the algorithmic small-world phenomenon take at least $\Omega(\log n)$ steps, which might be due to the fact that (almost) all previous works were on homogeneous graphs.

Since there are no comparative studies, it is hard to judge which of the proposed models is most realistic. However, GIRGs and their special case, hyperbolic random graphs, stand out, because this model has not only attracted theoretical research but has also been experimentally validated. On the theory side, it was proven for GIRGs and hyperbolic random graphs that they are sparse and scale-free~\cite{gugelmann2012random,fountoulakis1,bringmann2015generalGIRG}, have constant clustering coefficient~\cite{DBLP:conf/waw/CandelleroF14,gugelmann2012random,bringmann2015euclideanGIRG}, a giant component and polylogarithmic diameter~\cite{FriedrichKrohmer15,kiwi2015bound,bringmann2015generalGIRG}, and average distance $\frac{2\pm o(1)}{|\log(\beta-2)|} \log \log n$, where $\beta$ is the power law exponent~\cite{bringmann2015generalGIRG}. Furthermore, these graphs have entropy $O(n)$~\cite{bringmann2015euclideanGIRG}, which coincides with the low entropy of the web graph~\cite{BoldiV03}. 
%Some more properties were shown with non-rigorous analytical methods for Hyperbolic Random Graphs~\cite{BogunaPK10,KrioukovPKVB,KrioukovPBV09,PapadopoulosKBV10}.
Recently, researchers have started to design algorithms that run fast on hyperbolic random graphs~\cite{friedrich2015cliques,blasius2016hyperbolic}. In terms of practical validation, Bogu\~n\'a et al.~\cite{BogunaPK10} computed a (heuristic) maximum likelihood fit of the internet graph into the hyperbolic random graph model and demonstrated its quality by showing that greedy routing in the underlying geometry of the fit finds near-optimal shortest paths. This successful embedding of a real-world network validates the model experimentally, and was repeated and extended in~\cite{papadopoulos2015network,blasius2016efficient}.

\end{enumerate}

\paragraph{Further Related Work.}
For psychological work related to Milgram's experiments we refer to the survey by Schnettler~\cite{schnettler2009structured}. 
The enormous amount of experimental and non-rigorous work on routing in real-world networks and random graphs is surveyed by Huang et al.~\cite{huang2014navigation}.
%Due to the sheer amount of experimental and non-rigorous work on routing in real-world networks and random graphs it is impossible to give a reasonable overview; we refer to the survey by Huang et al.~\cite{huang2014navigation}. 
%It is impossible to give a reasonable overview of experimental and non-rigorous work on routing in real-world networks and random graphs, as it is done for example in the physics community, so we can only humbly refer to the survey by  
Regarding theoretical work, a monograph by Watts summarizes the early work~\cite{WattsMonograph}. Let us highlight a few more recent theory papers on greedy routing and related decentralized routing schemes: Watts et al.~\cite{watts2002identity} and Kleinberg~\cite{kleinberg2002small} studied hierarchical network models that are well motivated by sociological principles. Alternatively, Fraigniaud et al.~\cite{FraigniaudDist14} augmented Kleinberg's model with long-range edges according to a power law. These approaches make the model somewhat more realistic, however, they still suffer from similar shortcomings as Kleinberg's original model, e.g., they either assume a perfect $b$-ary tree structure or a perfect lattice. 
On the other hand, after a long line of research~\cite{LibenNowellPNAS05,fraigniaud2005greedy,LibenNowellESA06,duchon2006could,FraigniaudESA06,abraham2006object,FraigniaudESA2007} Fraigniaud et al.~\cite{FraigniaudSTOC10} replaced the lattice structure in Kleinberg's model with arbitrary connected base networks. They showed how these base networks can be augmented with long-range edges to make them navigable, which in the most general case means that variants of greedy routing succeed with a subpolynomial number of steps, provided that all nodes know the base network. Note that this available information is much less local than in our model, where each node only needs information about its direct neighbors. In particular, since every node knows in advance a path to the target, greedy routing succeeds with probability one. 
%Liben-Nowell et al.~\cite{LibenNowellPNAS05, LibenNowellESA06} studied similar long-range edge augmentation schemes, experimentally showed that they actually occur in social networks and proved that geometric routing takes polylogarithmic number of steps. 
Arguably, all these networks were designed specifically for studying greedy routing, while our approach is rather to study greedy routing on a preexisting model for real-world networks. Also, none of the previous theoretical work showed ultra-small routing distance or considered patching strategies.

\paragraph{Organization of the paper.} \label{par:structureofpaper} In Section~\ref{sec:model} we formally introduce the model, in Section~\ref{sec:results} we list and explain our main results, and in Section~\ref{sec:experiments} we discuss related experimental results. In Section~\ref{sec:patching} we discuss the patching criteria (P1)-(P3) and related literature on patching algorithms, and give two simple examples for patching algorithms which satisfy (P1)-(P3). Section~\ref{sec:proofsketch} contains a high level discussion of the typical trajectory of greedy routing, and it also contains a sketch of the main ideas of our proofs. Section~\ref{sec:preliminaries} contains preparations for the proofs. Thereby, in Section~\ref{subsec:notation} we introduce the notation needed for the proofs. This section also contains a reference table for the notation used in the various proofs. In Section~\ref{subsec:mainlemma} we prove the main Lemma~\ref{lem:main}, which lies at the heart of all subsequent proofs, and finally we give the proofs of our main results in Sections~\ref{sec:analysis} (basic algorithm), \ref{sec:proofpatching} (patching), \ref{sec:relaxations} (relaxations), and~\ref{sec:hyperbolic} (hyperbolic random graphs). %In Section~\ref{sec:patching} we also discuss related literature on patching algorithms and give examples of algorithms which satisfy the patching properties (P1)-(P3).

%For any vertex $v$, we denote by $\delta_v$ its distance to the target $t$. Then the probability that $v$ is connected to the target $t$ is by definition $\Theta(\min\{1,(\delta_v^{-d}\w{v}\w{t}W^{-1})^{\alpha}\})$, which is increasing in $\w{v}/\delta_v^d$. Therefore, for the greedy forwarding the choice $\phi(v) := \frac{\w{v}}{\delta_v^d}$ is natural. 

\section{Random Graph Model \& Greedy Routing} \label{sec:model}

\subsection{Geometric Inhomogeneous Random Graphs (GIRGs)} \label{subsec:girg}

In this paper we study greedy routing on Geometric Inhomogenous Random Graphs (GIRGs). 
%Other than in the 10-pages abstract we do not hide the minimal possible weight $\wmin$ in $O$-notation, since Theorem~\ref{thm:greedysuccess1} links the failure probability of the routing protocol to $\wmin$. The dependence on $\wmin$ is chosen such that a vertex of weight $w$ has expected degree $\Theta(w)$, where the hidden constant does not depend on $\wmin$.
Informally speaking, a GIRG is a graph $G= (V,E)$ where both the set of vertices $V$ and the set of edges $E$ are random. Informally speaking, each vertex $v$ has a random position $\x v$ in a geometric space $\Space$ and a random weight $\w v \in \R^+$, where the distribution of the weights follows a power law. For each pair of vertices we flip an independent coin to determine whether they are connected by an edge, where the probability to be connected increases with the weights and decreases with the distance of the vertices. We now give the formal definition of the model, following~\cite{koch2016bootstrap}.

\paragraph{Geometric space:}\label{sec:model-geoSpace}
Let $d \in \N$ be some (constant) parameter. As our \emph{geometric space} we consider the $d$-dimensional torus $\Space=\R^d / \Z^d$, which can be described as the $d$-dimensional cube $[0,1]^d$ where opposite boundaries are identified with each other. As distance function we use the $\infty$-norm on $\Space$, i.e., for $x,y \in \Space$ we define 
$$\|x-y\| := \max_{1 \le i \le d} \{\min\{|x_i-y_i|,1-|x_i-y_i|\}\}.$$
The choice of the geometric space $\Space$ is in the spirit of the classical random geometric graphs~\cite{penrose2003}. We prefer the torus to the hyper-cube for technical simplicity, as it yields symmetry. However, it is not hard to replace $\Space$ by $[0,1]^d$. Similarly, we may replace the maximum norm by any other norm and obtain the same model since we allow constant factor deviations in the definitions~\eqref{eq:puv} and~\eqref{eq:puv2} below.\footnote{Note that when embedding a real world network into the GIRG model, e.g., as in \cite{BogunaPK10}, the real world coordinates do not necessarily resemble the coordinates in the geometric space. It rather resembles an ambient space which additionally includes more parameters such as hobbies.}

Let the set of vertices V be given by a Poisson point process on $\Space$ with intensity $n$, cf.~\cite{streit2010poisson} for a formal definition of Poisson point processes.\footnote{A very similar model is obtained by placing $n$ points uniformly at random in $\Space$, e.g.~\cite{bringmann2015euclideanGIRG}. We prefer the Poisson point process since then the number of vertices in disjoint regions of $\Space$ is independent.} In particular, the expected number of vertices is $n$, the positions of all vertices are uniformly at random in $\Space$, and for each measurable $A \subset \Space$, the number of vertices in $A$ is Poisson distributed with mean $n\cdot \Vol(A)$. Moreover, if $A, B \subseteq \Space$ are disjoint measurable sets, then the number of vertices in $A$ and in $B$ are independent random variables.

\paragraph{Weights:}\label{sec:model-weights}
Let $2 < \beta < 3$ and $\wmin > 0$ be fixed constants. Then for each vertex $v\in V$ independently we draw a weight according to the probability density function
%the set of vertices, their positions and their weights is given by a homogeneous Poisson point process on the space
%$\probSpace := \Space \times \R$ with intensity $n \in \N$. More precisely, let $\lambda$ be the standard Lebesgue measure on $\Space$, let $\nu$ be the measure on $\R$ given by the density

\[f(w) := \begin{cases} \;\Theta\big(\wmin^{\beta-1}w^{-\beta}\big), &w \ge \wmin\\ \;0, &w < \wmin .\end{cases}
\]
%and let $\mu := \lambda \times \nu$ be the product measure on $\probSpace$. Note that $\mu(\probSpace)=1$. For every $\mu$-measurable set $A \subseteq \probSpace=\Space \times \R$, we denote by $V \cap A$ the set of vertices $v$ with $(\x v, \w v) \in A$. Then we require that (a) $|V \cap A|$ is Poisson distributed with mean $n \mu(A)$, that is, for every integer $k \ge 0$,
%$$\Pr[|V \cap A|=k]=\Pr[\Pois(n \mu(A))=k] = \frac{(n \mu(A))^k \exp(-n \mu(A))}{k!},$$
%and (b) for two disjoint sets $A_1, A_2 \subset \probSpace$, the random variables $|V \cap A_1|$ and $|V \cap A_2|$ are independent.  Consequently, the number of vertices $|V|$ is itself distributed as $Pois(n)$. See \cite{streit2010poisson} for details on the Poisson point process.

%This process yields the random vertex set $V$. For a vertex $v \in V$ its position $\x v$ is uniformly at random on $\Space$, and its weight $\w v$ is at random on $\R$ according to the defined density $f$. 
The density $f$ is chosen such that  every vertex takes a random weight according to a power-law with parameter $\beta$ and minimal weight $\wmin$, and the probability that a random vertex has weight at least $w$ is proportional to $w^{1-\beta}$. The choice of the power law parameter ($2 <\beta < 3$ ) is standard~\cite{dorogovtsev2002evolution}.

\paragraph{Edges:}\label{sec:model-edges}

Given the vertex set $V$ and all positions and weights, we connect vertices $u \ne v$ independently with probability $p_{uv}$, which depends on the weights $\w u,\w v$ and on the distance $\|\x u - \x v\|$. We require for some constant $\alpha > 1$ the following edge probability condition:
\begin{equation}\label{eq:puv} p_{uv} = \Theta\bigg( \min\bigg\{ \frac{1}{\|\x u - \x v\|^{\alpha d}} \cdot \Big( \frac{\w u \w v}{\wmin n}\Big)^{\alpha} ,1 \bigg\} \bigg). \tag{EP1}
\end{equation}
In particular, there exists a constant $c_1>0$ such that $p_{uv}=\Theta(1)$ if $u$ and $v$ are very close to each other, i.e., $\|\x u - \x v\|^d \le c_1\tfrac{\w u \w v}{\wmin n}$.

The parameter $\alpha$ determines how fast the connection probability decays with increasing distance of the nodes. In the \emph{threshold case} $\alpha = \infty$ the connection probably instantly drops to zero if the vertices are too far apart: instead of~\eqref{eq:puv} we require that there are constants $0 < c_1 \le c_2$ such that
\begin{equation}\label{eq:puv2}
 p_{uv} = \begin{cases} \Theta(1) & \text{if } \|\x u - \x v\|^d \le c_1\tfrac{\w u \w v}{\wmin n} \\ 0 & \text{if } \|\x u - \x v\|^d \ge c_2\tfrac{\w u \w v}{\wmin n}. \end{cases}  \tag{EP2}
 \end{equation}
Note that there can be an interval $[ (c_1\tfrac{\w u \w v}{\wmin n})^{1/d},  (c_2\tfrac{\w u \w v}{\wmin n})^{1/d}]$ for $\|\x u - \x v\|$ where the behaviour of $p_{uv}$ is arbitrary. In many settings the connection probability is $1$ if $u$ and $v$ are sufficiently close to each other. %$\Theta$ in (\ref{eq:puv}) and (\ref{eq:puv2}) hide a constant $1$. 

The formulas~\eqref{eq:puv} and~\eqref{eq:puv2} depend on the parameter $\wmin$ in such a way that the expected degree of a vertex of weight $w$ is $\Theta(w)$. Moreover, for two vertices $u$ and $v$ of fixed weights, but with random positions, the connection probability is $\Ex_{\x u, \x v}[p_{uv} \mid \w u, \w v] = \Theta(\min\{\tfrac{\w u \w v}{\wmin n},1\})$. For details see Section~\ref{sec:preliminaries}.\bigskip

Summarizing, the free parameters of the model are the intensity (i.e., the expected number of vertices) $n$, minimum weight $\wmin$, the power law parameter $2<\beta<3$, the dimension $d\in\N$, the decay parameter $\alpha>1$, and the actual probability functions $p_{uv}$ which satisfy $(EP1)$ (if $\alpha < \infty$) or $(EP2)$ (if $\alpha = \infty$). In all subsequent theorems and proofs, we assume that all parameters except $n$ are constants. In particular, the $O$-notation may hide any dependence on the parameters $d, \alpha, \beta$ and on the hidden constants in \eqref{eq:puv} and \eqref{eq:puv2}, but not on $\wmin$ as Theorem~\ref{thm:greedysuccess2} will link the failure probability of the routing protocol to $\wmin$. 

\subsection{Greedy routing}\label{subsec:routing}

\paragraph{Routing protocol:}\label{subsec:routing-routing-protocol}
We will evaluate the performance of greedy routing in a GIRG, which is the following process. A message should be sent from a starting vertex $s$ (source) to a target vertex $t$. The \emph{address} of every vertex $v$ is  the pair $(\x{v},\w{v})$. Every vertex has local information, i.e., it knows the address of itself and of its neighbors. In addition, the address of the target is written on the packet. Then the routing proceeds in rounds and in every hop, the packet is sent from the current vertex $v$ to a neighbor $u$ which maximizes a given objective function $\phi$, cf.~below. When a dead end is reached, the process stops and the packet is dropped. If $u=t$, then we say that the routing was \emph{successful}. A pseudocode description is given with Algorithm~\ref{alg:greedy} below. Note that we will analyze the process for arbitrary vertices $s$ and $t$, i.e., we may explicitly choose weights and positions of $s$ and $t$, while everything else of the graph is still drawn randomly.

\paragraph{Objective function:}It is natural and in spirit of Milgram's experiment~\cite{travers1969experimental} that in every round, the packet should be sent to a neighbor $v$ which maximizes $p_{vt}$. Note that if $\alpha<\infty$, maximizing $p_{vt}$ is equivalent to maximizing $\w v/(\wmin n \|\x v - \x t\|^d)$.\footnote{We keep the normalization factor $n$ since it will turn out to be the most natural normalization. In this way, $\phi(v)$ will always increase by the same exponent in the second phase of the algorithm, cf.\! Section~\ref{sec:proofsketch}. We use the additional $\wmin$-factor only for technical reasons, as it simplifies several calculations.} Hence, for our analysis we use the objective function
$$\phi(v) := \frac{\w v}{\wmin n \|\x v - \x t\|^d }$$
and observe that the target vertex $t$ globally maximizes $\phi$, which is a necessary condition for any reasonable objective function. In particular, if $\{s,t\} \in E$, then the algorithm will send the packet directly to the target.

We emphasize that each vertex only needs to know the positions and weights of its direct neighbors, and the geometric position of $t$ (which we assume to be part of the message). This goes in line with Milgram's experiment, where participants knew geometric positions and professions of their acquaintances and made their choices accordingly \cite{milgram1967small,travers1969experimental}.

\begin{algorithm}
\begin{algorithmic}[1]
\Function{GREEDY}{$s$, $m$} \Comment{$s$: source, $m$: message (contains target information $t$)}
  \If{$s == t$} deliver message
  \Else
    \State{$v \gets \argmax\{\phi(u) \mid  u \in \Gamma(s)\}$}
    \If{$\phi(v) > \phi(s)$} GREEDY($v$, $m$)
    \Else ~return \emph{failure}
    \EndIf 
  \EndIf
\EndFunction
\end{algorithmic}
\caption{Greedy Routing Algorithm}
\label{alg:greedy}

\end{algorithm}

\section{Results} \label{sec:results}
In this Section we list and explain all our main results. Unless stated otherwise, we assume throughout the paper that $s$ and $t$ are fixed, while the rest of the graph is drawn randomly. That is, an adversary may pick weights and positions of $s$ and $t$, while the remaining vertices and all edges are drawn randomly as described in Section~\ref{sec:model}.\footnote{Equivalently, we can first choose the random graph $G$, then draw $s$ and $t$ at random, and condition the resulting probability space on $\x s$, $\x t$, $\w s$, and $\w t$.} Then we consider greedy routing from $s$ to $t$. Proof sketches and intuitive reasons for the results can be found in Section~\ref{sec:proofsketch}. The formal proofs can be found in Section~\ref{sec:analysis} (basic algorithm), Section~\ref{sec:proofpatching} (patching), Section~\ref{sec:relaxations} (relaxations), and Section~\ref{sec:hyperbolic} (hyperbolic random graphs).

\paragraph{Success Probability.} As our first result we show that greedy routing from $s$ to $t$ succeeds with constant probability. The proofs of the following theorem is in Section~\ref{subsec:constantsuccess}.

\begin{theorem} \label{thm:greedysuccess1}
Greedy routing succeeds with probability $\Omega(1)$.
\end{theorem}

In general, we cannot expect anything better than constant conection probability, since we only assumed a connection probability of $\Theta(1)$ in \eqref{eq:puv} and~\eqref{eq:puv2} even if two vertices are very close to each other. So we may end in a vertex which is extremely close to $t$, but which fails to actually connect to $t$. If we assume connection probability $1$ in such cases, then we get a better connection probabilty. More precisely, assume that the constant $c_1$ in \eqref{eq:puv} or \eqref{eq:puv2} is chosen such that the following condition holds.
\begin{equation}\label{eq:puvlarge}
 p_{uv} = 1 \qquad \text{if} \qquad \|\x{u}-\x{v}\|^d \le c_1\tfrac{\w{u}\w{v}}{\wmin n}.
\tag{EP3}
\end{equation}
Then we can strengthen Theorem~\ref{thm:greedysuccess1} by showing that the failure probability decays exponentially with $\wmin$.
%Let us shortly discuss this result. Suppose $\|\x s - \x t\|$ is so small that $s$ maximizes the objective $\phi$ among all vertices except $t$. Then the routing succeeds if and only if $\{s,t\} \in E$. By (EP1) and (EP2), the probability that this edge is present is $\Omega(1)>0$. In this sense, our success probability of $ \Omega(1)>0$ is optimal and in general, we can not push it closer to $1$. However, the model allows us to make stronger assumptions on the edge-probabilities as we are free to choose the constant factors hidden in (EP1) and (EP2). In particular, it is very natural to set the model such that $p_{uv}=1$ whenever the distance between two vertices $u,v$ is sufficiently small, i.e., $\|\x u - \x v\|^d \le \frac{c_1 \w u \w v}{\wmin n}$. Thereby we can improve the success probability of the greedy routing and obtain the follows result.

\begin{theorem}\label{thm:greedysuccess2}
\leavevmode
If \eqref{eq:puvlarge} holds in addition to the model assumptions, then:
\begin{enumerate}[(i)]
\item Greedy routing succeeds with probability $1-O(e^{-\wmin^{\Omega(1)}})$.
\item If $\min\{\w s,\w t\}=\omega(1)$ then greedy routing succeeds
with probability $1-\min\{\w s, \w t\}^{-\Omega(1)}$. In particular, in this case greedy routing succeeds a.a.s..
\end{enumerate}
\end{theorem}
We will prove this result in  Section~\ref{subsec:expsuccess}. The first part of the theorem is also almost-optimal. The degree of a vertex $v \in \probSpace$ is distributed as $\Pois(\Theta(\w v))$, as we will see in Lemma~\ref{lem:degvertex} below. Therefore a vertex $v$ of weight $\w v$ is isolated with probability $e^{-\Theta(\w v)}$. If for example $\w t = \Theta(\wmin)$, then the success probability cannot be higher than $1-e^{-\Theta(\wmin)}$. 
%Below in Lemma~\ref{lem:degvertex}, the degree of a vertex $v \in V$ is distributed as $\Pois(\w v)$, therefore a vertex $v$ of weight $\w v$ is isolated with probability $O(e^{-\w v})$. If for example $\w t = O(\wmin)$, then the success probability can not be higher than $1-O(e^{-\wmin})$. In this sense, the first part of the theorem is almost-optimal. 
%In fact, statement~(i) holds even if we slightly soften~\eqref{eq:puvlarge} and only require that for any two vertices $u,v$ satisfying $\|\x{u}-\x{v}\|^d \le c_1\frac{\w{u}\w{v}}{\wmin n}$ we have the lower bound $p_{uv} \ge 1-\exp(-\wmin^{\Omega(1)})$. 

For both statements we remark that we only need~\eqref{eq:puvlarge} for the very last step. Even without~\eqref{eq:puvlarge}, with sufficiently high probability greedy routing finds a vertex $u$ such that $u$ and $t$ satisfy the precondition of~\eqref{eq:puvlarge}. However, without~\eqref{eq:puvlarge} the vertex $u$ only has probability $\Theta(1)$ to connect to $t$ directly, and there is a considerable chance that there is no better neighbor of $u$, i.e., that $u$ is a local optimum.\medskip
 
%, if both $\min\{\w s, \w t\}=\omega(1)$ and we don't require any further assumptions on the model, then a.a.s.\ greedy routing finds a vertex $u$ which is very close to $t$ and satisfies $\phi(u) = \omega(\w t^{-1})$.\jl{Careful, this would give the a.a.s., but not the more precise formula}. Thus we need the assumption $p_{ut}=1$ only for the very last step.

\paragraph{Stretch.} One of the most crucial efficiency measures for routing protocols is the stretch, i.e., the ratio of the routing paths length compared to the length of a shortest path. Clearly, routing can only succeed if the source $s$ and the target $t$ are in the same component of the graph. In \cite{bringmann2015generalGIRG} it was proven that a.a.s. a GIRG possesses a giant connected component of linear size in which the average distance is $\tfrac{2\pm o(1)}{|\log(\beta-2)|} \log \log n$. The next theorem shows that the number of steps of the greedy routing algorithm (either the target is found or the routing stops) is a.a.s.\! equal to shortest paths, up to a stretch of $1+o(1)$. The proof is given in Section~\ref{subsec:lengthofrouting}.

\begin{theorem}\label{thm:length} A.a.s., greedy routing stops after at most $\tfrac{2+o(1)}{|\log(\beta-2)|} \log \log n$ steps and either finds $t$ or ends in a dead end. More precisely, for every $f_0(n) = \omega(1)$ a.a.s.\! the stronger bound
\begin{equation}\label{eq:length2}
\tfrac{1+o(1)}{|\log(\beta-2)|} (\log \log_{\w s} \phi(s)^{-1} +  \log \log_{\w t} \phi(s)^{-1}) + O(f_0(n)).
\end{equation}
holds.
\end{theorem}

Note that Theorems~\ref{thm:greedysuccess1} and~\ref{thm:length} together imply that even if we condition on greedy routing being successful, then still a.a.s.\! greedy routing needs at most $\tfrac{2+o(1)}{|\log(\beta-2)|} \log \log n$ steps. Since this agrees with the average distance in the giant, it is optimal. In particular, for two vertices $s, t$ with random weights and positions, if greedy routing between $s$ and $t$ is successful then a.a.s.\! the stretch is $1+o(1)$.\footnote{Technically speaking, here we use the slightly stronger statement that a.a.s.\! a $1-o(1)$ fraction of all pairs of vertices in the giant have distance $\tfrac{2\pm o(1)}{|\log(\beta-2)|} \log \log n$. This statement is not explicitly formulated as a lemma in~\cite{bringmann2015generalGIRG}, but is in the proof of Theorem~5.9 in~\cite{bringmann2015generalGIRG}.}

%We remark that we actually show a stronger statement. If $f(n)=\omega(1)$ is any growing function then a.a.s.\! the length of the greedy path is at most
%\begin{equation}\label{eq:length2}
%\tfrac{1+o(1)}{|\log(\beta-2)|} (\log \log_{\w s} \phi(s)^{-1} +  \log \log_{\w t} \phi(s)^{-1}) + O(f(n)).
%\end{equation}
In the most typical situation (e.g., if $s$ and $t$ are chosen randomly),~\eqref{eq:length2} evaluates to the simpler expression $\tfrac{2+o(1)}{|\log(\beta-2)|} \log \log n$ in Theorem~\ref{thm:length}. For example, this is the case if $s$ and $t$ have constant weights and are far apart from each other: $\w s, \w t = O(1)$, and $\|\x s-\x t\| = \Omega(1)$. However, if $s$ and $t$ have large weights or are close to each other, then~\eqref{eq:length2} gives a much better bound (but never a worse bound, since $\phi(s) = \Omega(1/n)$ is always true).

%With our next result we analyze the running time of the greedy routing and the length of the greedy path. Note that when the routing stops, it either found $t$ or it ended in a dead end.

%\begin{theorem}\label{thm:length}
%Let $f(n)=\omega(1)$ be a function which grows arbitrarily slow in $n$. Then a.a.s., the length of the greedy path is at most
%$$\tfrac{1+o(1)}{|\log(\beta-2)|} (\log \log_{\w s} \phi(s)^{-1} +  \log \log_{\w t} \phi(s)^{-1}) + O(f(n)) \le \tfrac{2+o(1)}{|\log(\beta-2)|}\log\log n.$$
%\end{theorem}

\paragraph{Patching.} As outlined in the introduction, greedy routing may have practical applications. For these applications, it is not acceptable that the algorithm simply fails if it enters a local maximum. Therefore, several patching methods have been proposed. We show that \emph{any} patching algorithm $A$ is efficient if it satisfies the following three basic conditions. For further discussions, in particular of the conditions, see Section~\ref{sec:patching}.
\begin{enumerate}[(P1)]
\item \emph{(Greedy choices.)} If $A$ is in some vertex $v$ and decides to visit an unexplored neighbor then it must choose the unexplored neighbor of largest objective. If $A$ visits a vertex $v$ for the first time, and $v$ has a neighbor of larger objective, then $A$ proceeds to the neighbor of $v$ with largest objective.\label{cond:patching}
\item \emph{(Poly-time exploration.)} If $A$ has explored $k$ vertices, and at least one of these vertices has an unexplored neighbor, then $A$ visits an unexplored vertex in $k^{O(1)}$ steps.
\item \emph{(Poly-time exhaustive search.)} Assume that $A$ visits a vertex $v$ that has larger objective than all previously visited vertices, and let $S$ be the connected component of $v$ in $G[V_{\geq \phi(v)}]$, i.e., $S$ is the set of all vertices that can be reached from $v$ without touching vertices of worse objective than $v$. Then $A$ visits the vertices in $S$ in the next $|S|^{O(1)}$ steps (or finds $t$ in one of these steps).\\
% --- probably we don't need this condition just for the loglogn-bound
%Then for all $1\leq k \leq |S|$, $A$ visits at least $k$ vertices in $S$ in the next $k^{O(1)}$ steps (or finds $t$ in one of these steps).\\
If $\phi(s) = \Omega(1)$ then there exists a sufficiently slowly falling $\phi_0 = o(1)$ such that $A$ explores the connected component $S$ of $s$ in $G[V_{\geq \phi_0}]$ in $|S|^{O(1)}$ steps.
\end{enumerate} 

Condition (P2) ensures that greedy routing is always successful if source $s$ and target $t$ are in the same component. The surprising result is that (P1)-(P3) are already sufficient to ensure \emph{efficient} greedy routing, i.e., a.a.s.\! the stretch is $1+o(1)$. The proof can be found in Section~\ref{sec:proofpatching}.

\begin{theorem}\label{thm:patching}
Assume that $A$ is a routing algorithm that satisfies (P1)-(P3). If $s$ and $t$ are in the same component\footnote{i.e., the a.a.s.~statement holds conditioned on the event that $s$ and $t$ are both in the same component.} then $A$ always successfully routes from $s$ to $t$, and a.a.s.\! this takes at most 
$\tfrac{2+o(1)}{|\log(\beta-2)|} \log \log n$
steps.\footnote{In fact, a.a.s.\ the precise number of steps is again given by expression \eqref{eq:length2}. For the sake of readability, we omit the proof of the more precise bound and restrict ourselves to the bound as claimed in the theorem.}
\end{theorem}

%Clearly, the routing can succeed only if $s$ and $t$ are in the same component of the graph. In \cite{bringmann2015generalGIRG} it was proven that if $\beta<3$, then with high probability a GIRG possesses a giant component of linear size in which the average distance is $(2+o(1))\frac{\log \log n}{|\log(\beta-2)|}$ in expectation and with probability $1-o(1)$. Hence, the obtained number of steps is asymptotically the same as the average distance in the giant and thus optimal. For $\beta\ge 3$, GIRGs can fall in small components, which is the reason why we restrict ourselves to the case $2<\beta<3$.

%For our first three theorems, we required that the objective $\phi$ needs to be computed exactly, which could be expensive in practice. If the protocol uses only approximations instead of exact values, then the routing will still be greedy routing as defined above, but w.r.t.\ a slightly different objective function $\tilde{\phi}$. With the following theorem, we generalize our results from the original $\phi$ to a general class of objective functions $\tilde{\phi}$, which are allowed to deviate from $\phi$ by at least a constant factor.

\paragraph{Relaxations.}\label{subsec:routing-objective-function} So far we assumed that the objective $\phi$ can be computed exactly, which may be unrealistic in practice. For example, in the Milgram experiment the participants can only estimate the objective values of their neighbors. So we also study an approximate version of greedy routing, where the greedy algorithm uses an objective function $\tilde{\phi}$ which is only an approximation of $\phi$. We allow that $\tilde{\phi}$ deviates from $\phi$ by a bounded constant factor, or even more. This adds quite some \emph{flexibility}: for any of the best $\min\{\w v, \phi(v)^{-1}\}^{o(1)}$ neighbors of a vertex $v$, there is a ``good enough'' approximation $\tilde{\phi}$ which puts this vertex to the top. In particular, it is not necessary to compute the optimal neighbor w.r.t.~$\phi$, which may be costly. Moreover, the routing process is also \emph{robust}: for example, it is no problem if some of the edges fail during execution of the routing, since the current vertex can send the message to any other good neighbor instead.

%\paragraph{Relaxations.} So far we assumed that the objective $\phi$ can be computed exactly, which may be unrealistic in practice. For example, in the Milgram experiment the participants can only estimate the objective values of their neighbors. So we also study approximate versions of greedy routing. We consider two different variants: first we study the greedy algorithm w.r.t.\ to a objective function $\tilde{\phi}$ which is only an approximation of $\phi$; second, we allow the algorithm to use the exact objective function $\phi$, but it may forward the package to any vertex which is \emph{among} the best neighbors, not necessarily to the optimal one. Note that the latter case is useful for practical applications, since finding some good neighbor may be much cheaper than finding the optimal neighbor. Moreover, the latter result also shows that the routing algorithm is \emph{flexible} and \emph{robust}: for example, it is no problem if some of the edges fail during execution of the routing, since the current vertex can send it to any other good neighbor instead.

\begin{theorem}\label{thm:relaxations}
Let $\tilde{\phi}: \probSpace \rightarrow \R$ be a function which is maximized with $t$ such that 
\begin{align}\label{eq:relaxation}
\tilde{\phi}(v) = \Theta\left(\phi(v) \cdot \min\{\w v, \phi(v)^{-1}\}^{o(1)}\right),
\end{align}
holds for all $v \in \probSpace$. Then Theorems~\ref{thm:greedysuccess1},~\ref{thm:greedysuccess2},~\ref{thm:length}, and~\ref{thm:patching} also hold for greedy routing with respect to $\tilde{\phi}$ instead of $\phi$. 

Moreover, let $\delta >0$ be a small constant. Then for Theorems~\ref{thm:greedysuccess1},~\ref{thm:greedysuccess2}, and~\ref{thm:length} the previous statement is still true if for vertices $v \in \probSpace$ with $\phi(v) \ge O(\w t^{-1+\delta})$ we relax Condition~\eqref{eq:relaxation} to the weaker condition
\begin{equation}\label{eq:weakrelaxation}
\tilde{\phi}(v) = \Omega\left(\w t^{-1+\delta/2}\right).
\end{equation}

Finally, for transferring Theorem~\ref{thm:patching} this weaker condition is also sufficient if $t$ is a random vertex.
\end{theorem}

\noindent{The proof of Theorem~\ref{thm:relaxations} can be found in Section~\ref{sec:relaxations}.}

A consequence of Theorem~\ref{thm:relaxations} is that the results of this paper also hold for variations of the model. For example, most of the experimental work is performed in slightly different models, e.g., they consider geometric routing on hyperbolic random graphs (see Section~\ref{sec:experiments} for details). In \cite{bringmann2015euclideanGIRG} it was proven that one-dimensional GIRGs ($d=1$) contain the model of hyperbolic random graphs as a special case, but still geometric routing in hyperbolic space is not exactly the same as greedy routing in our sense. Theorem~\ref{thm:relaxations} shows that the differences are negligible. Formally, geometric routing on hyperbolic random graphs induces an objective function $\phi_H$ on the corresponding GIRG, which falls into the class of objective functions considered in Theorem~\ref{thm:relaxations}. The proof and further background on hyperbolic random graphs is contained in Section~\ref{sec:hyperbolic}.

%So far, we assumed that the packet is always sent to \emph{the} best neighbor. From Theorem~\ref{thm:relaxations} it follows that it is sufficient to proceed to a vertex which is \emph{among} the best neighbors. If the current vertex has high degree, this has the advantage that not every neighbor needs to be processed and the packet can be forwarded as soon as a neighbor with sufficiently high objective is found.
%
%\begin{corollary}\label{cor:oneofbest}Theorem~\ref{thm:greedysuccess2} and \ref{thm:length} hold as well for the routing w.r.t.\ $\phi$ when the routing proceeds always to a vertex which is among the $\min\{\w v, \phi(v)^{-1}\}^{o(1)}$ best neighbors.
%\end{corollary}

%In \cite{bringmann2015euclideanGIRG} it was proven that the one-dimensional GIRG contains hyperbolic random graphs as an example. Therefore, geometric routing on hyperbolic random graphs induces an objective function $\phi_H$ on a GIRG with $d=1$. It turns out that  $\phi_H$ falls into the class of objective functions considered in Theorem~\ref{thm:relaxations}, hence our results transfer to hyperbolic random graphs.

\begin{corollary}\label{cor:hyperbolic}All theorems about success probability and path length apply as well for geometric routing on hyperbolic random graphs. Moreover, if $t$ is a random vertex, then Theorem~\ref{thm:patching} for patching algorithms also holds for hyperbolic random graphs.
\end{corollary}

%Regarding patching algorithms, we claim without proof that the statement is still true for any choice of $s$ and $t$ and that the assumption of $s$ and $t$ being random vertices is not necessary.

\section{Comparison with Experimental Results}\label{sec:experiments}

There are several recent papers which experimentally study 'greedy processes' in graphs with hidden metric spaces \cite{BogunaK09, BogunaKC09, BogunaPK10, CvetkovskiC09, KrioukovPBV09, PapadopoulosKBV10}. In these papers, either GIRGs are sampled for various parameter combinations, sometimes under different names; or given real networks (e.g., the internet graph) are embedded into the space $\Space \times \mathcal{W}$, where $\mathcal W$ is the space of weights (mostly in the special case $d=1$, where $\Space \times \mathcal{W}$ can be interpreted as hyperbolic disc~\cite[Theorem 6.3]{bringmann2015euclideanGIRG}, cf.\! Corollary~\ref{cor:hyperbolic}). Afterwards, in both settings, the authors experimentally study how greedy routing algorithms perform on these graphs. Thus our paper gives theoretical explanations for the effects that were observed in these papers, which we discuss in detail below.
%
%
%There are several recent papers which experimentally study 'greedy processes' in graphs with hidden metric spaces \cite{BogunaK09, BogunaKC09, BogunaPK10, CvetkovskiC09, KrioukovPBV09, PapadopoulosKBV10}. In these papers, either GIRGs are sampled for various parameter combinations, sometimes under different names; or given real networks (e.g., the internet graph) are embedded into the space $\Space \times \mathcal{W}$, where $\mathcal W$ is the space of weights (mostly in the special case $d=1$, where $\Space \times \mathcal{W}$ can be interpreted as hyperbolic disc, cf.\! Corollary~\ref{cor:hyperbolic}). Afterwards, in both settings, the authors experimentally study how greedy routing algorithms perform on these graphs. Thus our paper gives theoretical explanations for the effects that were observed in these papers, which we discuss in detail below.

In contrast, \cite{BogunaK09} and \cite{BogunaKC09} study geometric greedy processes, i.e., the routing is degree-agnostic and only uses the distance in an homogeneous and isotropic space. It turns out that this geometric type of routing still works in some settings, but is far less efficient and robust (e.g., it completely fails for some values of $\beta \in [2,3]$). This suggests that greedy routing as considered in this paper is superior to geometric routing. It remains an intriguing open task to understand theoretically how geometric routing behaves for realistic network models.

\paragraph{Success Probability:} As the main contributions the aforementioned papers observe a large success probability of greedy routing for most studied combinations of parameters. This is explained by Theorem~\ref{thm:greedysuccess1}, which gives a general lower bound $\Omega(1)$ on the success probability. Often, the experimentally observed success probability is surprisingly large (i.e., larger than  $97\%$ in \cite{BogunaPK10}). We believe that we can give the reason for this observation with Theorem~\ref{thm:greedysuccess2} (i), where we show that the failure probability drops exponentially with the parameter $\wmin$ controlling the minimum expected degree. Thus, even with very moderate values of $\wmin$ it is not surprising to obtain large success probabilities. 
%The results in our paper explain why these large success probabilities are achieved and formally prove many of the conjectures in the aforementioned papers. In particular, Theorem \ref{thm:greedysuccess1} shows that the success probability is constant and 
%	 Theorem \ref{thm:greedysuccess2} quantifies this by showing that, asymptotically, the probability of non-success is exponentially small in the minimum degree. 	

\paragraph{Trajectory of a Greedy Path:} In their experiments the authors in \cite{BogunaPK10,  KrioukovPBV09,KrioukovPKVB, PapadopoulosKBV10} observed that a successful greedy path starting at a small degree node will first visit nodes of larger degrees and continue along this path until it reaches the core of the network; afterwards, the greedy path leaves the core towards nodes with smaller degrees but with much closer geometric distance to the target until it reaches its destination. We prove this observation formally. In our proofs, we partition our network into layers, and show that the greedy path follows closely this layer structure. In the first half of the process, the layers are characterized by increasing weights, in the second by decreasing geometric distances to the target. We prove that a.a.s.\! the greedy algorithm visits each layer at most once, and that it visits a $(1-o(1))$-fraction of all layers. Thus we characterize accurately the trajectory that greedy routing takes.

\paragraph{Stretch:} Besides the success probability, the stretch of a greedy path, i.e., the quotient of a greedy paths length vs.\! the length of a shortest path, is one of the most meaningful measures to evaluate the usefulness of greedy routing in applications. Experiments revealed a surprisingly small stretch close to $1$, both for embedded real graphs (\cite{BogunaPK10}) and for sampled random graphs \cite{KrioukovPBV09,KrioukovPKVB, PapadopoulosKBV10}. In Theorem \ref{thm:length} we give the theoretical explanation and prove formally that successful greedy routing has a stretch of $1+o(1)$ in the case of success. 

\paragraph{Patching:} Patching has been experimentally found to be highly efficient, and to maintain an average stretch close to $1$~\cite{CvetkovskiC09,huang2014navigation, KrioukovPBV09, KrioukovPKVB, PapadopoulosKBV10}. However, most experiments in sampled models have been performed with a small number of vertices (e.g., $n=1000$). Also, not all methods guarantee successful routing if $s$ and $t$ are in the same component. Interestingly, the gravity-pressure method considered in~\cite{CvetkovskiC09} shows good results, although it does not fall into our class of patching algorithms, and it suffers from the problems described in Section~\ref{sec:patching}. So it remains open whether this algorithm is also theoretically efficient, or whether the good results were simulation artefacts (see also~\cite{sahhaf2013link}). %Experiments with simpler patching protocols (\cite{KrioukovPBV09, PapadopoulosKBV10, huang2014navigation}), which enlarge the success probability but not to $100\%$, are also highly efficient. 
 
As Theorem~\ref{thm:patching} shows, a broad class of natural patching algorithms are highly efficient, even if they boost the success probability to $100\%$ for vertices in the same component. Most notably, we prove in Theorem~\ref{thm:patching} that a stretch of $1+o(1)$ is not only achievable, but that it is guaranteed under the weak and natural assumptions (P1)-(P3). % also shows that the number of steps is a.a.s. in the same order as shortest paths. 

\section{Patching Algorithms: Discussion and Examples} \label{sec:patching}
The naive greedy protocol gets stuck if it enters a local optimum, which is not acceptable for technical applications. Several different patching strategies have been suggested to overcome this issue, and experimental results showed that many of them work quite well in practice~\cite{CvetkovskiC09, huang2014navigation,KrioukovPBV09,KrioukovPKVB, PapadopoulosKBV10,sahhaf2013link}. We formally prove that in fact \emph{any} patching strategy works well if it satisfies some basic conditions, namely (P1)~\emph{greedy choices}, (P2)~\emph{poly-time exploration}, and (P3)~\emph{poly-time exhaustive search} (cf.\! page~\pageref{cond:patching}). 

Note that all three conditions are natural: the first one is that the algorithm makes greedy choices whenever there is an obvious decision to be made. Note that still there are many choices left to the algorithm. For example, if it enters a local optimum, then it is free to decide if and how much it backtracks before continuing exploration. The second condition ensures that $A$ never gets stuck (or only for a polynomial time). Finally, the third condition (P3) forces the algorithm to search exhaustively a connected set of ``good'' candidates in reasonable time before turning to less promising candidates. The second part of (P3) is concerned with the exceptional case $\phi(s) = \Omega(1)$ (for random $s$ and $t$ this only happens with probability $O(1/n)$) and is necessary to avoid trivialities: if $\phi(s) = \Omega(1)$ then there may be no (or only few) vertices of better objective than $s$, and the precondition of the first part of (P3) may never be fufillled. We remark that the requirement that $\phi_0$ is ``sufficiently slowly falling'' can be specified to $\phi_0 = (\log \log n)^{-o(1)}$.

The last condition (P3) does exclude some algorithms like the gravity-pressure algorithm introduced in~\cite{PapadopoulosKBV10}. This algorithm always visits the best neighbor $v'$ of the current search point $v$, even if $\phi(v') < \phi(v)$, and thus does not satisfy (P3). E.g., assume that the algorithm visits two consecutive vertices $v_1,v_2$ of increasing objective, and assume further that $v_2$ has exactly one more neighbor $u$ of very bad objective ($\phi(u) \ll \phi(v_1)$), while $v_1$ has other neighbors of better objective. Then continuing the search from $u$ before exploring the other neighbors of $v_1$ may be a bad strategy, and this is what (P3) rules out. For example, the gravity-pressure algorithm will prefer to visit any unexplored vertex over returning to $v_1$. So if $v_1$ lies on the only path to $t$ (which happens with probability $\Omega(1)$), then the algorithm may potentially explore large parts of the giant before returning to $v_1$, and thus before finding $t$. This may explain that the algorithm is especially vulnerable in sparse networks~\cite{sahhaf2013link}.

Generally, there are two relevant types of algorithms if the global structure is unknown to the vertices: either information about the routing history is stored in the message (e.g., the protocol SMTP for emails~\cite{postel1982simple}); or for every message a small amount of information is stored in each vertex, yielding rather exploration than classical routing (e.g., flooding algorithms like~\cite{harras2005delay,ko2000location}, the gravity pressure algorithm~\cite{PapadopoulosKBV10}, or tree-based approaches~\cite{sahhaf2013link}). In both cases, it is very easy to design a routing algorithm that satisfies conditions (P1)-(P3), as we outline in the following. For simplicity, we assume that $\phi(s) = o(1)$.

For the first case, we may simply store the list of visited vertices in the message, and for each vertex $v$ we additionally store the objective of the best unexplored incident edge in the message (i.e., the objective of the best neighbor $u$ of $v$ for which the algorithm did not traverse the edge $uv$). Compared to SMTP this only increases the required memory by a single value per visited node. With this information, a trivial way to satisfy conditions (P1)-(P3) is to use the greedy algorithm if possible (i.e., if we are not in a local optimum), and otherwise explore the best unexplored edge that goes out from any visited vertex. 

For the second case, a variant of depth-first search satisfies condition (P1)-(P3), in which the message and each visited vertex only need to store a constant number of pointers and objective values. More precisely, for a vertex $v \in V$ and a value $\Phi \in \R^+$ consider a greedy depth first search on the subgraph $G[V_{\geq \Phi}]$ of vertices which have objective at least $\Phi$, starting in $v$. ``Greedy'' here means that if there are several unexplored edges going out from a vertex, then the DFS algorithm picks the edge that leads to the vertex of highest objective. We call this algorithm (greedy) $\Phi$-DFS for short. Note that since $G[V_{\geq \Phi}]$ does not need to be connected, the $\Phi$-DFS does not necessarily visit all vertices in $G[V_{\geq \Phi}]$.

The idea is now the following. Whenever we encounter a vertex $v$ which has strictly larger objective than all previously visited vertices, then we start a $\phi(v)$-DFS at $v$. We do this recursively, so when in this $\phi(v)$-DFS we encounter another vertex $v'$ which has larger objective than all previously visited vertices (including the vertices visited during the $\phi(v)$-DFS), then we pause the $\phi(v)$-DFS, and start a $\phi(v')$-DFS in $v'$. It may happen that the $\phi(v')$-DFS is completed without success, i.e., that we return to $v'$ after recursively exploring all its better neighbors without finding $t$. In this case, we simply discard the $\phi(v')$-DFS and resume the paused $\phi(v)$-DFS. Note that we treat all vertices visited during the $\phi(v')$-DFS as unvisited for the resumed $\phi(v)$-DFS. However, we do store the best objective value that we have ever seen (regardless in which DFS we have seen them) in the message. A pseudocode description is given in Algorithm~\ref{alg:patching}. The algorithm is distributed, i.e, there is no shared global memory, and no global stack of function calls is required to execute the algorithm. Moreover, at each time only one vertex is active, and each vertex only needs to know positions and weights of its direct neighbors.

Let us argue that this algorithm only needs to store a constant number of pointers and objective values in the message and the visited vertices. During a $\Phi$-DFS, we store at each visited vertex the identity of the parent vertex (to allow backtracking) together with the value $\Phi$, and we call this pair the \emph{$\Phi$-information}. In the message we keep track of the current value of $\Phi$, the identity of the last visited vertex, and the best seen objective. Note that this information already allows us to perform the $\Phi$-DFS in a greedy depth-first manner, and also to decide whether we want to start a new DFS. If we start a new DFS in a vertex $v$, we need to store in $v$ all information that is necessary to resume the previous DFS: the parent vertex, the previous value of $\Phi$, and a flag indicating that we started a new DFS in $v$. So for each value of $\Phi$, the $\Phi$-DFS requires a constant memory in each vertex. It only remains to show that no vertex ever needs to store the $\Phi$-information for two different values of $\Phi$ simultaneously. So assume we enter a vertex during a $\Phi$-DFS, and it already contains the $\Phi'$-information for some $\Phi' > \Phi$. Then we started the $\Phi'$-DFS after the $\Phi$-DFS, and the only way to resume the $\Phi$-exploration is that the $\Phi'$-DFS terminated without finding $t$. So in this case we can safely delete the $\Phi'$-information. On the other hand, if we run a $\Phi$-DFS then we claim that we can never encounter a vertex $v$ which carries a $\Phi'$-information for some $\Phi' < \Phi$. Indeed, this could only happen if $\phi(v)\geq \Phi$ (to visit $v$ during the $\Phi$-DFS), and if moreover we would have visited $v$ during the $\Phi'$-DFS, i.e., before the start of the $\Phi$-DFS. But if such a vertex would exist then we wouldn't have started the $\Phi$-DFS in the first place, since $\Phi$ wouldn't have been the best encountered objective at this point. So indeed we can not visit a vertex $v$ which carries some $\Phi'$-information for $\Phi' < \Phi$.  Summarizing, we never need to store the $\Phi$-information for more than one value of $\Phi$ in the same vertex. 

\clearpage
\begin{algorithm}
\thispagestyle{empty}
%\caption{Greedy Routing with Patching}
\begin{algorithmic}[1]
\Function{ROUTING}{$s$, $m$} \Comment{$s$: source, $m$: message (contains target information)}
\State{$m$.best\_seen\_objective $\gets$ $-\infty$}\Comment{initialisation start}
\State{$m$.Phi $\gets$ $-\infty$} 
\State{$m$.last\_visited\_vertex $\gets$ s} %\Comment{last visited vertex}
\State{$s$.Phi $\gets$ $\phi(s)$}
%\State{SET\_NEW\_PHI($s$, $m$)}
%\State{INIT\_VERTEX($s$, $m$)} \Comment{initialisation end}
\State{EXPLORE($s$, $m$)}
\EndFunction

\Function{EXPLORE}{$v$, $m$} \Comment{we visit $v$ by an unexplored edge in current $\Phi$-DFS}
  \If{$v$.Phi == $m$.Phi} \Comment{if $v$ already visited in current $\Phi$-DFS ...}
     % \State{$m$.last\_visited\_vertex $\gets$ $v$} \Comment{$m$ stores the last visited vertex}
    \State{BACKTRACK\_TO($m$.last\_visited\_vertex, $m$)} \Comment{... then we immediately return}
  \Else \Comment{if $v$ not visited in current $\Phi$-DFS }
	\If{$\phi(v) > m$.best\_seen\_objective} \Comment{potentially start new DFS with $\Phi = \phi(v)$}
      \State{SET\_NEW\_PHI($v$, $m$)} 
    \EndIf
	  \State{INIT\_VERTEX($v$, $m$)} 
    
    \If{$\{u \in\Gamma(v) \mid \phi(u) \geq m\text{.Phi}\} \neq \emptyset$} \Comment{if there is neighbor better than $\Phi$}
      \State{$u$ $\gets$ $\argmax\{\phi(u) \mid u \in\Gamma(v)\}$} \Comment{go to best neighbor}
      \State{EXPLORE($u$, $m$)}
      \Else{ BACKTRACK\_TO($m$.last\_visited\_vertex, $m$)} \Comment{no unexplored neighbors better than $\Phi$}
    \EndIf
  \EndIf
\EndFunction

\Function{BACKTRACK\_TO}{$v$, $m$} \Comment{we backtrack from $m$.last\_visited\_vertex to $v$}
    \State{$S := \{u\in\Gamma(v)\mid u\neq v\text{.parent},  m\text{.Phi} \leq \phi(u) < \phi(m\text{.last\_visited\_vertex})\}$}
    \If{$ S \neq \emptyset $} \Comment{if there is still unexplored child better than $\Phi$ ...}
      \State{$u$ $\gets$ $\argmax\{\phi(u) \mid u\in S\}$} \Comment{... go to best of them}
      %\State{$m$.last\_visited\_vertex $\gets$ $v$} \Comment{update last visited vertex}
      \State{EXPLORE($u$, $m$)}
      \Else \Comment{if no child left to explore}
        \If{$v$.started\_new\_dfs == TRUE} \Comment{if DFS for current $\Phi$ was unsuccessful ...}
          \State{RESET\_TO\_OLD\_PHI($v$, $m$)}  \Comment{... we resume $\Phi$-DFS for previous value of $\Phi$ ...}
          \State{$m$.last\_visited\_vertex $\gets$ $v$.parent}\label{algline:resume} \Comment{... in the state where we left it, coming from $v$.parent}
          \State{EXPLORE($v$, $m$)} 
       \Else \Comment{if current $\Phi$-DFS continues, but no child left to explore ...}
      %\State{$m$.last\_visited\_vertex $\gets$ $v$} \Comment{... we update last visited vertex ...}
      \State{BACKTRACK\_TO($v$.parent, $m$)}  \Comment{...  backtrack to parent}
      %\State{}\Comment{if $v$.parent == NIL then $v$ = $s$, and $t$ is not in the same component as $s$}
      \EndIf
    \EndIf
\EndFunction

\Function{SET\_NEW\_PHI}{$v$, $m$} \Comment{$v$ starts a new DFS with $\Phi = \phi(v)$.}
  \State{$m$.best\_seen\_objective $\gets \phi(v)$}
  \If{$\{u \sim v \mid \phi(u) \geq \phi(v)\} \neq \emptyset$} \Comment{we only start the new DFS if there are actually better neighbors}
    \State{$v$.started\_new\_dfs $\gets$ TRUE}
    \State{$v$.previous\_Phi $\gets$ $m$.Phi} \Comment{store this in case we need to resume previous $\Phi$-DFS}
    \State{$m$.Phi $\gets$ $\phi(v)$}
  \EndIf
\EndFunction

\Function{RESET\_TO\_OLD\_PHI}{$v$, $m$} \Comment{to resume $\Phi$-DFS for previous value of $\Phi$}
  \State{$v$.started\_new\_dfs $\gets$ FALSE}
  \State{$m$.Phi $\gets$ $v$.previous\_Phi}
  \State{$v$.Phi $\gets$ $v$.previous\_Phi}
\EndFunction

\Function{INIT\_VERTEX}{$v$, $m$} \Comment{we visit $v$ for the first time in the current $\Phi$-DFS}
  \State{$v$.Phi $\gets$ $m$.Phi} \Comment{$v$ has been visited in current $\Phi$-DFS}
  \State{$v$.parent $\gets$ $m$.last\_visited\_vertex} \Comment{ parent for backtracking}
%  \State{$m$.last\_visited\_vertex $\gets$ $v$} \Comment{$m$ updates last visited vertex}
\EndFunction

\end{algorithmic}
\caption{Example of a distributed exploration algorithm that satisfies (P1)-(P3) and that uses only a constant number of pointers and variables in the message and in each node. \hspace{5cm} 
Whenever EXPLORE or BACKTRACK\_TO are called, the parent process may terminate.
For readability, we did not spell out updates of $m$.last\_visited\_vertex except the special case in line \ref{algline:resume}. Furthermore we assumed that no vertex has two neighbors of equal objective.
}\label{alg:patching}
\end{algorithm}

\clearpage

%To enhance readability, for the pseudo-code we assumed that no vertex has two neighbors of the same objective value. Note that the algorithm is distributed, i.e, there is no shared global memory, and no global stack of function calls is required to execute the algorithm. Moreover, at each time only one vertex is active, and each vertex only needs to know positions and weights of its direct neighbors. 

Finally we argue that the algorithm satisfies (P1)-(P3). It is evident that the algorithm satisfies conditions (P1) and (P2). Moreover, it also satisfies~(P3): if $x \in V$ and $S \subseteq V$ is the set of all vertices of objective at least $\phi(v)$, then a $\phi(v)$-DFS takes at most $O(|E(S,S)|) = O(S^2)$ steps. For every vertex $u \in S$, we may interrupt the $\phi(v)$-DFS for another $\phi(u)$-DFS, but each of these also takes at most $O(S^{2})$ steps. So in total we will explore the set $S$ in at most $O(S^3)$ steps. 

\section{Proof Sketches and the Typical Trajectory of a Greedy Path}\label{sec:proofsketch}

In this section we describe the evolution of the basic greedy process and give a very rough outline of the proof ideas and the main difficulties. At the same time, we describe the trajectory of a typical greedy path, which is also depicted in Figure~\ref{fig:greedypath} below. Full proofs can be found in the subsequent sections.

\paragraph{Mean field analysis and typical trajectories:}For simplicity, assume that both the source $s$ and the target $t$ have constant weight and assume that $\|\x s - \x t\|=\Omega(1)$, which is the typical case. As $s$ has small weight, a.a.s. all its neighbors (if any) have relatively small weight as well and are located in a ball of small geometric distance around $\x s$, which we call ``region of influence'' for illustrative purposes. Since the region of influence of $s$ is small and far away from $t$, all points in this region have the same distance from $t$, up to factors $(1+o(1))$. In particular, all neighbors of $s$ have the same distance from $t$, up to factors $(1+o(1))$. On the other hand, the weights of the neighbors fluctuate by non-negligible constant factors. (So do the distances \emph{from $s$}, but they do not influence the objective values.) Hence, the fluctuations of the weights dominate the fluctuations of the distances from $t$,  and consequently, routing proceeds to a neighbor of $s$ with higher weight. More precisely, we can expect that $s$ has a neighbor of weight roughly $\w s^{1/(\beta-2)}$. We repeat the argument and observe that the routing exhibits a first phase in which the weight of the current vertex increases by an exponent $\approx \frac{1}{\beta-2}>1$ within every step. This phase stops when a vertex $v$ with the property $\phi(v) \gtrsim \w v^{-1/(\beta-2)}$ is reached. Since the weight increases by an exponent of $1/(\beta-2)$ with every hop, this first phase needs only $\approx \log\log n / |\log(\beta-2)|$ steps.

Once the routing reaches a vertex $v$ such that $\phi(v) \gtrsim \w v^{-1/(\beta-2)}$, the relation between the current weight and the current objective change, since now the region of influence of $v$ contains $t$. That does not mean that $v$ is adjacent to $t$ (in fact, this is very unlikely), but $v$ does have neighbors which are much closer to $t$. These neighbors typically have smaller weight than $v$, but the gain in distance is enough to make up for the smaller weight. In this way we can hope for a neighbor of objective $\approx \phi(v)^{\beta-2}$. (Note that $\phi(v)<1$, so indeed the objective increases.) Moreover, the best neighbor still satisfies the relation $\phi(v) \gtrsim \w v^{-1/(\beta-2)}$, so we can repeat the argument.  Thus in this second phase, the objective increases by the same exponent in every step. %While the geometric distance to $t$ decreases, also the weights start decreasing as a small ball around $t$ does not contain vertices with high weight. 
After $\approx \log\log n /|\log(\beta-2)|$ rounds, this second phase reaches a vertex $v$ of objective $\phi(v)=\Omega(1)$. Then, in the case $\alpha<\infty$ we have $p_{vt}\ge (\phi(v) \w t)^{\alpha}=\Omega(1)$, so with  probability $\Omega(1)$ vertex $v$ connects directly to the target vertex $t$. If $\alpha=\infty$, we need a small number of additional steps to find $t$. This already gives an idea why greedy routing succeeds with probability $\Omega(1)$ (Theorem~\ref{thm:greedysuccess1}), and why it only needs $\tfrac{2+o(1)}{|\log(\beta-2)|} \log \log n$ steps (Theorem~\ref{thm:length}). It also gives a typical trajectory of the greedy path.

\paragraph{Main difficulties:}
 In the formal proofs of our statements we show that with sufficiently high probability, the routing does not deviate too much from this typical trajectory. For a single vertex $v$ it is not difficult to calculate the weight $\w u$ and the objective $\phi(u)$ of its best neighbor $u$. However, there is a major difficulty: the greedy algorithm only proceeds to $u$ when $u$ is the \emph{best} neighbor of $v$. So if we want to compute the probability that the best neighbor of $u$ is again what we expect, then we need to compute a \emph{conditional} probability, namely conditioned on $v$ having no better neighbor. Unfortunately, this introduces dependencies which are impossible to handle directly.

The main technical contribution of this paper is in overcoming these dependencies. The basic idea is to uncover the vertices by increasing objective. However, note that the greedy path $P_{\phi}$ in the graph $G_{\leq \phi}$ induced by vertices of objective at most $\phi$ does not need to coincide with the greedy path $P$ in $G$, since the vertices in $P_{\phi}$ may have a neighbor of objective larger than $\phi$, thus shortcutting the rest of $P_{\phi}$. Unfortunately, the number of vertices which lie on a greedy path in \emph{some} $G_{\leq \phi}$ is so large that many of them have undesired properties. Thus we can not prove iteratively that all nodes on the temporary greedy paths have good properties, because they do not. On the other hand, uncovering the next node on $P$ is impossible without introducing the aforementioned dependencies.

\paragraph{Layer technique:} Instead of looking at individual vertices, we define carefully crafted layers $A_i$ for the first phase (defined via weights) and the second phase (defined via objectives) such that the typical trajectory contains at most one vertex per layer. Then we consider the greedy path $P_i$ induced by the first $i$ layers and prove that with sufficiently high probability either $P_i$ has no vertex in $A_i$ or the first vertex $v \in A_i \cap P_i$ has a neighbor $u$ outside the first $i$ layers with better objective than all neighbors \emph{inside} the first $i$ layers (which we already have uncovered). Note that this is an event where we do not suffer from dependencies: the vertex $v$ is defined only in terms of the first $i$ layers, so we can determine $v$ without uncovering any parts of the remaining graph. Consequently, the number of vertices of $v$ in larger layers is Poisson distributed, and it suffices to show that its expectation is large. It is not hard to see that the above events together imply that greedy routing succeeds throughout all layers, and that it visits no layer twice. So we may apply a union bound over all layers, and obtain that the routing algorithms succeeds with probability $\Omega(1)$ and visits at most one vertex per layer.\footnote{We do not visit every layer. But since a.a.s.\! the stretch is $1+o(1)$, we visit a.a.s.\! a $(1-o(1))$-fraction of all layers.} The whole idea is formalized in our main Lemma~\ref{lem:main}. Then Theorem~\ref{thm:greedysuccess1} directly follows from the main lemma. A more careful analysis along the same lines gives the relaxation result, Theorem~\ref{thm:relaxations}. All other theorems also rely heavily on the layer technique, but require more tricks to handle start and end phases.
 
\begin{figure}[ht!]
  \centering
        \includegraphics[width=\textwidth, clip = true, trim= 35 380 0 0]{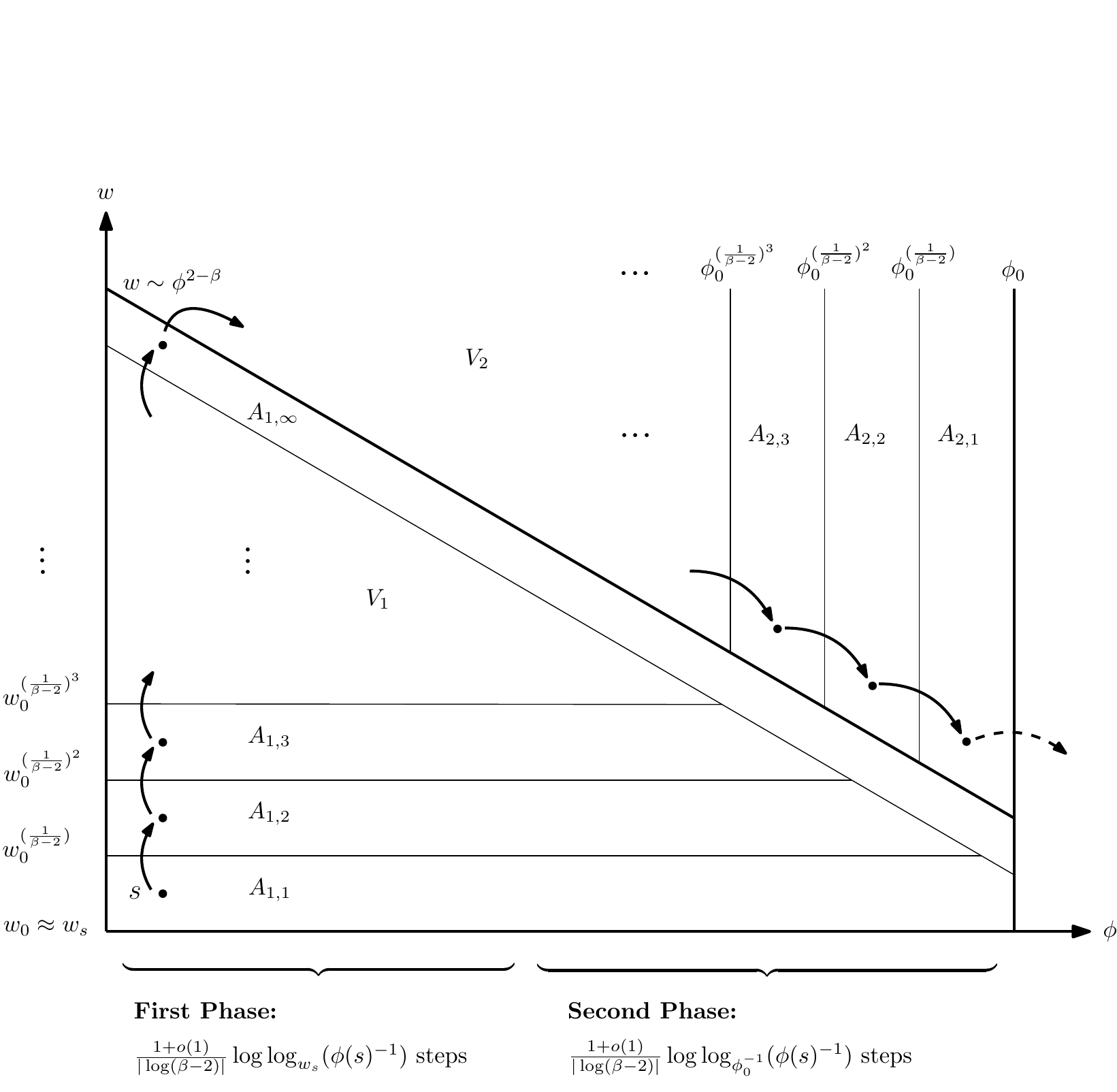}
  \caption{A typical greedy path. In the main Lemma~\ref{lem:main} we prove that if for some $\phi_0 = o(1)$ we uncover the graph $G[V_{\leq \phi_0}]$ then the trajectory looks as depicted if the process gets started at all. Moreover in the last layer we find a vertex which has in expectation (over the randomness in $V_{> \phi_0}$) still many neighbors in the next layer. Note that the greedy path in $G$ does not need to coincide with the greedy path on $G[V_{\leq \phi_0}]$, but it can only deviate by jumping to a vertex of objective larger than $\phi_0$. The special layer $A_{1,\infty}$ is tine and will be needed only for technical reasons.\label{fig:greedypath}}
\end{figure}

\paragraph{Start and end phases:}
It turns out that the failure probability of greedy routing is dominated by the first few steps (when the weight is still constant) and the last few steps (when the objective is constant). In each of these steps the algorithm has probability $\Omega(1)$ to fail. For example, in the typical case that $s$ has a constant weight, there is a constant probability that $s$ has no neighbors at all. However, the number of neighbors of $s$ is Poisson distributed with mean $\Theta(\w s)$. In particular, the probability to have no neighbors decays exponentially with $\w s$. Similar considerations apply for the number of neighbors with better objective. This is the reason why the failure probability decays exponentially with $\wmin$ (Theorem~\ref{thm:greedysuccess2}). The actual proof is much more tricky and is carried out in Section~\ref{subsec:expsuccess}.

Finally, to prove the patching result (Theorem \ref{thm:patching}) we show three intermediate results. \emph{(i)} If we have explored $k$ vertices, starting from $s$, then with probability $1-\exp\{-k^{\Omega(1)}\}$ at least one of these vertices is adjacent to a vertex of weight at least $k^{\Omega(1)}$. By condition (P2) it takes only $k^{O(1)}$ steps to explore $k$ vertices, so after a short exploration phase of $o(\log \log n)$ steps, a.a.s.\! we find a vertex of weight $\omega(1)$. \emph{(ii)} As in the purely greedy case, starting from a vertex of weight $\omega(1)$ a.a.s.\! we follow a typical trajectory of the greedy algorithm as described above, until we find a vertex $v$ with almost constant objective (say, with objective $\phi(v) \approx (\log\log \log n)^{-1}$). This middle phase is purely greedy and does not require patching. In particular, $v$ has better objective than any previously visited vertices. \emph{(iii)} We study the graph $G_{\geq \phi(v)}$ induced by vertices of objective larger than $v$ and find that a.a.s. $G_{\geq \phi(v)}$ contains at most $O(\phi(v)^{-1})$ vertices and a giant component, that contains both the vertices $t$ and $v$. Thus by condition (P3) the algorithm will explore the giant component of $G_{\geq \phi(v)}$ and find $t$ in additional $o(\log \log n)$ steps.

\section{Notation \& Basic Properties of GIRGs}\label{sec:preliminaries}
Before proving our main results, we make several preparations in this chapter. First we introduce some basic notation. Then, in Section~\ref{subsec:girgprops} we give some properties of the GIRG model which hold in general and have no direct connection to greedy routing. Finally, Section~\ref{subsec:basiclemmas} contains lemmas which are more directly related to greedy routing. In particular, there we calculate for a given vertex $v \in \probSpace$ where we expect to find its neighbors with highest objective.

\subsection{Notation}\label{subsec:notation}

Throughout the paper $G= (V,E)$ always denotes an (undirected) random graph given by the GIRG model with vertex set $V$ and edge set $E$. For a vertex $v$, we denote its weight and position by $\w v$ and $\x v$, respectively, and $\Gamma(v)$ is the set of neighbors of $v$ in $G$. For $A \subseteq V$, we denote by $G[A]$ the subgraph of $G$ induced by $A$. Moreover, for $\phi \in \R$ we denote by $V_{\leq \phi}$ and $V_{> \phi}$ the set of vertices with objective at most $\phi$ and larger than $\phi$, respectively.

We remind the reader that the Landau notation $O, \Omega$, etc. may hide the power law exponent $\beta$, the decay parameter $\alpha$, the dimension $d$, and the constants in~\eqref{eq:puv} and~\eqref{eq:puv2}, but not the minimal weight $\wmin$. The Landau notation is always used with respect to $n\to\infty$. We say that an event $\mathcal{E}= \mathcal{E}(n)$ holds \emph{asymptotically almost surely (a.a.s.)} if $\Pr[\mathcal{E}(n)] \to 1$ for $n\to \infty$. We refer to Table~1 on page~\pageref{tab:notation} for general notation used throughout the paper.

\subsection{Properties of GIRGs}\label{subsec:girgprops}
In this short section we collect some basic properties of GIRGs.
%and some general lemmas for the routing.  
We start with some basic lemmas taken from~\cite{bringmann2015generalGIRG}.\footnote{The model in~\cite{bringmann2015generalGIRG} slightly differs from our version (apart from being more general, ~\cite[Theorem 7.3]{bringmann2015generalGIRG}) because the vertices are not given by a Poisson point process, but rather $n$ vertices are placed uniformly at random in $\Space$. However, for any two given vertices $u$ and $v$, the distribution of weights and positions and the connection probabilities are the same. Moreover, the models coincide if we condition our model on the number of vertices. 
%Since $|V|$ is Poisson distributed with mean $n$, this event has probability $\Theta(1/\sqrt{n})$. 
In particular, if a statement holds a.a.s.~in the model in~\cite{bringmann2015generalGIRG} then it also holds a.a.s.~in our model.\label{foot:models}}
%a formula for the marginal probability for an edge between two vertices $u$ and $v$ with given weights, which is taken from~\cite{bringmann2015generalGIRG}.\footnote{The model in~\cite{bringmann2015generalGIRG} differs from the version here because the vertices are not given by a Poisson point process, but rather $n$ vertices are placed uniformly at random in $\Space$. However, after conditioning on the existence of $u$ and $v$ the connection probabilities coincide.}
First we give a formula for the marginal probability for an edge between two vertices $u$ and $v$ with given weights.

\begin{lemma}[{\cite[Lemma 4.3]{bringmann2015generalGIRG}}] \label{lem:marginal}
Let $u,v$ be two vertices of a GIRG. Then
$$\Pr_{\x u, \x v}\big[\{u,v\} \in E \mid \w u, \w v\big] = \Theta\left(\min\left\{\frac{\w u \w v}{\wmin n},1\right\}\right).$$
\end{lemma}

Note that this marginal probability is the same as for Chung-Lu graphs. This is why GIRGs can be interpreted as a geometric variant of Chung-Lu random graphs. Furthermore, the above expression does not change if we fix one of the two coordinates (but not both!). From Lemma~\ref{lem:marginal}, we can directly deduce the following statement about the degree of a given vertex.

\begin{lemma} \label{lem:degvertex}
Let $v \in V$ be a fixed vertex with weight $\w v$. Then the random variable $\deg(v)$ is distributed as $\Pois(\Theta(\w v))$, and in particular $\Ex[\deg(v)]=\Theta(\w v)$.
\end{lemma}
\begin{proof}
The fact that $\deg(v)$ is a Poisson distributed random variable follows immediately since the vertices are given by a Poisson point process. The formula $\Ex[\deg(v)]=\Theta(\w v)$ was shown in~\cite[Lemma 4.4]{bringmann2015generalGIRG}, and the expectation of this model agrees with the expectation in our model, cf.~Footnote~\ref{foot:models}.
\end{proof}

The next lemma asserts that there is a unique giant component (i.e., a component of linear size). Moreover, the average distance in the giant component is the same as the number of steps that greedy routing takes (Theorem~\ref{thm:length}), up to a stretch factor of $1+o(1)$.

\begin{lemma}[{\cite[Theorem 2.2, 2.3]{bringmann2015generalGIRG}}]\label{lem:giant}
A.a.s.~the largest component has linear size, while all other components have size at most $\log^{O(1)} n$. Moreover, a.a.s.~the average distance between vertices in the largest component is $\tfrac{2 \pm o(1)}{|\log(\beta-2)|}\log \log n$. 
\end{lemma}
%\section{Notation Table}\label{sec:notation-table}

%%%%%% Table for general Parameters %%%
\begin{tabular}{ |c|c|c| }
%\begin{table} 
%\begin{center}
%\begin{tabular}{ |c|c|c| }
 \hline
 Notation
 & Definition & Defined on page \\ 
 \hline
 \bf General&&\\
 \hline
 $G=(V,E)$ & random graph from the GIRG model with & \pageref{sec:model} \\
 & vertex set $V$ and edge set $E$ & \\
% $G$ & random graph from the GIRG model& \pageref{sec:model} \\
% $V$ & vertex set of $G$ & \pageref{sec:model} \\ 
% $E$ & edge set of $G$ & \pageref{sec:model} \\ 
 $\Gamma(v)$ & set of neighbors of $v$ in $G$ & \pageref{subsec:notation} \\
 $G[S]$ & subgraph of $G$ induced by $S \subseteq V$ & \pageref{subsec:notation} \\
 $\deg(v)$ & degree of vertex $v$ in $G$ & \pageref{lem:degvertex} \\
  $\Vol(A)$ & volume of $A \subseteq \Space$ &  \pageref{sec:model-geoSpace}\\
 a.a.s. & ``asymptotically almost surely'', probability  $\to 1$ as $n\to \infty$ &  \pageref{subsec:notation}\\
 $\Pois(n)$ &  Poisson distribution with parameter $n$ & \pageref{sec:model} \\
\hline
\bf Model &&\\
\hline
 $n$ & expected number of vertices & \pageref{sec:model-geoSpace}\\
 $d$ & dimension of the geometric space & \pageref{sec:model-geoSpace} \\ 
 $\beta$ & $2 < \beta < 3$, power-law parameter & \pageref{sec:model-weights} \\
 $\wmin$ & minimal weight of a vertex & \pageref{sec:model-weights} \\
 $\alpha$ & decay parameter, which balances between & \pageref{sec:model-edges} \\ & vertex degrees and the geographical distance in $p_{uv}$ & \\
 $p_{uv}$ & edge connection probability between $u$ and $v$ & \pageref{sec:model-edges} \\
  $\Space$ & considered geometric space:$d$ dim. torus ($\R^d / \Z^d$)& \pageref{sec:model-geoSpace} \\
 $\x v$ & position of $v$ in the geometric space $\Space$ & \pageref{sec:model} \\
 $f$ & probability density function of the weights & \pageref{sec:model-weights} \\
 %$(\x v, \w v)$ & representation of a vertex $v$ & \pageref{subsec:routing-routing-protocol} \\ 
 $\w v$ & weight of $v$, $\w v \in \R$ & \pageref{sec:model} \\
 %$\probSpace$ & $\Space \times \R$, is the ground space & \pageref{sec:model-geoSpace} \\
 $c_1, c_2$ & constants in the connection probability $p_{uv}$ & \pageref{sec:model-edges} \\

\hline
\bf Routing&&  \\
\hline
 $s, t$ & starting and destination vertex of the routing & \pageref{sec:model} \\
% $t$ & destination vertex of the routing & \pageref{sec:model} \\
 $m$ & message to be routed & \pageref{alg:greedy}\\
 $\phi$ & objective function & \pageref{subsec:routing-routing-protocol} \\
 $V_{\leq \phi}$, ($V_{> \phi}$) & set of vertices with objective $\leq \phi$, ($> \phi$) & \pageref{subsec:notation}\\
 $\tilde{\phi} $ & approximative objective function & \pageref{thm:relaxations} \\
% $\phi_H$ & objective on random hyperbolic graphs & \pageref{defn:phi_H} \\
% $\phi(v)$ & objective at vertex $v$ of the greedy routing & \pageref{subsec:routing-objective-function} \\
 %$\lambda$ & Lebesgue measure on $\Space$ & \pageref{sec:model} \\
 %$\nu$ & measure on $\R$ & \pageref{sec:model} \\
 %$\mu$ & $ \lambda \times \nu$, the product measure on $\probSpace$ & \pageref{sec:model} \\
 %$r$ & distance $\|\x u - \x v\|$ between $u$ and $v$ \an{doubly defined?} & \pageref{} \\
\hline
\bf Proofs && \\
\hline
% $\Ex[\mathcal{E}]$ & expected value of an event $\mathcal{E}$ & \pageref{lem:degvertex}\\
% $\varepsilon$ & small constant $ >0$ & \pageref{subsec:basiclemmas} \\
 $\gamma(\varepsilon)$ & $ \frac{1-\varepsilon}{\beta-2}$, a constant & \pageref{subsec:basiclemmas} \\
 $\varepsilon_1$ & $\varepsilon_1(\alpha,\beta)>0$, a fixed global constant & \pageref{subsec:basiclemmas} \\
 $\probSpace_1$ & vertex set corresponding to the first phase of routing, & \pageref{subsec:basiclemmas}  \\
 & $\{v \in \probSpace: \phi(v) \le \w v^{-\gamma(\varepsilon_1)}\}$ &  \\
 $\probSpace_2$ & vertex set corresponding to the second phase of routing, & \pageref{subsec:basiclemmas} \\
 & $\{v \in \probSpace: \phi(v) \ge \w v^{-\gamma(\varepsilon_1)}\}$ & \\
 $\zeta$ & $\max\{\frac{3}{2},\frac{2\alpha-1}{2\alpha+4-2\beta}\}$ for $\alpha<\infty$,\, $\frac{3}{2}$ for $\alpha=\infty$ & \pageref{subsec:basiclemmas-definition-zeta} \\
 $V^+(v,\varepsilon)$ & set of good vertices: large weight and larger objective & \pageref{subsec:basiclemmas-definition-zeta} \\
 $V^-(v,\varepsilon)$ & set of bad vertices: small weight and large objective  & \pageref{subsec:basiclemmas-definition-zeta} \\
% & \an{do we need to distinguish between 1st, 2nd phases?} & \\
 $w_1(\varepsilon)$ & $O(e^{d/\varepsilon})$ & \pageref{subsec:mainlemma} \\
 $\phi_1(\varepsilon)$ & $\Omega(e^{-d/\varepsilon})$ & \pageref{subsec:mainlemma} \\
% $\varepsilon_2$ & $(\log\log\log n)^{-1}$ & \pageref{lem:main-varepsilon_2-defn} \\
 %merged from here
 $w$-grid & grid splitting $\Space$ into $\frac{n}{w}$ equal and disjoint cubes & \pageref{def:grid}\\
% $w_0$ & $\omega(1)$ & \pageref{lem:bulk}\\
% $r$ & $\omega(n^{-d})$ & \pageref{cor:subcore} \\ 
% $\delta$ & $(\frac{c_1 \w v^{1+\gamma(\varepsilon)}}{\wmin n})^{1/d}$ & \pageref{lem:firstphase}\\
% $\delta(w)$ & $\|\x v - \x t\| (w \w v^{-\gamma(\varepsilon)})^{1/d}$ & \pageref{defn-4nt:delta(w)} \\
% $\delta'(w)$ & $(\frac{w w_0^c}{\wmin n})^{1/d}$ & \pageref{defn-4nt:delta'(w)}\\
% $w_1$ & $\max\{\log (w_0),w_t^{\delta}\}$ & \pageref{eq:patchingphiu1} \\
 %$f_0(n)$ & $\omega(1)$ & \pageref{lem:main} \\
 %$\phi_0$ & $w_1^{-1/f_0}$ & \pageref{eq:patchingphiu1} \\  %Removed \phi_0
 %$\phi_0$ & $(\log w_1)^{-1}$ & \pageref{eq:patchingphiu1} \\
 $\overline{\probSpace}(w,\phi)$ & $\{v \in \probSpace_1 \mid \w v \ge w  \wedge \phi(v) \le \phi\} \cup \{v \in \probSpace_2 \mid \phi(v) \le \phi\}$ & \pageref{subsec:mainlemma}\\
% $w_0'\,\&\, \phi_0'$ & $w_0^{(\gamma(\zeta\varepsilon_1)^{f_0(n)})}\,\&\,\phi_0^{(\gamma(\varepsilon_1)^{f_0(n)})}$(resp.), thresholds of layering & \pageref{lem:main-varepsilon_2-defn} \\
% $y_{j+1}$ &  $y_j^{\gamma(\zeta\varepsilon_1)}$ if $y_j < w_0'$, \, $y_j^{\gamma(\varepsilon_2)}$ if $y_j \ge w_0'$, \, $y_0 = w_0$ & \pageref{def:yj}\\ 
% $\probSpace_1'$ & $\{v \in\overline{\probSpace}(w_0,\phi_0) \mid \phi(v)\w v^{\gamma(\varepsilon_2)} \le 1 \}$ & \pageref{def:Aij}\\
% $A_{1,j}$ & $\{v \in \probSpace_1' \mid y_{j-1} \le \w v < y_j \}~~~\forall j\ge 1$ & \pageref{def:Aij}\\
% $A_{1,\infty}$ & $\{v \in \overline{\probSpace}(w_0,\phi_0) \mid \phi(v)\w v^{\gamma(\varepsilon_1)} \le 1 \le \phi(v) \w v^{\gamma(\varepsilon_2)} \}$ & \pageref{def:Aij}\\
% $\psi_{j+1}$ & $\psi_j^{\gamma(\varepsilon_1)}$ if  $\psi_j > \phi_0'$, \, $\psi_j^{\gamma(\varepsilon_2)}$ if $\psi_j \le \phi_0'$, \, $\psi_0=\phi_0$ & \pageref{def:psi_j}\\
% & & \\
 \hline
%\end{tabular}
%\end{center}
%\caption{A general notation table}
%\label{tab:notation}
\end{tabular}
%\end{table} 

 \label{tab:notation}
\begin{center}
Table~1: A general notation table
\end{center}
\clearpage

For a fixed vertex $v$, the following lemma shows that it is unlikely for $v$ to have a neighbor of much larger weight. We will need this auxiliary result only for the proof of Theorem~\ref{thm:relaxations}, i.e., for studying the relaxations.

\begin{lemma} \label{lem:weighttoohigh}
Let $\varepsilon>0$ and let $v \in V$ be a fixed vertex of weight $\w v$. Then the expected number of neighbors of $v$ with weight at least $w^+ := \w v^{(1+\varepsilon)/(\beta-2)}$ is $O(\wmin^{\beta-2}\w v^{-\varepsilon})$.
\end{lemma}
\begin{proof}
We calculate the expected value with a straight-forward integral over all vertex weights larger than $w^+$. The probability measure is given by the density $f(w)$, and for a vertex $u$ of fixed weight $\w u$ and random coordinate $\x u$, the probability to have an edge with $v$ is at most $O(\frac{\w u \w v}{\wmin n})$ by Lemma~\ref{lem:marginal}. Then the expected number of such neighbors of $v$ is at most
$$n \cdot \int_{w^+}^{\infty} f(w) \cdot O\left(\frac{\w v w}{\wmin n}\right) dw= O\left(\wmin^{\beta-2} \w v \int_{w^+}^{\infty} w^{1-\beta} dw\right)=O\left(\wmin^{\beta-2} \w v^{-\varepsilon}\right).$$
\end{proof}

Next, we compute the total number of vertices whose objective is larger than a given $\phi_0$.

\begin{lemma} \label{lem:verticesofhighobjective}
Let $\phi_0$ be a given objective, and denote by $V_{\ge \phi_0}$ the set of vertices which have objective at least $\phi_0$. Then with probability $1-e^{-\Omega(\phi_0^{-1})}$, it holds $|V_{\ge \phi_0}|=\Theta(\phi_0^{-1})$.
\end{lemma}

\begin{proof}
By definition of the objective function $\phi$, a vertex $u$ satisfies $\phi(u)\ge \phi_0$ if and only if $\|\x u - \x t\|^d \le \frac{\w u}{\phi_0 \wmin n}$. We integrate over all weights $w$ and see that the expected number of vertices with objective at least $\phi_0$ is 
$$n \wmin^{\beta-1}\int_{\wmin}^{\infty} w^{-\beta} \frac{w}{\phi_0\wmin n} dw = \Theta(\phi_0^{-1}).$$
Then by a Chernoff bound, with probability $1-e^{-\Omega(\phi_0^{-1})}$, this random variable is concentrated around its expectation and thus $\Theta(\phi_0^{-1})$.
\end{proof}

The following  lemma is a consequence of the Poisson point process. It will help us below to deal with events which are not independent as its statement shows how certain random variables are correlated.
\begin{lemma} \label{lem:correlatedprobs}
Let $A_1, A_2 \subset \Space$ and let $u_1, u_2 \notin A_1 \cup A_2$ be two vertices. Then
$$\Pr[\Gamma(u_2) \cap A_2 = \emptyset \mid \Gamma(u_1) \cap A_1 = \emptyset] \ge \Pr[\Gamma(u_2) \cap A_2 = \emptyset].$$
\end{lemma}
\begin{proof}
Let $A_1, A_2 \subset \Space$ be two subsets and let $u_1, u_2$ be two vertices not contained in the sets $A_1, A_2$. Suppose we know that $u_1$ has no neighbors in $A_1$. Clearly, conditioning on this event decreases the expected number of vertices in $A_1$ and in every subset of $A_1$, hence
$$\Pr[\Gamma(u_2) \cap (A_1 \cap A_2) = \emptyset \mid \Gamma(u_1) \cap A_1 = \emptyset] \ge \Pr[\Gamma(u_2) \cap (A_1 \cap A_2) = \emptyset].$$
Furthermore, by the Poisson point process, for every set $A_3$ which is disjoint from $A_1$, the number of vertices in $A_1$ and $A_3$ is independent. We put $A_3 := A_2 \setminus A_1$ and deduce
$$\Pr[\Gamma(u_2) \cap (A_2 \setminus A_1) = \emptyset \mid \Gamma(u_1) \cap A_1 = \emptyset] = \Pr[\Gamma(u_2) \cap (A_2 \setminus A_1) = \emptyset].$$
Combining the two inequalities gives the desired property.
\end{proof}
 
Some of the proofs require that we split the geometric space into small subspaces. We will do this with grids which divide the torus into smaller $d$-dimensional cubes.

\begin{definition} \label{def:grid}For a given weight $w=w(n)$, a $w$-grid is a grid which splits the geometric ground space $\Space$ into $\frac{n}{w}$ equal and disjoint cubes of side lengths $(\frac{w}{n})^d$.
\end{definition}

Here, for simplifying notations, we assume that $w$ divides $n$. One application of $w$-grids is the following lemma. For a given $w$-grid and a given vertex set $S$, it states that if $S$ contains vertices of sufficiently many different cells of the grid, then at least one vertex $v \in S$ will have a neighbor of weight at least $w$.

\begin{lemma}[Bulk lemma] \label{lem:bulk} 
Let $w_0=w(n)= \omega(1)$ be a weight growing in $n$ and let $S$ be a set of vertices of weight at most $w_0$ such that $S$ contains vertices in at least $w_0$ cells of a fixed $w_0$-grid. Then a.a.s., there exists a vertex $v \in S$ and a vertex $u$ of weight at least $w_0$ such that $u$ and $v$ are neighbors and $u$ and $v$ are contained in the same cell of the $w_0$-grid.
\end{lemma}
\begin{proof}
We uncover the graph in two steps: First we only consider the vertices of weight less than $w_0$ and afterwards we uncover the remaining vertices. Let $S$ be a fixed set of vertices with weight less than $w_0$, and let $C_1, \ldots, C_{w_0}$ be cells of a $w_0$-grid which contain at least one vertex of $v_i \in S \cap C_i$.

Now, we insert the remaining vertices of high weight at random, according to the definition of the model. Together, the $w_0$ cells have volume $w_0 \cdot \frac{w_0}{n}$. Therefore, a single vertex $u$ of weight $\w u \ge w_0$ falls with probability $\frac{w_0^2}{n}$ into one of the $w_0$ cells. Furthermore, if $u$ falls into such a cell $C_i$, then $\|\x u-\x {v_i}\|^d=O(\frac{w_0}{n})$. By definition, $v_i$ has weight at least $\wmin$, and $(EP1)$ or $(EP2)$ imply
$$\Pr[\{u,v_i\} \in E \mid u,v_i \in C_i] = \Omega(1).$$
Thus every vertex of weight at least $w_0$ is connected to the set $\{v_1, \ldots, v_{w_0}\}$ with probability at least $\Omega(\frac{w_0^2}{n})$. Note that the expected number of vertices of weight at least $w_0$ is $\Theta(n \wmin^{\beta-1}w_0^{1-\beta})$, hence the expected number of vertices of weight at least $w_0$ which are connected to a vertex of $\{v_1, \ldots, v_{w_0}\}$ is at least $\Omega(\wmin^{\beta-1}w_0^{3-\beta})=\omega(1)$. By a Chernoff bound, it follows that a.a.s., there exists a vertex $v_i$ in a cell $C_i$ such that $v_i$ has at least one neighbor $u$ with $\w u \ge w_0$.
\end{proof}

Note that the vertices of a subcube $A$ induce a random subgraph. With the next lemma we show that this induced subgraph is itself a GIRG. %It is based on our studies on generalisations of the present random graph model in \cite{bringmann2015generalGIRG}.
%\ym{\cite{bringmann2015generalGIRG} studies a different model}
\begin{lemma} \label{lem:subgirg}
Let $r \leq 1/4$, let $A$ be a subcube of $\Space$ with radius $r$. Then with appropriate rescaling the induced subgraph $G[A]$ itself is a GIRG with intensity $n \Vol(A)$, except with geometric space $[0,1]^d$ instead of $\Space$. %In particular, it is a scale-free network with underlying geometry as discussed in $\cite{bringmann2015generalGIRG}$.
\end{lemma}
%\ym{I am not really satisfied with this Lemma. Is the Lemma itself already proven in \cite{bringmann2015generalGIRG}?}
\begin{proof}
We rescale the cube $A$ by a factor $1/\Vol(A) = (2r)^{-d}$, thus transforming it into $A' = [0,1]^d$. Then the vertices of $A'$ are given by a Poisson point with intensity $n_{A'}:=n\Vol(A)$, and the weights follow the same distribution as before. Note that by the condition $r\leq 1/4$ the distances in $A$ are the same as the distances in $A'$, only scaled up by a factor $2r$. We only consider the case $\alpha <\infty$; the case $\alpha = \infty$ is analogous. For any two vertices $u,v \in A$, by \eqref{eq:puv} the connection probability is
\[
p_{uv} = \Theta\Big( \min\Big\{ \frac{1}{\|\x u - \x v\|_A^{\alpha d}} \cdot \Big( \frac{\w u \w v}{\wmin n}\Big)^{\alpha} ,1 \Big\} \Big) = \Theta\Big( \min\Big\{ \frac{1}{\|\x u - \x v\|_{A'}^{\alpha d}} \cdot \Big( \frac{\w u \w v}{\wmin n_{A'}}\Big)^{\alpha} ,1 \Big\} \Big),
\]
which is exactly~\eqref{eq:puv} for the distances of $A'$ and intensity $n'$. %This proves the first statement. For the second statement it suffices to observe that~\cite[Theorem 2.2]{bringmann2015generalGIRG} (our Lemma~\ref{lem:giant}) was proven for more general geometric spaces, including the space $[0,1]^d$.
%We briefly show that for a fixed subcube $A$, the induced subgraph $G[A]$ itself is a random graph which satisfies all desired properties to be contained in this class of generalized GIRGs. Let $n'$ be the number of vertices in $G[A]$. Clearly, the subcube $A$ has volume $(2r)^d$, therefore $\Ex[n']=(2r)^d n$, and since $r^d=\omega(n^{-1})$, this number is concentrated and a.a.s.\ it holds $n'=\Theta(r^d n)$. Furthermore, the weights of the vertices still follow a power law. Finally, let $u,v$ be two vertices of $G[A]$ such that $\w u, \w v, \x u$ are fixed but $\x v$ is still a random variable. Then it is not difficult to calculate that
%$$\Ex_{\x v}\left[p_{uv} \mid \w u, \w v, \x u\right]=\Theta\left(\min\left\{1, \frac{\w u \w v}{\wmin n} \cdot s^{-d}\right\}\right)=\Theta\left(\min\left\{1, \frac{\w u \w v}{\wmin n'}\right\}\right).$$
%Hence the random graph $G[A]$ can be seen as an example of this general model.
\end{proof}

Lemma~\ref{lem:subgirg} is useful since all results from~\cite{bringmann2015generalGIRG} were proven for more general geometric spaces, including the space $[0,1]^d$. Thus they apply to the induced graph $G[A']$. In particular, assume $r=\omega(n^{-d})$, so that the intensity is $n'=\omega(1)$. Then a.a.s.\!~there is a giant component of size $\Theta(n')$ in $G[A]$ by~\cite[Theorem 2.2]{bringmann2015generalGIRG}. Moreover, by~\cite[Lemma 5.2]{bringmann2015generalGIRG} there exists a constant $c>0$ such that a.a.s.\! every vertex of weight at least $(\log n')^c$ is contained in the giant connected component of such a random graph. Thus we obtain the following corollary.

\begin{corollary} \label{cor:subcore}
Let $r=\omega(n^{-d})$, $r \leq 1/4$, let $A$ be a subcube of $\Space$ with radius $r$, and let $G[A]$ be the random subgraph induced by all vertices in $A$. Then there exists a constant $c>0$ such that all vertices in $A$ of weight at least $(\log (r^d n))^c$ are in the same connected component of $G[A]$.
\end{corollary}
%\rk{I know that this is the precise statement which we obtain by applying Lemma~5.2 from the other paper. But for using this Corollary in the patching proof, it would be convenient if we say that a.a.s.\ the vertices of weight at least $n'^{\varepsilon}$ form a connected subgraph. The proof of that lemma directly gives us this property, although we didn't state it like this... opinions?}\jl{I would say it is fine if we give the statement you want, and then have a footnote like ``Actually the formulation in~\cite[Lemma 5.2]{bringmann2015generalGIRG} was only that aas for all vertices of weight at least $n'^{\varepsilon}$ there exists a path to a set of very heavy vertives, which is aas connected. However, in the proof it was actually shown that all these paths are completely contained in $V_{\geq n'^{\varepsilon}}$ (the paths are increasing in weights), so we get our formulation here.''}

\subsection{Where to Expect Neighbors \& Unlikely Jumps} \label{subsec:basiclemmas}

We start by introducing technical notation. For all $\varepsilon>0$, we put $\gamma(\varepsilon) := \frac{1-\varepsilon}{\beta-2}$. Next, let \mbox{$\varepsilon_1 = \varepsilon_1(\alpha,\beta)>0$} be a fixed constant which we choose sufficiently small during proofs, e.g., such that $\gamma(\varepsilon_1)>1$. Next, we define the following two classes of vertices:
$$\probSpace_1 := \{v \in \probSpace: \phi(v) \le \w v^{-\gamma(\varepsilon_1)}\}
\quad 
\text{and}
\quad \probSpace_2 := \{v \in \probSpace: \phi(v) \ge \w v^{-\gamma(\varepsilon_1)}\}.$$

In every hop, the process chooses a neighbor $u$ of the current vertex $v$ which maximizes $\phi(u)$. As discussed above in Section~\ref{sec:proofsketch}, it turns out that if $v \in \probSpace_1$, then typically $\w u$ is by an exponent $\gamma(\varepsilon)$ larger than $\w v$. On the other hand, if $v \in \probSpace_2$, then we expect that $\phi(u) \sim \phi(v)^{1/\gamma(\varepsilon)}$. Therefore, vertices in $\probSpace_1$ will correspond to the first phase of the routing, where the current weight increases by an exponent with every hop. The vertices in $\probSpace_2$ correspond to the second phase where the objective is increased by an exponent with every hop (due to $\phi(v)<1$, indeed the objective increases). 

In order to find such a trajectory, we want to classify the vertices in $\probSpace$ which accelerate the routing as desired. Let $\zeta := \max\{\frac{3}{2},\frac{2\alpha-1}{2\alpha+4-2\beta}\}$ for $\alpha<\infty$ and $\zeta := \frac{3}{2}$ in the limit case $\alpha=\infty$.\label{subsec:basiclemmas-definition-zeta} Then for $v \in \probSpace_1$ and $\varepsilon=\varepsilon(n)>0$ we define
\begin{align} \label{eq:areas1}
\begin{split}
V^+(v,\varepsilon) & :=\{u \in \probSpace \mid \w u \ge \w v^{\gamma(\varepsilon)} \wedge \phi(u) \ge \phi(v) \w v^{\gamma(\varepsilon)-1} \}, \text{ and}\\
V^-(v, \varepsilon) & :=\{u \in \probSpace \mid  \w u \le \w v^{\gamma(\zeta\cdot\varepsilon)} \wedge \phi(u) \ge  \phi(v) \w v^{\gamma(\varepsilon)-1}\}.
\end{split}
\end{align}
In the first phase, the goal will be to increase the weight within every hop of the routing. Therefore, $\Gamma(v) \cap V^+(v,\varepsilon)$ is the set of \emph{good} neigbors as every $u \in \Gamma(v) \cap V^+(v,\varepsilon)$ has significantly larger weight than $v$. On the other side, the set $\Gamma(v) \cap V^-(v,\varepsilon)$ contains the \emph{bad} neighbors as its vertices have small weight but nevertheless large objective and therefore could force the routing to proceed with a low-weight-vertex.

In the second phase, the goal is to increase the objective by an exponent in every step of the routing. Furthermore, once we reach vertices of $\probSpace_2$, we want to stay in $\probSpace_2$. Let $v \in \probSpace_2$. Similarly as above, we introduce a set $V^+(v,\varepsilon)$ of good vertices and a set $V^-(v,\varepsilon)$ of bad vertices. More precisely, for $\varepsilon=\varepsilon(n)$ we define
\begin{align} \label{eq:areas2}
\begin{split}
V^+(v,\varepsilon) & :=\{u \in \probSpace_2 \mid \phi(u) \ge \phi(v)^{1/\gamma(\varepsilon)} \}, \text{ and}\\
V^-(v, \varepsilon) & :=\{u \in \probSpace_1 \mid  \phi(u) \ge \phi(v)^{1/\gamma(\varepsilon)} \}.
\end{split}
\end{align}

In the following we collect several technical statements about the expected number of good and bad neighbors. For $0<\varepsilon \le \varepsilon_1$, let $w_1(\varepsilon) := O(e^{d/\varepsilon})$ and $\phi_1(\varepsilon) := \Omega(e^{-d/\varepsilon})$ with sufficiently large respectively small hidden constants. We start with the first lemma which considers vertices in $\probSpace_1$ and shows that the expected number of good neighbors is large whereas the expected number of bad neighbors is small. Notice that for $\varepsilon=\varepsilon_1$, every vertex $v \in \probSpace_1$ satisfies $\phi(v) \w v^{\gamma(\varepsilon_1)} \le 1$ by the definition of $\probSpace_1$, and in this case, the statement can be applied for all vertices in $\probSpace_1$.

\begin{lemma}[Neighborhoods in first phase] \label{lem:firstphase} 
Let $0<\varepsilon \le \varepsilon_1$ and let $v \in \probSpace_1$ be a vertex such that $\phi(v) \w v^{\gamma(\varepsilon)} \le 1$. Then 
\begin{enumerate}[(i)]
\item $\Ex[|\Gamma(v) \cap V^+(v,\varepsilon)|]=\Omega\left(\wmin^{\beta-2}\w v^{\varepsilon}\right)$, and
\item if in addition $\w v \ge w_1(\varepsilon)$, then $\Ex[|\Gamma(v) \cap V^-(v,\varepsilon)|]=O\left(\wmin^{\beta-2}\w v^{-\Omega(\varepsilon)}\right)$.
\end{enumerate}
\end{lemma}
\begin{proof}
 Let $0<\varepsilon\le\varepsilon_1$ be such that $\phi(v) \w v^{\gamma(\varepsilon)} \le 1$. Let $v \in \probSpace_1$.
\medskip

\textbf{Proof of {\boldmath$(i)$}:}
 In order to prove the first statement, we want to lower-bound the expected number of neighbors of $v$ in $V^+(v,\varepsilon)$, i.e., the neighbors with larger weight and slightly larger objective. For this, define $\delta:=(\frac{c_1 \w v^{1+\gamma(\varepsilon)}}{\wmin n})^{1/d}$, where $c_1$ is the constant given by (EP1) and (EP2). Then let $A(v,\varepsilon)$ be the following set of vertices:
\[A(v,\varepsilon) :=\{u\in \probSpace \mid (1): w_u\geq \w v^{\gamma(\varepsilon)}; (2): \|\x u - \x t\| \le \|\x v - \x t\|; (3): \|\x u - \x v\| \le \delta\}.\] 
Note that $ A(v,\varepsilon)\subseteq V^+(v,\varepsilon) $, because for $u\in A(v,\varepsilon)$ 
$$\phi(u) = \frac{\w u}{\wmin n\|\x u - \x t\|^d } \ge \frac{\w v^{\gamma(\varepsilon)}}{\wmin n\|\x v - \x t\|^d } \ge \phi(v)\w v^{\gamma(\varepsilon)-1}$$
holds.
Hence for proving the first statement it is sufficient to lower-bound $\Ex[|\Gamma(v) \cap A(v,\varepsilon)|]$. 
By the probability distribution of the weights there are in expectation $n \wmin^{\beta-1} \w v^{\gamma(\varepsilon)(1-\beta)}$ vertices satisfying Condition~$(1)$ and Condition~$(3)$ is satisfied independently with probability $\delta^d$.
Now, observe that by assumption the vertex $v$ satisfies $\phi(v)\w v^{\gamma(\varepsilon)} \le 1$, which implies
$$
\delta^d = \frac{c_1 \w v^{1+\gamma(\varepsilon)}}{\wmin n} =c_1\phi(v)\w v^{\gamma(\varepsilon)}\|\x v - \x t\|^d\le c_1\|\x v - \x t\|^d.
$$
Hence a random vertex $v$ fufillling $(3)$ satisfies with constant probability $(2)$ as well.
  Furthermore, Condition~$(3)$ ensures that by (EP1) and (EP2), every vertex $u \in A(v,\varepsilon)$ is connected to $v$ with constant probability. Combining this yields
$$\Ex[|\Gamma(v) \cap A(v,\varepsilon)|] = \Omega\left(n \wmin^{\beta-1} \w v^{\gamma(\varepsilon)(1-\beta)}\delta^d\right)=\Omega\left(\wmin^{\beta-2}\w v^{1-\gamma(\varepsilon)(\beta-2)}\right)=\Omega\left(\wmin^{\beta-2}\w v^{\varepsilon}\right).$$

\textbf{Proof of {\boldmath$(ii)$}:} We want to upper-bound the expected number of bad neighbors of $v$, i.e., the neighbors in $V^-(v,\varepsilon)$. 

At first, we will show that the geometric distance between a bad node $u\in V^-(v,\varepsilon)$ and $v$ is not much less than the geometric  distance between $v$ and $t$. This will imply a sufficiently small connection probability between $u$ and $v$. For this purpose, define \label{defn-4nt:delta(w)}$\delta(w) := \|\x v - \x t\| (w \w v^{-\gamma(\varepsilon)})^{1/d}$ for all weights $w \le \w v^{\gamma(\zeta \cdot \varepsilon)}$.  Due to $\w v \ge w_1(\varepsilon)=2^{O(d/\varepsilon)}$, for $w_1$ large enough every weight $w \le \w v^{\gamma(\zeta \cdot \varepsilon)}$ satisfies $\delta(w)^d \le \|\x v - \x t\|^d \w v^{\gamma(\zeta\cdot\varepsilon)-\gamma(\varepsilon)} \le (0.5 \|\x v - \x t\|)^d$ (the second inequality follows since $\zeta>1$). Moreover, for a vertex $u \in V^-(v,\varepsilon)$ we observe that the condition $\phi(u)\ge \phi(v) \w v^{\gamma(\varepsilon)-1}$  implies $\|\x u - \x t\| \le \delta(\w u)$ and with the triangle inequality, we obtain that $u$ fufillls $\|\x u - \x v\| \ge 0.5 \|\x v - \x t\|$. 

In order to upper-bound $\Ex[|\Gamma(v) \cap V^-(v,\varepsilon)|]$, we proceed with a case distinction regarding the parameter $\alpha$. 
\medskip

\textbf{Case {\boldmath$\alpha=\infty$}:}
Let $c_2>0$ be the constant given by $(EP2)$, and let $w_1(\varepsilon)$ be large enough. Then 
$$0.5^d \ge c_2 w_1^{-\varepsilon(\zeta-1)/(\beta-2)} \ge c_2\w v^{-\varepsilon(\zeta-1)/(\beta-2)}=c_2 \w v^{\gamma(\zeta \cdot \varepsilon)-\gamma(\varepsilon)}.$$ 
Furthermore, $\phi(v) \w v^{\gamma(\varepsilon)}\le 1$ implies 
$$\frac{\w v^{1+\gamma(\varepsilon)}}{\wmin n} = \phi(v)\w v^{\gamma(\varepsilon)}\|\x v - \x t\|^d\le \|\x v - \x t\|^d.$$ 
Together with the above observation we deduce
$$\|\x u - \x v\|^d \ge 0.5^d \|\x v - \x t\|^d \ge c_2\|\x v - \x t\|^d \w v^{\gamma(\zeta\cdot\varepsilon)-\gamma(\varepsilon)}\ge \frac{c_2\w v^{1+\gamma(\zeta\cdot\varepsilon)}}{ \wmin n} \ge \frac{c_2\w v \w u}{\wmin n}.$$
But then, by (\ref{eq:puv2}) the two vertices $u$ and $v$ can not be connected, hence $v$ has \emph{deterministically} no neighbor in $V^-(v,\varepsilon)$ in the case $\alpha=\infty$.
\medskip

\textbf{Case {\boldmath$\alpha<\infty$}:} By the geometrical property $\|\x u - \x v\| \ge 0.5 \|\x v - \x t\|$ and (\ref{eq:puv}) the probability that a vertex $u \in V^-(v,\varepsilon)$ is connected to $v$ is at most 
$$O\left(\left(\frac{\w u \w v}{\|\x u - \x v\|^d \wmin n}\right)^{\alpha}\right)=O\left(\left(\frac{\w u \w v}{\|\x v - \x t\|^d \wmin n}\right)^{\alpha}\right)=O\left((\w u \phi(v))^{\alpha}\right).$$ 
Furthermore, we have already seen that such a vertex of weight $w$ needs to be in distance at most $\delta(w)$ to $t$, which a random vertex of weight $w$ does with probability $\delta(w)^d$. We integrate over all weights between $\wmin$ and $\w v^{\gamma(\zeta \cdot \varepsilon)}$, where we use the density function $f(w)$ for the distribution of the weights. We deduce that 
\begin{align*}
\Ex[|\Gamma(v) \cap V^-(v,\varepsilon)|] & = O\left(n \int_{\wmin}^{\w v^{\gamma(\zeta \cdot \varepsilon)}}f(w)\delta(w)^d (\w u \phi(v))^{\alpha} dw \right)\\
& = O\left(\wmin^{\beta-2}\phi(v)^{\alpha-1} \w v^{1-\gamma(\varepsilon)}\int_{\wmin}^{\w v^{\gamma(\zeta\cdot\varepsilon)}} w^{\alpha+1-\beta}dw\right)\\
& = O\left( \wmin^{\beta-2}\phi(v)^{\alpha-1} \w v^{1-\gamma(\varepsilon)+\gamma(\zeta\cdot\varepsilon)(\alpha+2-\beta)} \right).
\end{align*}
Next we use that $\phi(v)\w v^{\gamma(\varepsilon)} \le 1$ holds. Then the above term is at most $O(\wmin^{\beta-2} \w v^{\tau})$, where
$$\tau = 1-\alpha\gamma(\varepsilon)+\gamma(\zeta\cdot\varepsilon)(\alpha+2-\beta)=\frac{\alpha\varepsilon - \zeta\varepsilon(\alpha+2-\beta)}{\beta-2} = -\Omega(\varepsilon)$$
by our choice of $\zeta$.
\end{proof}
% OLD VERSION:
%\begin{lemma}[neighborhoods in first phase] \label{lem:firstphase} For all vertices $v \in \probSpace_1$, it holds 
%\begin{enumerate}[(i)]
%\item $\Ex[|\Gamma(v) \cap (V^+(v) \cup V^0(v))|]=\Omega\left(\wmin^{\beta-2}\w v^{2\varepsilon}\right)$.
%\end{enumerate}
%Moreover there exists a constant weight $w_1 = 2^{O(d/\varepsilon)}$ such that for all $v \in \probSpace_1$ with $\w v \ge w_1$,
%\begin{enumerate}
%\item[(ii)] $\Ex[|\Gamma(v) \cap V^+(v)|]=\Omega\left(\wmin^{\beta-2}\w v^{2\varepsilon}\right)$, and
%\item[(iii)] $\Ex[|\Gamma(v) \cap (V^0(v) \cup V^-(v))|]=O\left(\wmin^{\beta-2}\w v^{-\Omega(\varepsilon)}\right)$.
%\end{enumerate}
%%Moreover, with probability $1-\exp(-\Omega(\wmin^{\beta-2}\w v^{2\varepsilon}))$ there exists a neighbor of $v$ in $V^+(v)$ and with probability $1-O(\wmin^{\beta-2}\w v^{-\Omega(\varepsilon)})$, $v$ has no neighbors in $V^-(v)$.
%\end{lemma}

Next, similar to Lemma~\ref{lem:firstphase} for the first phase we give a lemma for the second phase where the routing process reached $\probSpace_2$. We calculate for a vertex $v\in \probSpace_2$ the expected number of good and bad neighbors. 

\begin{lemma}[Neighborhoods in second phase] \label{lem:secondphase} Let $0<\varepsilon\le\varepsilon_1$ and let $v \in \probSpace_2$ be a vertex such that $\phi(v) \le 1$. Then 
\begin{enumerate}[(i)]
\item $\Ex[|\Gamma(v) \cap V^+(v,\varepsilon)|]=\Omega\left(\wmin^{\beta-2}\phi(v)^{-\Omega(\varepsilon)}\right)$, and
\item if in addition $\phi(v) \le \phi_1(\varepsilon)$, then $\Ex[|\Gamma(v) \cap V^-(v,\varepsilon)|]=O\left(\wmin^{\beta-2}\phi(v)^{\Omega(\varepsilon)}\right)$.
\end{enumerate}
\end{lemma}
% OLD STATEMENT:
%\begin{lemma}
%For all vertices $v \in  \probSpace_2$ with $\phi(v) \le 1$, it holds
%\begin{enumerate}[(i)]
%\item  $\Ex[|\Gamma(v) \cap (V^+(v) \cup V^0(v))|]=\Omega(\wmin^{\beta-2}\phi(v)^{-\Omega(\varepsilon)})$.
%\end{enumerate}
%Moreover there exists a constant objective $\phi_1 = 2^{-O(d/\varepsilon)}$ such that for all $v \in \probSpace_2$ with $\phi(v) \le \phi_1$, we have
%\begin{enumerate}
%\item[(ii)] $\Ex[|\Gamma(v) \cap V^+(v)|]=\Omega(\wmin^{\beta-2}\phi(v)^{-\Omega(\varepsilon)})$, and
%\item[(iii)] $\Ex[|\Gamma(v) \cap (V^0(v) \cup V^-(v))|]=O(\wmin^{\beta-2}\phi(v)^{\Omega(\varepsilon)})$.
%\end{enumerate}
%%Moreover, with probability $1-\exp(-\Omega(\wmin^{\beta-2}\phi(v)^{-\Omega(\varepsilon)}))$ there exists a neighbor of $v$ in $V^+(v)$ and with probability $1-O(\wmin^{\beta-2}\phi(v)^{\Omega(\varepsilon)})$, $v$ has no neighbors in $V^-(v)$.
%\end{lemma}

\begin{proof}
We prove the statement similar as we proved Lemma~\ref{lem:firstphase}.  Let $0<\varepsilon\le\varepsilon_1$.
\medskip

\textbf{Proof of {\boldmath $(i)$}:}
Let $v \in \probSpace_2$ such that $\phi(v)\le 1$.
 In order to prove the first statement, we want to lower-bound the expected number of neighbors of $v$ in the set $V^+(v,\varepsilon)$. In particular, such neighbors have objective at least $\phi(v)^{1/\gamma(\varepsilon)}$. Let $\delta := (\wmin^{-1}n^{-1}\phi(v)^{-1-1/\gamma(\varepsilon)})^{1/d}$ and let $c' = \min\{1,c_1/2^d\}$, where $c_1$ is the constant given by $(EP1)$ and $(EP2)$. We consider the set 
$$A(v,\varepsilon) :=\left\{u \in \probSpace \mid (1) :  \|\x u - \x t\|\le \delta ; (2): \w u \ge (c'\phi(v))^{-1} \ge \phi(v)^{-1}\right\}.$$
By the definition of $\delta$ and Condition~(2), a vertex $u \in A(v,\varepsilon)$ satisfies 
$$\phi(u) \ge\frac{\w u }{\delta^d \wmin n} \ge \w u\phi(v)^{1+1/\gamma(\varepsilon)}\ge \phi(v)^{1/\gamma(\varepsilon)}.$$ 
On the other hand, by $(2)$ and the assumption $\phi(v) \le 1$ it holds $\w u \ge 1$, and then we have
$\w u\phi(v)^{1+1/\gamma(\varepsilon)} \ge \w u^{1/\gamma(\varepsilon)} \ge \w u^{-\gamma(\varepsilon_1)}$. Thus, $A(v,\varepsilon) \subseteq V^+(v,\varepsilon)$.

Our next goal is to upper-bound the geometric distance between vertices of $A(v,\varepsilon)$ and the vertex $v$ itself. We observe that the assumptions $v \in \probSpace_2$ and $\varepsilon\le\varepsilon_1$ imply 
$$\phi(v)^{-1/\gamma(\varepsilon)} \le \w v^{\gamma(\varepsilon_1)/\gamma(\varepsilon)} \le \w v,$$ and it follows
$$n \wmin \delta^d = \phi(v)^{-1-1/\gamma(\varepsilon)}  \le \phi(v)^{-1} \w v  =  n \wmin \|\x v - \x t\|^d.$$
By the triangle inequality we obtain $\|\x u - \x v\| \le 2 \|\x v - \x t\|$ for all $u \in A(v,\varepsilon)$. 

Next we want to verify that all vertices in $A(v,\varepsilon)$ are connected to $v$ with constant probability. We notice that for all $u \in A(v,\varepsilon)$, by Condition $(2)$ we have 
$$\frac{c_1\w u \w v}{\wmin n} =c_1\w u \phi(v) \|\x v - \x t\|^d \ge \frac{c_1}{c'} \|\x v - \x t\|^d \ge \frac{c_1}{2^{d}c'}\|\x u - \x v\|^d\ge \|\x u - \x v\|^d.$$ 
by our choice of $c'$. Then, indeed by (\ref{eq:puv}) and (\ref{eq:puv2}), all vertices of $A(v,\varepsilon)$ are connected to $v$ with constant probability. The expected number of vertices with weight at least $(c' \phi(v))^{-1}$ is $\Omega(n (\wmin \phi(v))^{\beta-1})$, and thus
$$
\Ex[|\Gamma(v) \cap A(v,\varepsilon)|]= \Omega\left(n (\wmin \phi(v))^{\beta-1}\delta^d\right) = \Omega\left( \wmin^{\beta-2}\phi(v)^{\beta-2-1/\gamma(\varepsilon)}\right) = \Omega\left( \wmin^{\beta-2}\phi(v)^{-\Omega(\varepsilon)}\right),
$$
as $\varepsilon$ is chosen small enough.
\medskip

\textbf{Proof of {\boldmath $(ii)$}:} Let $v \in \probSpace_2$ such that $\phi(v) \le \phi_1$ and let $u \in V^-(v,\varepsilon)$. We first prove that $u$ has the property $\|\x u - \x v\| \ge 0.5 \|\x v - \x t\|$. Since $u \in V^-(v,\varepsilon)$, we have $\phi(u)\ge\phi(v)^{1/\gamma(\varepsilon)}$, and by the assumption $\phi(v)\le \phi_1$ for $\varepsilon_1$ small enough this is at least $\phi(v)^{1-\varepsilon} \ge 2^d \phi(v)$, where we used $\phi(v) \le \phi_1$ in the last step. Because $u \in \probSpace_1$, it follows
$$\|\x u - \x t\|^d = \phi(u)^{-1}\frac{\w u }{\wmin n} \le  \frac{\phi(u)^{-1-1/\gamma(\varepsilon_1)}}{\wmin n}\le \frac{(2^d \phi(v))^{-1-1/\gamma(\varepsilon_1)}}{\wmin n}.
$$
On the other hand, $v \in \probSpace_2$, and then
$$\|\x u - \x t\|^d = \frac{(2^d \phi(v))^{-1-1/\gamma(\varepsilon_1)}}{\wmin n} \le \frac{0.5^d  \phi(v) \w v}{\wmin n} = 0.5^d\|\x v - \x t\|^d.
$$
By the triangle inequality $\|\x u - \x v\| \ge 0.5 \|\x v - \x t\|$ holds as desired. Moreover, $u \in \probSpace_1$ and thus $\w u \le \phi(u)^{-1/\gamma(\varepsilon_1)}$. Since $\phi(u)\ge\phi(v)$, for $\varepsilon_1$ small enough it follows 
$$\w u \phi(v) \le \phi(u)^{-1/\gamma(\varepsilon_1)}\phi(v) \le \phi(v)^{1-1/\gamma(\varepsilon)}\le \phi(v)^{\varepsilon}.$$ 

\textbf{Case {\boldmath$\alpha=\infty$}:} We obtain
$$\frac{\w u \w v}{\|\x u - \x v\|^d} \le \frac{\w u \w v}{0.5^d\|\x v - \x t\|^d}=2^d\w u \phi(v) \wmin n \le 2^d \phi(v)^{\Omega(\varepsilon)}\wmin n.$$
Together with the assumption $\phi(v) \le \phi_1$ this implies $\frac{c_2\w u \w v}{\|\x u - \x v\|^d \wmin n} \le c_2 2^d \phi_1^{\Omega(\varepsilon)}$, and by taking $\phi_1$ small enough, this expression is smaller than $1$. Then (\ref{eq:puv2}) yields $u \notin \Gamma(v)$ deterministically.
\medskip

\textbf{Case {\boldmath$\alpha<\infty$}:} We have 
$$p_{uv}=O\left(\left(\frac{\w u \w v}{\|\x u - \x v\|^d \wmin n}\right)^{\alpha}\right)=O\left(\left(\frac{\w u \w v}{\|\x v - \x t\|^d \wmin n}\right)^{\alpha}\right)=O\left((\w u \phi(v))^{\alpha}\right).$$ 
Recall that we observed already above that $\w u \le \phi(v)^{-1/\gamma(\varepsilon)}$, and note that every vertex $u \in V^-(v,\varepsilon)$ has the geometric property $\|\x u - \x t\|^d \le \phi(v)^{-1/\gamma(\varepsilon)} \frac{\w u }{\wmin n}$, which a random vertex of given weight $w$ satisfies with probability at most $\phi(v)^{-1/\gamma(\varepsilon)} \frac{w }{\wmin n}$. We conclude that
\begin{align*}
\Ex[|\Gamma(v) \cap V^-(v)|] & = O\left(\int_{\wmin}^{\phi(v)^{-1/\gamma(\varepsilon)}}n \wmin^{\beta-1}w^{-\beta} \phi(v)^{-1/\gamma(\varepsilon)} \frac{w}{\wmin n} (w \phi(v))^{\alpha}dw\right)\\
& = O\left(\wmin^{\beta-2} \phi(v)^{-1/\gamma(\varepsilon)+\alpha} \int_{\wmin}^{\phi(v)^{-1/\gamma(\varepsilon)}} w^{\alpha+1-\beta}dw\right)\\
& = O\left(\wmin^{\beta-2} \phi(v)^{\frac{\tau}{1-\varepsilon}}\right),
\end{align*}
where $\tau=(\alpha+2-\beta)(1-1/\gamma(\varepsilon))-\alpha\varepsilon=\Omega(\varepsilon)$ for $\varepsilon\le\varepsilon_1$ and $\varepsilon_1$ small enough. This proves $(ii)$.

\end{proof}

In order to prove our main results, we want to show that with sufficiently high probability, greedy routing follows our anticipated trajectory. If $v \in \probSpace_1$ is a vertex on the greedy path, we want that the best neighbor of $v$ is in $V^+(v,\varepsilon)$, i.e., that its weight is significantly larger than $\w v$. By Lemma~\ref{lem:firstphase}~(ii) the expected number vertices in $\Gamma(v)$ with small weight and relatively large objective is polynomially small in $\w v$. By Markov's inequality, this will give us an error probability which is polynomially small in $\w v$. If $v$ has constant weight, this is not precise enough as for the proof of Theorem~\ref{thm:greedysuccess2}, we need an error probability which is \emph{exponentially} small in $\wmin$. With our next lemma, we improve this as follows: We consider vertices $v\in\probSpace_1$ whose objective is relatively small compared to the weight $\w v$, and calculate the expected size of $V^-(v,\varepsilon)$ depending on $\phi(v)$. Later on, this statement can be applied in particular to vertices of constant weight with large geometric distance to the target.
%
%\begin{equation}\label{eq:starcond} \phi(v) \w v^{\gamma_2} 2^d \quad\le  \quad\begin{cases} (\w v e^{\Theta(\w v^{\varepsilon})})^{-1/(\alpha-1)}, & \mbox{$\alpha<\infty$}\\
%  \Theta(1), & \mbox{$\alpha=\infty$}
%  \end{cases} \tag{$\star$}
%\end{equation}
\begin{lemma} \label{lem:smweights} There exists a constant $\bar{c}>0$ such that for all $0<\varepsilon\le\varepsilon_1$ and all vertices $v \in \probSpace_1$ with $\w v \ge w_1(\varepsilon)$ and $\phi(v) \w v^{\gamma(\zeta \cdot \varepsilon)} \le \bar{c}$, it holds 
$$\Ex[|\Gamma(v) \cap V^-(v,\varepsilon)|] = O(\phi(v)^{\Omega(1)}\w v^{O(1)}).$$
\end{lemma}
\begin{proof}Let $0<\varepsilon<\varepsilon_1$ and let $v \in \probSpace$ be a vertex such that $\phi(v) \w v^{\gamma(\zeta\cdot\varepsilon)}  \le\bar{c}:= 1/(2^d c_2)$, where $c_2$ is the constant given by (EP2). We want to show that with relatively high probability, every neighbor of $v$ with small weight is located close to $v$ on the torus. For all weights $w \le \w v^{\gamma(\zeta\cdot\varepsilon)}$ we define a critical radius 
$$r_0(w):= 0.5\left(\frac{\w v^{1-\gamma(\zeta\cdot\varepsilon)} w}{\wmin n \phi(v)}\right)^{1/d}.$$ 
Next we define
$$ B(v,\varepsilon) := \{u \in \probSpace \mid (1): \w u \le \w v^{\gamma(\zeta\cdot\varepsilon)}; (2): \|\x v - \x u\|\ge r_0(\w u)\}.$$
We first show that the expected size of the set $\Gamma(v) \cap B(v,\varepsilon)$ is very small. Afterwards, we will prove that $V^-(v,\varepsilon) \subseteq B(v,\varepsilon)$ holds under the assumptions of the lemma.

With the assumption $\phi(v) \le\bar{c}\w v^{-\gamma(\zeta\cdot\varepsilon)}$ we deduce that $r_0(w)^d \ge   \frac{w \w v}{2^d \bar{c}\wmin n} = \frac{c_2w \w v}{\wmin n}$ holds for all weights $w \le \w v^{\gamma(\zeta\cdot\varepsilon)}$.
\medskip

\textbf{Case {\boldmath$\alpha=\infty$}:} By $(\ref{eq:puv2})$ the vertex $v$ has deterministically no neighbors in $B(v,\varepsilon)$.
\medskip

\textbf{Case {\boldmath$\alpha<\infty$}:} By $(\ref{eq:puv})$ a vertex $u \in B(v)$ with distance $r$ to $v$ is connected to $v$ with probability $O((\frac{\w v \w u}{r^d \wmin n})^{\alpha})$. We calculate $\Ex[|\Gamma(v) \cap B(v,\varepsilon)|]$ by integrating over all weights smaller than $\w v^{\gamma(\zeta\cdot\varepsilon)}$ and all distances $r \ge r_0(w)$. Using the density function $f(w)=\Theta(\wmin^{\beta-1}w^{-\beta})$ we deduce
\begin{align*}
\Ex[|\Gamma(v) \cap B(v,\varepsilon)|] & = O\left( \int_{\wmin}^{\w v^{\gamma(\zeta\cdot\varepsilon)}} n \wmin^{\beta-1} w^{-\beta} \int_{r_0(w)}^{1/2}\left(\frac{\w v w}{r^d \wmin n}\right)^{\alpha}  r^{d-1}dr dw\right)\\
& =  O\left(n \wmin^{\beta-1} \int_{\wmin}^{\w v^{\gamma(\zeta\cdot\varepsilon)}} \left(\frac{\w v w}{\wmin n}\right)^{\alpha}r_0(w)^{d(1-\alpha)}w^{-\beta}dw\right)\\ 
%& =  O\left(n \wmin^{\beta-1} \int_{\wmin}^{\w v^{\gamma(\zeta\cdot\varepsilon)}} \frac{\w v w}{\wmin n} \w v^{\gamma_2(\alpha-1)}\phi(v)^{\alpha-1} w^{-\beta}dw\right)\\
& =  O\left(n \wmin^{\beta-1} \int_{\wmin}^{\w v^{\gamma(\zeta\cdot\varepsilon)}} \frac{\w v w}{\wmin n} \w v^{\gamma(\zeta\cdot\varepsilon)(\alpha-1)}\phi(v)^{\alpha-1} w^{-\beta}dw\right)\\
& = O\left(\wmin^{\beta-2}\w v^{\gamma(\zeta\cdot\varepsilon)(\alpha-1)+1}\phi(v)^{\alpha-1} \int_{\wmin}^{\w v^{\gamma(\zeta\cdot\varepsilon)}} w^{1-\beta}dw\right)\\
& = O\left(\w v^{\gamma(\zeta\cdot\varepsilon)(\alpha-1)+1}\phi(v)^{\alpha-1}\right).
\end{align*}

We finish the proof of the lemma by showing that  $V^-(v,\varepsilon) \subseteq B(v,\varepsilon)$ holds under the assumption that $v$ satisfies in addition $\w v \ge w_1(\varepsilon)$.  By assumption, $v \in \probSpace_1$ and all $u \in V^-(v,\varepsilon)$ fufilll $\w u \le \w v^{\gamma(\zeta\cdot\varepsilon)}$. Next we observe that by the definition of the set $V^-(v,\varepsilon)$ for a vertex $v \in \probSpace_1$, all $u \in V^-(v,\varepsilon)$ satisfy
$$\|\x u - \x t\|^d = \frac{\w u}{n\wmin\phi(u)} \le \frac{\w v^{\gamma(\zeta\cdot\varepsilon)}}{n\wmin \phi(v)\w v^{\gamma(\varepsilon)-1}}  = \|\x v - \x t\|^d \w v^{\gamma(\zeta\cdot\varepsilon)-\gamma(\varepsilon)}= \|\x v - \x t\|^d \w v^{-\Omega(\varepsilon)}.$$
Hence $\|\x u - \x t\| \le \|\x v - \x t\| \w v^{-\Omega(\varepsilon)/d} \le 0.5 \|\x v - \x t\|$ as we assume $\w v \ge w_1(\varepsilon)$. By the triangle inequality we deduce 
$\|\x v - \x u\| \ge 0.5\|\x v - \x t\|$. On the other side, since $\w u \le \w v^{\gamma(\zeta\cdot\varepsilon)}$, the definition of $r_0(w)$ implies
$$r_0(\w u)^d = 0.5^d \frac{\w v^{1-\gamma(\zeta\cdot\varepsilon)} \w u}{\wmin n \phi(v)} \le 0.5^d \|\x v- \x t\|^d.$$ 
It follows 
$\|\x v - \x u\| \ge 0.5\|\x v - \x t\| \ge r_0(\w u)$,
which proves $V^-(v,\varepsilon) \subseteq B(v,\varepsilon)$.
\end{proof}

We turn to the last lemma of this chapter. In order to prove Theorem~\ref{thm:length} and Theorem~\ref{thm:patching}, we will need to control the trajectory in particular at the start and at the end of the routing. For technical reasons, we require that the routing will not see any edges which increase the objective too fast. Statement $(i)$ of this lemma states that a.a.s.\ there are no vertices of both high weight and high objective at all. Statements~$(ii)$ and $(iii)$ upper-bound the probability that a vertex of small weight has an incident edge of unexpected length, i.e., a neighbor which is located geometrically far away. Finally $(iv)$ upper-bounds this probability for a vertex of high weight.

\begin{lemma}[Unlikely jumps] \label{lem:nojumps}
Let $w_0=w_0(n)=\omega(1)$ be a growing function in $n$ and $v \in \probSpace$. Then the following statements hold.
\begin{enumerate}[(i)]
\item With prob. $1-w_0^{-\Omega(1)}$, there exists no vertex $u$ (except potentially $t$) s.t. $\w u \ge w_0$ and $\phi(u) \ge w_0^{-\Omega(1)}$.
\end{enumerate}
Let $\w v \le w_0$:
\begin{enumerate}[(i)]
\setcounter{enumi}{1}
\item With prob. $1-O(w_0^{-5})$ there exists no neighbor $u$ of $v$ with $\w u \le w_0$ and $\|\x u - \x v\|^d \ge \frac{w_0^{O(1)}}{n}$.

\item Let $\phi(v) \le w_0^{-\Omega(1)}$. With probability $1-O(w_0^{-5})$ there exists no neighbor $u$ of $v$ with $\w u \ge w_0$ and
$\phi(u) \le \phi_0^{O(1)}$.
\end{enumerate}
Let $\w v \ge w_0$:
\begin{enumerate}[(i)]
\setcounter{enumi}{3}
\item Let  $\phi(v) \le w_0^{-1-\Omega(1)}$, and let $M := \min\{\w v,\phi(v)^{-1}\}$. \\
With prob. $1-O(M^{-\Omega(1)})$ there exists no neighbor $u$ of $v$ with $\w u \le w_0$ and $\phi(u) \ge w_0^{-\Omega(1)}$.
\end{enumerate}
\end{lemma}
\begin{proof}
Let $c>0$ be a constant chosen sufficiently large, $\varepsilon$ be a constant chosen sufficiently small, let $r_0 := (w_0^c/n)^{1/d}$ and  for every weight $w \ge \wmin$ define $\delta(w) := (\frac{w w_0^{\varepsilon}}{\wmin n})^{1/d}$.
\medskip

\textbf{Proof of {\boldmath$(i)$}:}
 We first observe that a vertex $u$ of given weight $\w u$ satisfies $\phi(u) \ge w_0^{-\varepsilon}$ if and only if $\|\x u - \x t\|^d \le \delta(\w u)^d$. Then excluding $t$ which is potentially fixed, the expected number of vertices with weight at least $w_0$ and objective at least $w_0^{-\varepsilon}$ is at most
$$
\int_{w_0}^{\infty} n \wmin^{\beta-1} w^{-\beta}\delta(w)^d dw = \int_{w_0}^{\infty} \wmin^{\beta-2} w_0^{\varepsilon} w^{1-\beta} dw = O\left(\wmin^{\beta-2}w_0^{2-\beta+\varepsilon} \right) = w_0^{-\Omega(1)},$$
where we used that for $\varepsilon$ sufficiently small it holds $\beta-2 \ge 2\varepsilon$. By Markov's inequality with probability $1-w_0^{-\Omega(1)}$ there exist no such vertices of weight at least $w_0$.
\medskip

\textbf{Proof of {\boldmath$(ii)$}:} Let $v$ be a vertex of weight at most $w_0$ and define
\[A_0(v):=\{u\in \probSpace \mid \w u \le w_0\text{ and }\|\x u - \x v\|^d \ge r_0\}.\]
For $c$ sufficiently large and for any vertex $u \in A_0(v)$ we have $\frac{\w u \w v}{\wmin n} \le \frac{w_0^2}{\wmin n} = o(r_0^d)$.
\medskip

\textbf{Case {\boldmath$\alpha=\infty$}:} The claim follows as by $(EP2)$ deterministically $u$ and $v$ do not share an edge, i.e., $\Gamma(v) \cap A_0(v) = \emptyset$. 
\medskip

\textbf{Case {\boldmath$\alpha < \infty$}:} By $(EP1)$ and the assumption $\w v \le w_0$, the probability $p_{uv}$ is at most $O\left((w_0^{2}/(n\|\x u - \x v\|^d))^{\alpha}\right)$. Thus we have
$$\Ex[|\Gamma(v) \cap A_0(v)|] =  O\left( n \int_{r_0}^{1/2} r^{d-1} \left(\frac{w_0^2}{n r^d}\right)^{\alpha}dr\right) = O\left(w_0^{2\alpha}n^{1-\alpha}r_0^{d(1-\alpha)}\right)=O\left(w_0^{2\alpha+c(1-\alpha)}\right).$$
If we take $c \ge \frac{5+2\alpha}{\alpha-1}$, then we obtain 
$\Ex[|\Gamma(v) \cap A_0(v)|] = O(w_0^{-5})$, and by Markov's inequality with probability at least $1-O(w_0^{-5})$, the set $\Gamma(v) \cap A_0(v)$ is empty.
\medskip

\textbf{Proof of {\boldmath$(iii)$}:} Let $v$ be a vertex of weight at most $w_0$ which satisfies $\phi(v) \le w_0^{-\varepsilon}$, and define
$$A_1(v) := \{u \in \probSpace \mid (1): \w u \ge w_0; (2): \phi(u) \le \phi(v)^{c}\}.$$
We want to show that a.a.s.\ $v$ has no neighbor in the set $A_1(v)$. Define $\delta'(w) := (\frac{w}{\phi(v)^c\wmin n})^{1/d}$ \label{defn-4nt:delta'(w)} and observe that a vertex $u\in\probSpace$ with $\w u\ge w_0$ is contained in $A_1(v)$ if and only if $\delta'(w_u) \le \|\x u - \x t\|$. However, due to $\w v \le w_0$, for $c>1/\varepsilon$ we have $\delta'(\w u)^d =  \omega(\frac{\w u \w v}{\wmin n})$.
\medskip

\textbf{Case {\boldmath$\alpha=\infty$}:}  The claim follows because $(EP2)$ implies that, deterministically, $v$ does not have neighbors in $A_1(v)$.
\medskip

\textbf{Case {\boldmath$\alpha < \infty$}:} Note that  for $w \ge w_0$ we have $\delta'(w) = \omega(\|\x v - \x t\|)$, and by the triangle inequality we deduce that for a vertex $u \in A_1(v)$ it holds $\|\x u - \x v\| \ge 0.5\|\x u - \x t\|$. After this preparation, we can calculate the expected number of vertices in $\Gamma(v) \cap A_1(v)$ by integrating over the complete set $A_1(v)$. We obtain
\begin{align*}
\Ex[|\Gamma(v) \cap A_1(v)|] & = O\left( \int_{w_0}^{\infty} n \wmin^{\beta-1} w^{-\beta} \int_{\delta'(w)}^{1/2} r^{d-1} \left(\frac{w \w v}{(0.5r)^d \wmin n}\right)^{\alpha}dr dw\right)\\
& = O\left( \int_{w_0}^{\infty} n \wmin^{\beta-1} w^{-\beta}  \left(\frac{w \w v}{\wmin n}\right)^{\alpha} \delta'(w)^{d(1-\alpha)} dw\right)\\
& = O\left( \wmin^{\beta-2} \w v^{\alpha} \phi(v)^{c(\alpha-1)} \int_{w_0}^{\infty} w^{1-\beta} dw \right)\\
& = O\left( \wmin^{\beta-2} \w v^{\alpha} \phi(v)^{c(\alpha-1)}w_0^{2-\beta} \right) = O\left( \w v^{\alpha} w_0^{\varepsilon c(1-\alpha)} \right).
\end{align*}

Since $\w v \le w_0$ by assumption, we can choose $c$ large enough such that this value is at most $O(w_0^{-5})$. Then Markov's inequality implies that indeed with probability $1-O(w_0^{-5})$, the set $\Gamma(v) \cap A_1(v)$ is empty. This proves Statement~$(iii)$.
\medskip

\textbf{Proof of {\boldmath$(iv)$}:} Let $v$ be a vertex with $\w v\ge w_0$ and objective $\phi(v)\le w_0^{-1-\varepsilon}$, and define $M := \min\{\w v,\phi(v)^{-1}\}$. We want to upper-bound the probability that $v$ has a neighbor in 
$$A_2(v) := \{u \in \probSpace \mid (1): \w u \le w_0; (2): \phi(u) \ge w_0^{-\varepsilon}\}.$$
We know that every vertex $u \in A_2(v)$ satisfies 
$$\|\x u - \x t\|^d \le \delta(\w u)^d=\frac{\w u w_0^{\varepsilon}}{\wmin n}=o\left(\frac{\w u \w v}{\wmin n}\right).$$ 
On the other hand, the vertex $v$ satisfies
$$\|\x v - \x t\|^d = \frac{\w v}{\phi(v)\wmin n} \ge \frac{\w v w_0^{1+\varepsilon}}{\wmin n} = \omega\left(\frac{\w u \w v}{\wmin n}\right).$$
By the triangle inequality it follows $\|\x v - \x u\|=\omega\left(\frac{\w u \w v}{\wmin n}\right)$.
\medskip

\textbf{Case {\boldmath$\alpha=\infty$}:} The claim follows by $(EP2)$ because there exists deterministically no edge between $u$ and $v$. 
\medskip

\textbf{Case {\boldmath$\alpha < \infty$}:} We deduce $p_{uv} = O((\phi(v) \w u)^{\alpha})$. By integrating over all weights \mbox{$w \in [\wmin,w_0]$} we obtain
\begin{align*}
\Ex[|\Gamma(v) \cap A_2(v)|] & = O\left(\int_{\wmin}^{w_0}n\wmin^{\beta-1}w^{-\beta} \delta(w)^d (\phi(v)w)^{\alpha}dw\right)\\ & = O\left(\wmin^{\beta-2} w_0^{\varepsilon} (\phi(v)w_0)^{\alpha}\int_{\wmin}^{w_0}w^{1-\beta}dw\right) = O\left(w_0^{\varepsilon}(\phi(v)w_0)^{\alpha} \right).
\end{align*}
We continue with a case distinction. First suppose $\phi(v) \le \w v^{-1-2\varepsilon}$. Since $\alpha>1$, $\phi(v) \le w_0^{-1}$, and $\w v \ge w_0$, in this case the above expression is at most
$$O\left(w_0^{1+\varepsilon}\phi(v)\right)=O(w_0^{1+\varepsilon}\w v^{-1-2\varepsilon})=O(\w v^{-\varepsilon})=O(M^{-\Omega(1)}).$$
On the other hand, if $\phi(v) \ge \w v^{-1-2\varepsilon}$ we use 
$$\phi(v) \le w_0^{-1-\varepsilon} \Leftrightarrow w_0 \le \phi(v)^{-1/(1+\varepsilon)}$$ 
and get
$$\Ex[|\Gamma(v) \cap A_2(v)|] = O\left(w_0^{\varepsilon+\alpha}\phi(v)^{\alpha}\right)= O\left(\phi(v)^{\alpha - \frac{\alpha+\varepsilon}{1+\varepsilon}}\right) = O\left(\phi(v)^{\frac{\alpha\varepsilon-\varepsilon}{1+\varepsilon}}\right)=O\left(M^{-\Omega(1)}\right).$$
Then again by Markov's inequality, with probability $1-O(M^{-\Omega(1)})$ the set $\Gamma(v) \cap A_2(v)$ is empty, which finishes the proof.
\end{proof}

\section{Proof: Success Probability and Length of Greedy Routing}\label{sec:analysis}

In this section we analyze the basic greedy process and prove Theorem~\ref{thm:greedysuccess1} and Theorem~\ref{thm:greedysuccess2} about the success probability and Theorem~\ref{thm:length} about the running time. We advise the reader to first read Section~\ref{sec:proofsketch} where we describe the high-level ideas that we use to prove our main Lemma~\ref{lem:main}. In a nutshell, we distinguish two phases of the process (governed by weights and by objectives, respectively), and for each of the phases we divide the space $\probSpace$ into layers. Then we prove that with sufficiently large probability a greedy path never visits more than one vertex per layer, and a.a.s.\! it does not fail in any layer except for the very first or the very last ones. This observation is formalized in Lemma~\ref{lem:main} and lies at the heart of all our main results, including the results for patching and for the relaxations. All algorithms follow very similar trajectories, only the starting phase and the end phase differs for the various algorithms. The layer structure is depicted in Figure~\ref{fig:greedypath} on page~\pageref{fig:greedypath}.

%We will start by describing the greedy routing process and our proof ideas on a high level. Then we will repeat the notations from Section~\ref{subsec:basiclemmas} and formulate a main lemma which studies the main part of the routing process. Finally we will give formal proofs of the three theorems. 

%\subsection{High-level description of the routing process} \label{subsec:intuition}

\subsection{Main Lemma} \label{subsec:mainlemma}

Let us start by repeating the notations which we introduced in Section~\ref{subsec:basiclemmas}. For all $\varepsilon>0$, we denote $\gamma(\varepsilon) := \frac{1-\varepsilon}{\beta-2}$. Furthermore $\varepsilon_1 = \varepsilon_1(\alpha,\beta)>0$ is a fixed constant, chosen sufficiently small, $\zeta=\zeta(\alpha,\beta)$ is also a fixed constant, $w_1(\varepsilon)=O(e^{d/\varepsilon})$, and $\phi_1(\varepsilon) = \Omega(e^{-d/\varepsilon})$. Then we have the following two classes of vertices:
$$\probSpace_1 := \{v \in \probSpace: \phi(v) \le \w v^{-\gamma(\varepsilon_1)}\}
\quad 
\text{and}
\quad \probSpace_2 := \{v \in \probSpace: \phi(v) \ge \w v^{-\gamma(\varepsilon_1)}\}.$$
Finally, for a given vertex $v \in \probSpace$ and a given $\varepsilon>0$ we defined good sets $V^+(v,\varepsilon)$ and bad sets $V^-(v,\varepsilon)$ above in (\ref{eq:areas1}) and (\ref{eq:areas2}). 

In Section~\ref{sec:proofsketch} we claimed that the expected trajectory of the routing process is such that soon after the start the current weight is $\omega(1)$. We also expect that, at the end of the routing, when reaching rather large objectives, very few additional hops are needed to reach $t$. Therefore the main part of the process plays in-between. In this section we state and prove a lemma which describes the routing process in this main part. For any pair $(w,\phi)$ of a weight and an objective we define
$$\overline{\probSpace}(w,\phi) := \{v \in \probSpace_1 \mid (1):\w v \ge w; (2): \phi(v) \le \phi\} \cup \{v \in \probSpace_2 \mid \phi(v) \le \phi\}.$$
Suppose that both $w$ and $\phi$ are constants. Then the main part of the routing process will happen in this set $\overline{\probSpace}(w,\phi)$. With the following lemma, we study the routing process inside such a set. We show that with sufficiently high probability, the routing will not die out in $\overline{\probSpace}(w,\phi)$ and will be ultra-fast in traversing this set. Thereby the probabilities for the failure events will depend on $w$ and $\phi$, which can be both growing functions as well. We will apply the lemma for both the basic greedy algorithm and the patching algorithms.

\begin{lemma}[Main lemma]
\label{lem:main} 
Let $w_0 \ge w_1(\varepsilon_1)$, $\phi_0 \le \phi_1(\varepsilon_1)$, and $M := \min\{w_0,\phi_0^{-1}\}$. Furthermore let \mbox{$f_0(n)=\omega(1)$} be any growing function such that $f_0(n)=o(\log\log n)$, and let $A$ be a routing protocol which satisfies $(P1)$, i.e., it makes greedy choices. Suppose $\phi(s) \le \phi_0$, denote by $P$ the greedy path induced by all vertices of objective at most $\phi_0$ starting from $s$, and suppose that there exists a first vertex $u_1 \in P \cap \overline{\probSpace}(w_0,\phi_0)$. Then with probability 
$$1-O\left(\wmin^{\beta-2} M^{-\Omega(1)}\right),$$ 
there exists a subpath $P' = (u_1,\ldots,u_\ell)$ of $P$ starting in $u_1$ with the following properties:
\begin{enumerate}[(i)]
\item $P'$ is contained in $\overline{\probSpace}(w_0,\phi_0)$.
\item Either $P' \subset \probSpace_1$, or $P' \subset \probSpace_2$, or there exists a vertex $u_{\ell'}$ such that $\{u_1, \ldots, u_{\ell'}\} \in \probSpace_1$ and $\{u_{\ell'+1}, \ldots, u_{\ell}\} \in \probSpace_2$.
\item Let $\{u_i,u_{i+1},u_{i+2}\}$ be three subsequent vertices on $P' \cap \probSpace_1$. Then $\w {u_{i+2}} \ge \w {u_i}^{\gamma(\zeta \varepsilon_1)}$.
\item Let $\{u_i,u_{i+1},u_{i+2}\}$ be three subsequent vertices on $P' \cap \probSpace_2$. Then $\phi(u_{i+2}) \ge \phi(u_i)^{1/\gamma( \varepsilon_1)}$.
\item The length $\ell$ of the path $P'$ is at most $\frac{2+o(1)}{|\log(\beta-2)|}\log\log n$. 

Moreover and more precisely, the length of $P'$ is upper bounded by
$$
\frac{1+o(1)}{|\log(\beta-2)|}\left(\log\log_{w_0}(\phi(u_1)^{-1})+\log\log_{\phi_0^{-1}}(\phi(u_1)^{-1})\right)+O\left(f_0(n)\right).$$ 
\item The following holds 
$$\Ex_{> \phi_0}[|\Gamma(u_\ell) \cap V^+(u_\ell,\varepsilon_1) \cap V_{> \phi_0}|]= \Omega\left(\wmin^{\beta-2}M^{\Omega(1)}\right).$$ 
%If $u_\ell \in \probSpace_2$ then $\Ex_{> \phi_0}[|\Gamma(u_\ell) \cap V^+(u_\ell,\varepsilon_1) \cap V_{> \phi_0}|] = \Omega\left(\wmin^{\beta-2}\phi_0^{-\Omega(1)}\right)$.\\ 
Here $\Ex_{> \phi_0}$ means expectation w.r.t.\! uncovering the vertices of objectives larger than $\phi_0$, conditioned on position and weight of $u_\ell$.
%The vertex $u_{\ell}$ has a neighbor in $V^+(u_{\ell},\varepsilon_1)$ of objective larger than $\phi_0$.
\end{enumerate}
\end{lemma}
%\ym{Do we need the left hand side of the statement in $(v)$? Otherwise I would remove it}\jl{I would rather keep it. We need it if we want to say how many steps we need if $s$ and $t$ do not have constant weights and distance. (As with all such comments, please remove after reading.)}
Note that Lemma~\ref{lem:main} naturally applies to the situation that we have uncovered all vertices of objective at most $\phi_0$, but not the rest of the graph. More precisely, conditions $(i)-(v)$ are independent of $V_{> \phi_0}$, and $(vi)$ makes a statement about the marginal expectation \emph{after} uncovering $V_{\leq \phi_0}$ and \emph{before} uncovering $V_{> \phi_0}$.
\begin{proof}
Let $P$ be the greedy path induced by all vertices of objective at most $\phi_0$ We analyze the structure of $P$ by partitioning the set $\overline{\probSpace}(w_0,\phi_0)$ into several small layers. Then the idea will be to show that with sufficiently high probability, there exists a vertex $u_{\ell}$ on the path $P$ such that until reaching $u_{\ell}$, $P$ visits every layer at most once, and the vertex $u_{\ell}$ has neighbors of objective at least $\phi_0$. There will be two classes of layers: The layers $A_{1,j}$ divide the area $\probSpace_1 \cap \overline{\probSpace}(w_0,\phi_0)$ and are defined via weights, and the layers $A_{2,j}$ divide the area $\probSpace_2 \cap \overline{\probSpace}(w_0,\phi_0)$ and are defined via objectives. 

We may assume $f_0(n)=o(\log\log n)$. Recall that we already put $\varepsilon_1=\Theta(1)$. In addition, we put $\varepsilon_2 := (\log\log f_0(n))^{-1}=o(1)$,\label{lem:main-varepsilon_2-defn} and we may assume $\eps_2 < \eps_1$. Furthermore, we define the landmarks
$$w_0' := w_0^{(\gamma(\zeta\varepsilon_1)^{f_0(n)})}\quad\text{and}\quad\phi_0' := \phi_0^{(\gamma(\varepsilon_1)^{f_0(n)})}.$$
These landmarks will be thresholds for slightly different definitions of the layers.
For later reference we note that, since $\gamma(\zeta \varepsilon_1)>1$,
\begin{equation}\label{eq:w0epsi}
w_0'^{\Omega(\varepsilon_2)}=w_0^{\Omega(\gamma(\zeta\varepsilon_1)^{f_0(n)} /\log\log f_0(n))}=\omega\left(w_0^{\Omega(1)}\right)\quad \text{ and } \quad\phi_0'^{\Omega(\varepsilon_2)} = o\left(\phi_0^{\Omega(1)}\right),
\end{equation}
where the second equation follows analogously to the first. 

\para{Layers by weight:} We define a sequence $y_0:=w_0 < y_1 < \ldots < y_j < \ldots$ for the weights such that this sequence grows doubly exponentially. More precisely, we require\label{def:yj}
$$y_{j+1} = \begin{cases} y_j^{\gamma(\zeta\varepsilon_1)} & \mbox{if } y_j < w_0',\\
y_j^{\gamma(\varepsilon_2)} & \mbox{if } y_j \ge w_0'.\end{cases}$$ 
Then we define the layers by weight which partition the set \label{def:Aij}
$\probSpace_1' := \{v \in\overline{\probSpace}(w_0,\phi_0) \mid \phi(v)\w v^{\gamma(\varepsilon_2)} \le 1 \}$. 
$$A_{1,j} := \{v \in \probSpace_1' \mid y_{j-1} \le \w v < y_j \}~~~\forall j\ge 1.$$
Note that by definition, $\probSpace_1' \subset \probSpace_1$ and no vertex of $\probSpace_1$ has larger weight than $(\wmin n)^{(1+\gamma(\varepsilon_1))^{-1}}$. Thus we only need to consider sets $A_{1,j}$ for which $y_{j-1} \le (\wmin n)^{(1+\gamma(\varepsilon_1))^{-1}}$. Furthermore, and only for technical reasons we introduce the additional layer
$$A_{1,\infty} := \{v \in \overline{\probSpace}(w_0,\phi_0) \mid \phi(v)\w v^{\gamma(\varepsilon_1)} \le 1 \le \phi(v) \w v^{\gamma(\varepsilon_2)} \},$$
which covers the remaining vertices of $\probSpace_1 \cap \overline{\probSpace}(w_0,\phi_0)$.

\para{Layers by objective:}
Similarly, define a sequence\label{def:psi_j} $\psi_0=\phi_0 > \psi_1 > \ldots > \psi_j > \ldots$ for the objectives such that this sequence falls doubly exponentially, i.e.,
$$\psi_{j+1} = \begin{cases}
\psi_j^{\gamma(\varepsilon_1)}& \mbox{if } \psi_j > \phi_0',\\
\psi_j^{\gamma(\varepsilon_2)}& \mbox{if } \psi_j \le \phi_0'.\end{cases}$$ 
Next, we define similar layers for the vertices in $\probSpace_2 \cap \overline{\probSpace}(w_0,\phi_0)$. 
$$A_{2,j} := \{v \in \probSpace_2 \cap \overline{\probSpace}(w_0,\phi_0) \mid \psi_{j-1} \ge \phi(v) >  \psi_j \}~~~\forall j\ge 1.$$
For this second set of layers, note that every vertex in the graph has objective at least $\wmin/n$ and we do not need to consider layers $A_{2,j}$ which would contain vertices of even smaller objective. 

\para{Definition of the events {\boldmath $\mathcal{E}$} and {\boldmath$\mathcal{E}_{i,j}$}:}

The layers $A_{i,j}$ as defined above classify all vertices of objective at most $\phi_0$. We want to show that with sufficiently high probability, $P$ contains at most one vertex per layer $A_{i,j}$ until reaching a vertex $u_{\ell}$ for which the expected number of neighbors in $V_{> \phi}$ is large. We will do this by considering the layers in the following consecutive order:
$$A_{1,1} \prec \ldots \prec A_{1,j} \prec A_{1,j+1} \prec \ldots \prec A_{1,\infty} \prec \ldots \prec A_{2,j} \prec A_{2,j-1} \prec \ldots \prec A_{2,1} .$$
Given this ordering of the layers, we denote by $B_{i,j}$ the union of the layer $A_{i,j}$ and all previous layers, and by $P_{i,j}$ we denote the greedy path induced by the set $B_{i,j} \subset \probSpace$. 

In the following, we consider for a given pair $(i,j)$ the first vertex $v \in P_{i,j} \cap A_{i,j}$, if such a vertex exists. We will first show that with sufficiently high probability the neighbor of $v$ with highest objective is located outside $B_{i,j}$, which we consider a ``good'' event.  

More precisely, let $i \in \{1,2\}$, $j\ge1$, and let $\varepsilon\in\{\varepsilon_1,\varepsilon_2\}$ be the $\varepsilon$ which was used for the definition of $A_{i,j}$. Moreover, let $\eta >0$ be a sufficiently small constant, and let $E := \eta\wmin^{\beta-2}M^{\eta}$. 
% and $E_2 := \eta\wmin^{\beta-2}\phi_0^{-\eta}$, which corresponds to the right hand sides of (vi) in the cases $m = w_0$ and $m=\phi_0^{-1}$, respectively. 
Note that for any $c>0$, by choosing $\eta$ small enough we can achieve for all $v \in A_{i,j}$ that $c\wmin^{\beta-2}\w v^{\eps} > 2 E$ and $c\wmin^{\beta-2}\phi(v)^{-c\eps} > 2 E$. 
Then we denote by 
$\mathcal{E}_{i,j}$ the event that either $P_{i,j} \cap A_{i,j} = \emptyset$, or the first vertex $v \in P_{i,j} \cap A_{j,i}$ 
\begin{itemize}
\item  satisfies condition $(vi)$, or 
\item has at least one \emph{good} neighbor, i.e, a neighbor $v' \in \overline{\probSpace}(w_0,\phi_0) \setminus B_{i,j}$ with $\phi(v')\ge\phi(v)$ s.t.\\  $\phi(v') > \phi(u)$ holds for all $u \in \Gamma(v) \cap B_{i,j}$. 
\end{itemize}
Finally, we denote by $\mathcal{E} := \cap_{i,j} \mathcal{E}_{i,j}$ the intersection of all good events.

\para{Lower bound for {\boldmath$\Pr[\mathcal{E}_{i,j}]$}:}
In this part, our goal is to lower-bound $\Pr[\mathcal{E}_{i,j}]$ for every pair $(i,j)$ in order to lower-bound $\Pr[\mathcal{E}]$. Thereby, our construction of layers gives rise to several different subcases, corresponding to different subphases of the routing.

Let $(i,j)$ be a given pair for which we want to bound $\Pr[\mathcal{E}_{i,j}]$ from below. Clearly we have $\Pr[\mathcal{E}_{i,j} \mid P_{i,j} \cap A_{i,j}=\emptyset] = 1$, and we only need to consider \mbox{$\Pr[\mathcal{E}_{i,j} \mid P_{i,j} \cap A_{i,j} \neq \emptyset]$}. Let $v$ be the first vertex on $P_{i,j} \cap A_{i,j}$. Again, if $v$ satisfies condition $(vi)$ then $\mathcal{E}_{i,j}$ holds and there is nothing to show, so we assume otherwise,
 \begin{equation}\label{eq:E_icondition}
 \Ex_{> \phi_0}[|\Gamma(v) \cap V^+(v,\varepsilon_1) \cap V_{> \phi_0}|] < E.
\end{equation}
In all cases let $\varepsilon \in \{\varepsilon_1, \varepsilon_2\}$ be the same $\varepsilon$ as used for the definition of the considered layer.\medskip

\paras{Case { $i=1$} and { $j<\infty$}:} 
We are in the first phase of the routing, and the layer $A_{1,j}$ is contained in $\probSpace_1' \subset \probSpace_1$. Then by Lemma~\ref{lem:firstphase}~(i) we have that the expected number of neighbors of $v$ in $V^+(v,\varepsilon)$ is large, i.e., $\Ex[|\Gamma(v) \cap V^+(v,\varepsilon)|] = \Omega(\wmin^{\beta-2}\w v^{\eps})$. Since $\eps \leq \eps_1$ we have $V^+(v,\varepsilon) \subseteq V^+(v,\varepsilon_1)$, so in particular we do not expect many of these neighbors in $V_{> \phi_0}$, i.e., 
\begin{equation}\label{eq:choiceofEi}
\Ex_{> \phi_0}[|\Gamma(v) \cap V^+(v,\varepsilon) \cap V_{> \phi_0}|] \stackrel{\eqref{eq:E_icondition}}{<} E \leq \frac12 \Ex[|\Gamma(v) \cap V^+(v,\varepsilon)|],
\end{equation}
where the second inequality comes from our choice of the constant $\eta$ in $E$. Thus we expect most of these good neighbors in $V_{\leq \phi_0}$, i.e.,  
$$\Ex[|\Gamma(v) \cap V^+(v,\varepsilon) \cap V_{\leq \phi_0}|] = \Omega(\wmin^{\beta-2}\w v^{\eps}).$$
 By a Chernoff bound, using  $\w v \ge y_{j-1}$, with probability at least $1-\exp(-\Omega(\wmin^{\beta-2}y_{j-1}^{\varepsilon}))$ there exists a neighbor $v'$ of $v$ in $V^+(v,\varepsilon) \cap V_{\leq \phi_0}$. By the definition of $V^+(v,\varepsilon)$, such a neighbor fufillls $v' \notin B_{1,j}$ and $\phi(v') \ge \phi(v) \w v^{\gamma(\varepsilon)-1}$. 

Now, we want to show that $v'$ actually is a good neighbor of $v$, i.e., every $u \in \Gamma(v) \cap B_{i,j}$ has objective less than \mbox{$\phi(v) \w v^{\gamma(\varepsilon)-1}$}. We observe that every vertex with objective at least $\phi(v) \w v^{\gamma(\varepsilon)-1}$ and weight at most $y_j$ is contained in $V^-(v,\varepsilon)$, thus we want to upper-bound $\Ex[| \Gamma(v) \cap V^-(v,\varepsilon)| ]$. Note that here we need to condition on that no earlier vertex on the path $P_{1,j}$ had a neighbor $u \in V^-(v,\varepsilon)$ (because otherwise, $u$ would be the first vertex in $A_{i,j} \cap P_{i,j}$ and not $v$). By Lemma~\ref{lem:correlatedprobs}, this condition decreases the expected number of vertices in $V^-(v,\varepsilon)$, and therefore also the expected number of neighbors of $v$ in $V^-(v,\varepsilon)$. Hence, in order to give an upper bound on this value, we are allowed to neglect the dependencies. Since $y_{j-1} \ge w_1(\varepsilon) = O(e^{d/\varepsilon})$, by Lemma~\ref{lem:firstphase}~(ii) we deduce
$$
\Ex[|\Gamma(v) \cap V^-(v,\varepsilon)| \mid v \in A_{1,j} \cap P_{1,j}]\le\Ex[|\Gamma(v) \cap V^-(v,\varepsilon)|] = O(\wmin^{\beta-2}y_{j-1}^{-\Omega(\varepsilon)}).
$$
Then we apply Markov's inequality to see that with probability at least $1-O(\wmin^{\beta-2}y_{j-1}^{-\Omega(\varepsilon)})$, $v$ has no neighbor in $V^-(v,\varepsilon)$ and therefore every $u \in \Gamma(v,\varepsilon) \cap B_{1,j}$ has objective at most \mbox{$\phi(u) < \phi(v) \w v^{\gamma(\varepsilon)-1} \le \phi(v')$} as desired. It follows
\begin{equation}\label{eq:pr1}
\Pr[\mathcal{E}_{1,j}] \ge 1-O(\wmin^{\beta-2}y_{j-1}^{-\Omega(\varepsilon)}).
\end{equation}
Note that depending on how the weight $y_{i-1}$ compares to $w_0'$, we are either using $\varepsilon_1$ or $\varepsilon_2$, and the last exponent of the above expression is either $-\Omega(1)$ or $-\Omega((\log\log f_0(n))^{-1})$.
\medskip

\paras{Case { $i=1$} and { $j=\infty$}:}
Next we give a lower bound for $\Pr[\mathcal{E}_{1,\infty}]$ in exactly the same way. Let $v$ be a vertex in this extra layer $A_{1,\infty}$, then by definition it has the property $\w v^{-\gamma(\varepsilon_2)} \le \phi(v) \le \w v^{-\gamma(\varepsilon_1)}$. Therefore, we can apply Lemma~\ref{lem:firstphase}~(i) only with $\varepsilon_1$ and not with $\varepsilon_2$. We observe that every vertex $v'$ in the considered set $V^+(v,\varepsilon_1)$ satisfies
$$\phi(v') \ge \phi(v) \w v^{\gamma(\varepsilon_1)-1} \ge \w v^{\gamma(\varepsilon_1)-1-\gamma(\varepsilon_2)} \ge \w{v'}^{(\gamma(\varepsilon_1)-1-\gamma(\varepsilon_2))/\gamma(\zeta \varepsilon_1)}\ge \w{v'}^{-\gamma(\varepsilon_1)}$$
for $\varepsilon_1$ small enough. Then indeed $v' \notin B_{1,\infty}$, and thus $V^+(v,\varepsilon_1) \subset \probSpace_2$. This allows us to argue exactly in the same arguments as above in the case $j<\infty$. Then we apply Lemma~\ref{lem:firstphase}~(ii) and deduce that
\begin{equation}\label{eq:pr1inf}
\Pr[\mathcal{E}_{1,\infty}] \ge 1-O(\wmin^{\beta-2}\w 0^{-\Omega(\varepsilon)}).\medskip
\end{equation} 

\paras{Case { $i=2$}:}
We continue with the events $\mathcal{E}_{2,j}$ for the layers $A_{2,j}$. Here, the argument works similar: Again we can assume that the set $A_{2,j} \cap P_{2,j}$ is non-empty, and that the first vertex $v \in A_{2,j} \cap P_{2,j}$ satisfies \eqref{eq:E_icondition}. %$\Ex_{> \phi_0}[|\Gamma(v) \cap V^+(v,\varepsilon_1) \cap V_{> \phi_0}|] < E_2$. 
By Lemma~\ref{lem:secondphase}~(i) we have $\Ex[|\Gamma(v) \cap V^+(v,\varepsilon) |] = \Omega(\wmin^{\beta-2}\phi(v)^{-\Omega(\varepsilon)})$, and as in~\eqref{eq:choiceofEi} our choice of $E$ ensures that also $\Ex[|\Gamma(v) \cap V^+(v,\varepsilon) \cap V_{\leq \phi_0}|] = \Omega(\wmin^{\beta-2}\phi(v)^{-\Omega(\varepsilon)})$. By the Chernoff bound, and since $\psi_{j-1} \ge \phi(v) > \psi_j$, with probability at least $1-\exp(-\Omega(\wmin^{\beta-2} \psi_{j-1}^{-\Omega(\varepsilon)}))$ the set $\Gamma(v) \cap V^+(v,\varepsilon) \cap V_{\leq\phi_0}$ is non-empty. In this case, there exists a neighbor $v'$ in $\overline{\probSpace}(w_0,\phi_0) \setminus B_{2,j}$ with objective larger than $\psi_{j-1}$. It remains to show that there exists no neighbor in $\probSpace_1$ with higher objective. By Lemma~\ref{lem:correlatedprobs} and the same arguments as above we can use Lemma~\ref{lem:secondphase}~(ii) to see that the expected number of neighbors of $v$ inside $B_{2,j}$ with objective larger than $\psi_{j-1}$ is $O(\wmin^{\beta-2}\psi_{j-1}^{\Omega(\varepsilon)})$. Then Markov's inequality shows that with probability $1-O(\wmin^{\beta-2}\psi_{j-1}^{\Omega(\varepsilon)})$, there exists no such neighbor. Thus
\begin{equation}\label{eq:pr2}
\Pr[\mathcal{E}_{2,j}] \ge 1-O(\wmin^{\beta-2}\psi_j^{\Omega(\varepsilon)}).
\end{equation}
Here, the last exponent is either $\Omega(1)$ or $\Omega((\log\log f_0(n))^{-1})$, depending on which $\varepsilon$ we consider.
\para{Lower bound for {\boldmath$\Pr[\mathcal{E}]$}:}

We have computed a lower bound for $\Pr[\mathcal{E}_{i,j}]$ for every layer $A_{i,j}$. Combining Equations~(\ref{eq:pr1}),(\ref{eq:pr1inf}), and (\ref{eq:pr2}) yields
\begin{align*}
\Pr[\mathcal{E}] &\ge 1 - \sum_j \Pr[\mathcal{E}_{1,j}^c]-\Pr[\mathcal{E}_{1,\infty}] - \sum_{j\ge1} \Pr[\mathcal{E}_{1,j}^c]\\
& = 1 - O\left(\wmin^{\beta-2}\left(w_0^{-\Omega(1)}+w_0'^{-\Omega(\varepsilon_2)}+\phi_0^{\Omega(1)}+\phi_0'^{\Omega(\varepsilon_2)} \right)\right) \\
& \stackrel{\eqref{eq:w0epsi}}{=} 1 - O\left(\wmin^{\beta-2}\min\{w_0,\phi_0^{-1}\}^{-\Omega(1)}\right).
\end{align*}
%By definition of $w_0'$ and $\varepsilon_2$,
%$$w_0'^{-\Omega(\varepsilon_2)}=w_0^{-\Omega(\gamma(\zeta\varepsilon_1)^{f_0(n)} /\log\log f_0(n))}=o\left(w_0^{-\Omega(1)}\right)$$
%as $\gamma(\varepsilon_1)>1$, and similarly we can bound $\phi_0'^{\Omega(\varepsilon_2)}$. Therefore,
%$$\Pr[\mathcal{E}] = 1 - O\left(\wmin^{\beta-2}\max\{w_0^{-1},\phi_0\}^{\Omega(1)}\right).$$

\para{{\boldmath$\mathcal{E}$} implies {\boldmath $(i)$-$(vi)$}:} 
Having proven that $\mathcal{E}$ occurs with sufficiently high probability, it remains to show that this event implies all desired properties $(i)$-$(vi)$. Hence suppose that $\mathcal{E}$ occurs and that $P$ visits at least one vertex of $\overline{\probSpace}(w_0,\phi_0)$. Then the first vertex $u_1 \in P \cap \overline{\probSpace}(w_0,\phi_0)$ is contained in some layer $A_{i,j}$. We show that there exists a subpath $P' = \{u_1, \ldots, u_{\ell}\} \subseteq P$ which visits every layer at most once in both phases and then implies $(i)$-$(vi)$.

\paras{First Phase ($i=1$):} 
Suppose that $u_1 \in \probSpace_1$. Then there exists a maximal subpath $$P_1 =\{u_1, \ldots, u_{\ell'}\} \subseteq P$$ starting at $u_1$ such that $P_1 \subset \probSpace_1$ and no vertex before $u_{\ell'}$ satisfies Condition~(vi). We start with an analysis of this subpath $P_1$ and show by induction that every vertex of $P_1$ is located in a higher layer than its predecessor. Clearly, the induction hypothesis holds for the base case $\{u_1\}$. Assume that the hypothesis holds for $\{u_1, \ldots, u_{i}\}$. Then $u_i$ is contained in layer $A_{1,j}$, and by the induction hypothesis, $u_i$ is the \emph{first} vertex that $P$ visits in this layer. Moreover, by the induction hypothesis it follows that $\{u_1, \ldots, u_i\}$ is also a subpath of the induced greedy path $P_{1,j}$. Then by the event $\mathcal{E}_{1,j}$ we have one of the following two cases: 
\begin{itemize}
\item $u_i$ satsifies condition $(vi)$: Then set $P'=(u_1,\ldots,u_{\ell})$ with $u_{\ell} := u_i$.
\item  $u_i$ has a good neighbor $v'$: Let $u_{i+1}$ be the best neighbor of $u_i$. Then $u_{i+1} \in \probSpace_2$, or $u_{i+1}$ is located in a layer $A_{1,j'}$ where $j'>j$ because 
$u_{i+1}$ has higher weight than every vertex in $A_{1,j}$. Since the protocol makes {greedy choices}, $u_{i+1}$ follows $u_i$ on the path $P$.
\end{itemize}
This proves the induction hypothesis. We see that $P_1$ traverses the layers $A_{1,j}$ according to our ordering and visits every layer at most once. Furthermore, applying the same argument for the last vertex $u_{\ell'} \in P_1$ shows that either $u_{\ell'}$ satisfies $(vi)$, or it has a good neighbor in a layer $A_{2,j}$. 
Note that in the former case, we would choose $u_{\ell} := u_{\ell'}$ and $P'=P_1$, and (i), (ii), (vi) follow directly. 
Otherwise, we claim that we will find a path $P'$ such that $P' \cap \probSpace_1 = P_1$.

\paras{Proof of $(iii)$:} If $\{u_i,u_{i+1},u_{i+2}\}$ are three subsequent vertices on the subpath $P_1$, then the weight increases at least by an exponent $\gamma(\zeta\varepsilon_1)$ between $u_i$ and $u_{i+2}$ since there exists at least one layer in-between containing $u_{i+1}$. Then (iii) follows as we are assuming that $P'$ will not visit any vertex in $\probSpace_1 \setminus P_1$.

\paras{Second Phase ($i=2$):} Suppose that either the last vertex $u_{\ell'} \in P_1$ does not satisfy (vi) or $s \in \probSpace_2$. Then there exists a first vertex $u_{\ell'+1} \in P \cap \probSpace_2'$, and a maximal subpath 
$$P_2 = P_1 \cup \{u_{\ell'+1}, \ldots, u_{\ell}\} \subseteq P$$
starting at $u_1$ such that $(P_2 \setminus P_1) \subset \probSpace_2$ and no vertex before $u_{\ell}$ satisfies Condition (vi). We show by induction that every vertex of $P_2 \setminus P_1$ is located in a higher layer (w.r.t.\ the objective $\phi$) than its predecessor. Again, the hypothesis holds for $P_1 \cup \{u_{\ell'+1}\}$.
Assume that the induction hypothesis holds for $P_1 \cup \{u_{\ell'+1}, \ldots, u_i\}$. Then $u_i$ is contained in a layer $A_{2,j}$, and by the induction hypothesis, $u_i$ is the \emph{first} vertex that $P$ visits in this layer. Furthermore, $P_1 \cup \{u_{\ell'+1}, \ldots, u_i\}$ is a subpath of the induced greedy path $P_{2,j}$. Then by the event $\mathcal{E}_{2,j}$, we have one of the following two cases:
\begin{itemize}
\item $u_i$ satisfies condition $(vi)$: Then set $P'=(u_1,\ldots,u_{\ell})$ with $u_{\ell} := u_i$.

\item $u_i$ has a good neighbor $v'$: Let $u_{i+1}$ be the best neighbor of $u_i$.
Then $u_i \in A_{2,j'}$ holds for $j'<j$. Since the protocol makes {greedy choices}, $u_{i+1}$ follows $u_i$ on the path $P$.
\end{itemize}
This proves the induction hypothesis, and we observe that $P_2$ traverses the layers $A_{2,j}$ according to our ordering and visits every layer at most once. However, by repeating the argument for the last vertex $u_{\ell}$ it follows that the best neighbor of $u_{\ell}$ can \emph{not} be in $\probSpace_1$. Then the only possibility is that $u_{\ell}$ satisfies Condition~(vi). Hence we put $P'=P_2$, and Properties~(i), (ii), and (vi) follow. In particular, we see that $P' \cap \probSpace_1 = P_1$ as claimed above.
 
\paras{Proof of $(iv)$:} For three subsequent vertices $\{u_i,u_{i+1},u_{i+2}\}$ on the subpath $\{u_{\ell'+1},\ldots,u_{\ell}\}$, the objective increases always at least by an exponent $\gamma(\varepsilon_1)$ between $u_i$ and $u_{i+2}$ as there lies at least one layer in-between containing $u_{i+1}$. This proves $(iv)$, and the only remaining property is $(v)$.

\paras{Proof of $(v)$:} We already know that the path $P'$ visits every layer at most once. Therefore we can upper-bound the length $|P'|$ by counting the total number of potentially visited layers. By construction we have $f_0(n)$ first layers $A_{1,j}$ which are defined via $\varepsilon_1$, and similarly there are $f_0(n)$ final layers $A_{2,j}$ defined via $\varepsilon_1$. Let $L_2$ denote the number of layers $A_{2,j}$ defined via $\varepsilon_2=o(1)$ which the routing potentially visited. In order to upper-bound $L_2$, we recall that between two neighboring vertices of $P'$, the objective is always increasing. Therefore $P'$ contains only vertices of objective at least $\phi(u_1)$. Clearly, this implies that we only need to consider layers $A_{2,j}$ where $\psi_{j-1}\ge \phi(u_1)$. Therefore $L_2$ is upper-bounded by the solution of $\phi_0^{-\gamma(\varepsilon_2)^{L_2}}=\phi(u_1)^{-1}$, and we obtain
$$L_2 = \frac{\log\log_{\phi_0^{-1}} (\phi(u_1)^{-1})}{\log \gamma(\varepsilon_2)} \le \frac{\log\log_{\phi_0^{-1}} (\phi(u_1)^{-1})}{\log ((\beta-2)^{-1/(1+o(1))})}= \frac{1+o(1)}{|\log(\beta-2)|} \log\log_{\phi_0^{-1}} (\phi(u_1)^{-1}).$$

It remains to upper-bound the number $L_1$ of layers $A_{1,j}$, defined via $\varepsilon_2$, which $P'$ potentially visits. We observe that every vertex $v \in P'$ satisfies $\phi(v)\ge\phi(u_1)$. Thus, every vertex $v \in P'$ with weight at least $\phi(u_1)^{-1/\gamma(\varepsilon_1)}<\phi(u_1)^{-1}$ belongs to $\probSpace_2$, and we only need to count layers $A_{1,j}$ which contain vertices of weight at most $\phi(u_1)^{-1/\gamma(\varepsilon_1)}$. Then $L_1$ is bounded from above by the solution of $w_0^{\gamma(\varepsilon_2)^{L_1}}=\phi(u_1)^{-1}$. It follows
$$L_1=\frac{\log\log_{w_0} (\phi(u_1)^{-1})}{\log \gamma(\varepsilon_2)} \le \frac{\log\log_{w_0} (\phi(u_1)^{-1})}{\log ((\beta-2)^{-1/(1+o(1))})}= \frac{1+o(1)}{|\log(\beta-2)|} \log\log_{w_0} (\phi(u_1)^{-1}).$$
Then the length of $P'$ is at most $L_1+L_2+O(f(n))$, which proves $(v)$.
\end{proof}

\subsection{Proof of Theorem~\ref{thm:greedysuccess1} (Success Probability)}
\label{subsec:constantsuccess}
In this subsection we prove that for all initial choices of $s$ and $t$, greedy routing succeeds with probability $\Omega(1)$. We do this in three steps: In a first step we show that with probability $\Omega(1)$, the best neighbor $u_1$ of $s$ has weight $\w {u_1} \ge w_0$, where $w_0$ is a given constant. In a second step we use Lemma~\ref{lem:main} to see that then with probability $\Omega(1)$, from $u_1$ the greedy algorithm on $V_{\leq \phi_0}$ finds a path to a vertex $u_{\ell}$ which has in expectation $\Omega(1)$ neighbors of objective at least $\phi_0$. Since the number of such neighbors is Poisson distributed, with probability $\Omega(1)$ there is at least one such neighbor, and we call the best of them $v_1$. Finally, in a last step we will show that with constant probability, the routing will reach $t$ from $v_1$ via at most one intermediate vertex $v_2$. For all steps, the success probability is a constant larger than zero, but not large enough such that we can do a union bound over the error events. Therefore, we always compute the probabilities conditioned on the previous steps being successful. %In most case this is straightforward, except for the case where we need the intermediate vertex $u_4$. In this case we start by first finding with probability $\Omega(1)$ a good candidate for $u_4$, and conditioning all other steps on the existence of this vertex.

\begin{proof}[Proof of Theorem~\ref{thm:greedysuccess1}]
We assume that the parameter $\wmin$ is lower-bounded by a small constant, say $\varepsilon_1$. Let us start by choosing a weight $w_0=\Theta(1)$ and an objective $\phi_0=\Theta(1)$ such that they satisfy the preconditions of Lemma~\ref{lem:main} and such that the events $(i)$-$(vi)$ of Lemma~\ref{lem:main} occur with constant probability $\Omega(1)>0$.

\para{Start of routing process:}
Suppose for now that $\phi(s) \le \phi_0$. We define the  event
 $$\mathcal{E}_s: \text{$s$ itself or its best neighbor have weight at least $w_0$}.$$
 If $\w s \ge w_0$, this happens with probability $1$. Thus we assume $\w s \le w_0$. Let \mbox{$r := (\frac{c_1\w s w_0}{\wmin n})^{1/d}$}, where $c_1$ is the constant given by \eqref{eq:puv} and \eqref{eq:puv2}. Since $\phi(s) \le \phi_0$, we have
\begin{equation}\label{eq:rs}
r^d = c_1 w_0 \phi(s) \|\x s - \x t\|^d \le c_1 w_0 \phi_0 \|\x s - \x t\|^d = O(\|\x s - \x t\|^d).
\end{equation}
Next we define a subset of $\probSpace$ in which we want to find a neighbor of $s$. Let
$$A_s := \left\{u \in \probSpace\  \middle| \begin{array}{l} (1): \w u \ge w_0; \\ (2): \|\x u - \x t\| \le \|\x s - \x t\|; \\ (3): \|\x u - \x s \| \le r \end{array}\right\}.$$
By $(\ref{eq:puv})$ and $(\ref{eq:puv2})$, every vertex $u \in A_s$ is connected to $s$ with probability $\Omega(1)$. Furthermore, by Equation~(\ref{eq:rs}) a random vertex which satisfies $(3)$ has property~$(2)$ with probability $\Omega(1)$ as well. Therefore a random vertex satisfies $(2)$ and $(3)$ with probability $\Omega(r^d)$, and by integrating over all weights larger than $w_0$ it follows that the expected number of vertices in $A_s$ is at least
$$\Omega(n \wmin^{\beta-1} w_0^{1-\beta}r^d)=\Omega(\wmin^{\beta-2} w_0^{2-\beta} \w s)=\Omega(1).$$
Hence with probability $\Omega(1)$, there exists at least one vertex $u \in \Gamma(s) \cap A_s$. For every \mbox{$u \in \Gamma(s) \cap A_s$} it holds $\phi(u)\ge\phi(s)$, as the distance to $t$ decreases and the weight increases. Then we take $u_1$ as the neighbor of $s$ in $A_s$ with maximal objective. By definition of $A_s$ we have $\w {u_1} \ge w_0$.

Since $\w s=O(1)$, the degree distribution of $s$ is a Poisson random variable with rate $\Theta(1)$, and thus with probability $\Omega(1)$, $s$ has no neighbor of weight less than $w_0$ with higher objective than $\phi(u_1)$. In this case, the routing indeed proceeds with the vertex $u_1$ and we obtain $\Pr[\mathcal{E}_s]=\Omega(1)$. 

\para{Main part of routing process:} Our aim is to show that, conditioned on $\mathcal{E}_s$, greedy routing proceeds to a vertex of objective larger than $\phi_0$. If $\phi(u_1) > \phi_0$ then there is nothing to show, so assume otherwise. We want to show that on the graph induced by the vertex set $V_{\le \phi_0}$, greedy routing finds a vertex $u_{\ell}$ which we expect to have neighbors in $V_{>\phi_0}$. More formally, we consider the event 
$$\mathcal{E}_m: \text{the algorithm on $G[V_{\leq \phi_0}]$ finds a vertex $u_{\ell}$ with $\Ex_{>\phi_0}[|\Gamma(u_{\ell}) \cap V_{>\phi_0}|]=\Omega(1)$}.$$ 
Then by our choice of $w_0$ and $\phi_0$, Lemma~\ref{lem:main}~(vi) gives the lower bound $\Pr[\mathcal{E}_m \mid \mathcal{E}_s] = \Omega(1)$.
%It follows that
%$$\Pr[\mathcal{E}_s \wedge \mathcal{E}_m] = \Pr[\mathcal{E}_s] \cdot \Pr[\mathcal{E}_m \mid \mathcal{E}_s]=\Omega(1).$$

\para{End of routing process:}
Suppose that either both events $\mathcal{E}_s$ and $\mathcal{E}_m$ hold, or that already $\phi(s)\ge\phi_0$. We want to prove that with constant probability, the routing finds the target $t$ while visiting at most two other vertices of $V_{> \phi}$. Let $\overline{w}$ be a constant weight chosen large enough and let
$$A_t := \{u \in \probSpace \mid \w u \ge \overline{w} \wedge \|\x u - \x t\|^d \le (n \phi_0)^{-1}\}.$$
We claim that the expected number of vertices in the set $A_t$ is $\Omega(1)$. Note that here, we need to condition on the events $\mathcal{E}_s$ and $\mathcal{E}_m$. However, $A_t$ is independent of $\mathcal{E}_m$ since $A_t \subseteq V_{>\phi_0}$, and conditioning on $\mathcal{E}_s$ only increases the expected size of $A_t$. Thus
$$
\Ex[|A_t| \mid \mathcal{E}_s \wedge \mathcal{E}_m] \ge \Ex[|A_t|] = \Omega(n \wmin^{\beta-1} (\phi_0)^{\beta-1}(n\phi_0)^{-1})=\Omega(1).
$$
Now, we uncover all remaining vertices, and all edges that involve at least one vertex from $V_{\le \phi_0}$, but not the edges within $V_{> \phi_0}$. Due to the Poisson point process, both $|A_t|$ and $|\Gamma(u_{\ell}) \cap V_{> \phi_0}|$ are Poisson random variables with constant mean, therefore with constant probability both sets are non-empty. In this case, there exists a vertex $v_1 \in \Gamma(u_{\ell}) \cap V_{> \phi_0}$ with maximal objective. Note that if $v_1=t$, then we are already done. Furthermore, in the case $\alpha<\infty$ we can toss a coin whether the edge $\{v_1,t\}$ is present. By \eqref{eq:puv} it holds $p_{v_1 t} = \Omega\big(\min\{1,(\phi(v_1) \w t)^{\alpha}\}\big) = \Omega(1)$, which is already sufficient for finding $t$ with constant probability.

It remains the case $\alpha=\infty$. Recall that we already know that there exists at least one vertex $v_3 \in A_t$. We claim that for every $u' \in V_{> \phi_0}$ and $v' \in A_t$ it holds  $p_{u'v'}=\Omega(1)$. Indeed, by the triangle inequality we have $\|\x{u'} - \x {v'}\| \le 2 \max\{\|\x {u'} - \x t\|,\|\x {v'} - \x t\|\}$, and for $\overline{w}$ large enough we deduce
$$\frac{\w {u'} \w {v'}}{\wmin n \|\x {u'} - \x {v'}\|^d} \ge \frac{1}{2}\max\{\phi(u')\w {v'}, \phi(v')\w {u'}\}=\Omega(\overline{w}\phi_0) \ge \frac{1}{c_1},$$
and then by \eqref{eq:puv2} it follows $p_{u'v'}=\Omega(1)$. In particular, it also holds $p_{t v'}=\Omega(1)$. 

We can finish the proof as follows: We toss the coins for all edges in $V_{> \phi_0}$ except the edges incident to $t$. By the above observation, with constant probability the edge $\{v_1, v_3\}$ is present. In this case, either the best neighbor $v_2$ of $v_1$ satisfies $\phi(v_2) \ge \phi(v_1)$ or we even have $\phi(v_1)\ge \phi(v_3)$ (and put $v_2 = v_1$). We see that 
$$\frac{\w {v_2} \w t}{\wmin n \|\x {v_2} \x t\|^d} = \phi(v_2) \w t \ge \phi(v_3) \w t \ge \frac{1}{c_1}.$$
Finally, we uncover the edge $\{v_2, t\}$ and independently of the previous events, this edge is also present with constant probability and with probability $\Omega(1)>0$, the routing finds $t$ via $v_2$.
\end{proof}

\subsection{Proof of Theorem~\ref{thm:greedysuccess2} (Failure Probability is Exponentially Small in {\boldmath$\wmin$})}
\label{subsec:expsuccess}

In this section, we assume in addition to (EP1) and (EP2) that whenever two vertices $u$ and $v$ satisfy $\|\x u - \x v\|^d \le \frac{c_1\w u \w v}{\wmin n}$, then the edge $\{u,v\}$ is present deterministically. This assumption is very natural and in particular satisfied by hyperbolic random graphs. We first prove the second statement of the theorem which is a direct consequence of Lemma~\ref{lem:main}.
\begin{proof}[Proof of Theorem~\ref{thm:greedysuccess2}~(ii)]
Let $s$ and $t$ be two vertices such that both $\w s = \omega(1)$ and $\w t = \omega(1)$. Then we apply Lemma~\ref{lem:main} for $w_0:=\w s$ and $\phi_0:=(c_1\w t)^{-1}$. By Statement~(vi) it follows that with probability 
$$1-O(\wmin^{\beta-2} \min\{\w s, \w t\}^{-\Omega(1)})=1-\min\{\w s, \w t\}^{-\Omega(1)}$$
the greedy routing finds a vertex $u'$ of objective at least $\phi_0 = (c_1\w t)^{-1}$. Let $P$ be the final greedy path, and let $u''$ be its first vertex which has objective at least $\phi_0$. Then either $u''=t$, or
$$c_1\frac{\w u'' \w t}{\wmin n\|\x u'' - \x t\|^d}\ge c_1 \phi_0 \w t =1.$$
In this case, $\{u'',t\}\in E$ and thus the routing finds $t$ via $u''$.
\end{proof}

For part~(i) of Theorem~\ref{thm:greedysuccess2}, note that the statement becomes trivial if $\wmin$ is small. So we can assume that $\wmin$ is larger than a given constant $c$ which depends on $\alpha,\beta,d$ and $\varepsilon$.  %The reason for this is that by choosing the hidden factors in the claimed probability large enough, the statement of the theorem becomes trivial for all smaller values of $\wmin$. 
In particular, this allows us to assume $\wmin > w_1(\varepsilon_1)$.

For the proof of Theorem~\ref{thm:greedysuccess2}~(i), we can not directly apply Lemma~\ref{lem:main} because this yields only a failure probability which is polynomially small in $w_0 \ge \wmin$. Instead, we will apply Lemma~\ref{lem:smweights} for the very first steps, which will allow us to increase the lower-bound on the success probability. Afterwards, we will apply Lemma~\ref{lem:main} for values $w_0$, $\phi_0$ which depend exponentially in $\wmin$ such that the lemma indeed yields a sufficiently small failure probability. Finally, we need to ensure that after seeing vertices of objective at least $\phi_0$, the routing does not die out and continues until it reaches vertices which connect deterministically to $t$. 

For the first step, we will apply the following lemma.

\begin{lemma} \label{lem:expwminfirststeps}
Let $w_0 := e^{\wmin^{\Omega(1)}}$ and assume that $\wmin$ is sufficiently large. Then there exists a constant $c'>0$ such that with probability
$$1-O(e^{-\wmin^{\Omega(1)}})$$
the greedy routing visits at least one vertex $u$ which satisfies either (i) $\w u \ge w_0$, (ii) $\phi(u)\ge w_0^{-c'}$, or (iii) $u \in \probSpace_2$.
\end{lemma}
\begin{proof}
We use similar techniques as in the proof of Lemma~\ref{lem:main}. Notice that if the starting vertex $s$ satisfies at least one of the three properties, there is nothing to show. Let $\phi_0 := w_0^{-c'}$ and assume that $\w s \le w_0$, $\phi(s) \le \phi_0$, and $s \in \probSpace_1$. We start by defining a weight sequence $y_{(i)}$ as follows: We put $y_0 := \w s$, and for $i \ge 1$ we put $y_i := y_{i-1}^{\gamma(\zeta\varepsilon_1)}$. Notice that there exists a first $y_j$ such that $y_j \ge w_0$. Moreover, for all $1 \le i \le j$ we define
$$A_i := \{v \in \probSpace_1 \mid (1): y_{i-1} \le \w v < y_i; (2): \phi(v) \le \phi_0\}$$
and
$$B_i := \{v \in \probSpace_1 \mid (1): \w v \le y_i; (2): \phi(v) \le \phi_0\}.$$
Furthermore, denote by $P_i$ the greedy path induced by the set $B_i$ and by $\mathcal{E}_i$ the event that either $P_i \cap A_i$ is empty or the first vertex $v \in P_i \cap A_i$ has at least one neighbor $v' \notin B_i$ such that $\phi(v') \ge \phi(v)$ and $\phi(v') \ge \phi(u)$ holds for all $u \in \Gamma(v) \cap B_i$.
Let $1\le i \le j$. Clearly, we have $\Pr[\mathcal{E}_i \mid P_i \cap A_i = \emptyset]=1$, and in order to lower-bound $\Pr[\mathcal{E}_i]$ we can assume $P_i \cap A_i \neq \emptyset$. Let $v$ be the first vertex of $P_i \cap A_i$. Then by Lemma~\ref{lem:firstphase}~(i), it holds $\Ex[|\Gamma(v) \cap V^+(v,\varepsilon_1)|]=\Omega(\wmin^{\beta-2} \w v^{\varepsilon_1})$. By a Chernoff bound, with probability at least $1-\exp(-\Omega(\wmin^{\beta-2} y_{i-1}^{\varepsilon_1}))$ there exists a neighbor $v'$ of $v$ in the set $V^+(v,\varepsilon_1)$. In particular $v' \notin B_i$ and $\phi(v') \ge \phi(v) \w v^{\gamma(\varepsilon_1)-1}$.

Next we observe that every vertex $u \in B_i$ with objective at least $\phi(v) \w v^{\gamma(\varepsilon_1)-1}$ is contained in $V^-(v,\varepsilon_1)$. Hence in order to verify that $v'$ is better than every neighbor of $v$ in $B_i$, it is sufficient that $v$ has no neighbor in $V^-(v,\varepsilon_1)$.   We observe that by Lemma~\ref{lem:correlatedprobs}, our conditioning on $v$ being the first vertex in $P_i \cap A_i$ decreases the expected size of $V^-(v,\varepsilon_1)$. Furthermore, 
for $c'$ large enough we have $\phi(v) \w v^{ \gamma(\zeta\varepsilon_1)} \le \phi_0 w_0^{\gamma(\zeta\varepsilon_1)} \le \bar{c}$, where $\bar{c}$ is the constant given by Lemma~\ref{lem:smweights}. Hence by Lemma~\ref{lem:smweights}, we obtain
$$\Ex[|\Gamma(v) \cap V^-(v,\varepsilon)|] = O(\phi(v)^{\Omega(1)}\w v^{O(1)})=O(\phi_0^{\Omega(1)}y_{i-1}^{O(1)}).$$
By Markov's inequality, with probability $1-O(\phi_0^{\Omega(1)}y_{i-1}^{O(1)})$ this set is empty and in this case, $\phi(v')$ is higher than the objective of every neighbor of $v$ in $B_i$. It follows that
$$\Pr[\mathcal{E}_i] \ge 1 - O(\phi_0^{\Omega(1)}y_{i-1}^{O(1)}) - \exp(-\Omega(\wmin^{\beta-2} y_{i-1}^{\varepsilon_1})).$$
Recall that the sequence $y_{(i)}$ grows exponentially, and we obtain
\begin{align*}
\Pr\left[ \wedge_{i=1}^j \mathcal{E}_i \right] & = 1 - O\left(\phi_0^{\Omega(1)}y_0^{O(1)}\right) - \exp\left(-\Omega(\wmin^{\beta-2} y_0^{\varepsilon_1})\right)\\
& = 1 - O\left(w_0^{-c' \cdot \Omega(1)+O(1)} + \exp\left(-\Omega(\wmin^{\beta-2} w_0^{\Omega(1)})\right)\right)\\
& = 1 - O\left(w_0^{-\Omega(1)}\right)
\end{align*}
if $c'$ is chosen sufficiently large. This allows us to assume that all good events $\mathcal{E}_i$ occur. Let us consider the final greedy path $P=\{s=v_0,v_1,v_2, \ldots\}$ and suppose that all good events $\mathcal{E}_i$ occur. By construction $s \in A_1$, and $\mathcal{E}_1$ implies that the routing does not die out at $s$ (i.e., $s$ is not isolated) and $v_1 \notin B_1$. Then we distinguish two cases: If $v_1 \notin B_j$, then we are done since $v_1$ would satisfy (i),(ii) or (iii). Else, $v_1$ is contained in a layer $A_i$, and in particular it is also the first vertex of $P_i \cap A_i$. Then we can repeat the argument, apply the event $\mathcal{E}_i$ and see again that there exists a subsequent vertex $v_2$ on the greedy path which is either located outside $B_j$ or contained in a layer $A_{i'}$ for $i' > i$. We can repeat this argument inductively. However, at the latest when we reach a vertex $v_k \in P \cap A_j$, the event $\mathcal{E}_j$ implies that the best neighbor $u$ of $v_k$ is located outside $B_j$. Then $u$ must fufill at least one of the properties (i)-(iii), which proves the lemma. 
\end{proof}

Next, we study the end-phase of the routing and show that once the routing arrives at a vertex of high objective, with sufficiently high probability it continues to a vertex $u$ such that $p_{ut}=\Omega(1)$.

\begin{lemma} \label{lem:expwminlaststeps}
Let $\phi_0=e^{-\wmin^{O(1)}}$ and assume that $\wmin$ is sufficiently large. Then if the greedy routing reaches a vertex $v$ with $\phi(v) \ge \phi_0$, with probability
$$1-O(e^{-\wmin^{\Omega(1)}})$$
it also visits a vertex $u$ with $\phi(u) \ge (c_1 \w t)^{-1}$.
\end{lemma}
\begin{proof}
If $\phi_0 \ge \phi_1 := \frac{1}{c_1 \w t}$, there is nothing to show, so we can assume $\phi_0 < \phi_1$. In the following we partition the vertex set with objective between $\phi_0$ and $\phi_1$ in small layers. Let $\psi_0 := \phi_0$, and for $i\ge1$ put $\psi_i := \psi_{i-1} \cdot \wmin^{(1-1/\gamma(\varepsilon_1))/2}$. We assume w.l.o.g. that there exists an index $j$ such that $\psi_j = \phi_1$. Then for all $1 \le i \le j$ we define $A_i := \{v \in \probSpace \mid \psi_{i-1} \le \phi(v) < \psi_i\}$. Our goal is to show that with probability sufficiently high, the greedy routing visits at most one vertex per layer $A_i$, and if so, it will not stop at such a vertex. Therefore, we denote by $P_i$ the greedy path induced by all vertices of objective less than $\psi_i$ and by $\mathcal{E}_i$ the event that either $P_i \cap A_i$ is empty or the first vertex $v \in P_i \cap A_i$ has a neighbor $v'$ such that $\phi(v') \ge \psi_i$.

Let $1 \le i \le j$. If $P_i \cap A_i = \emptyset$, then $\mathcal{E}_i$ holds with probability $1$. Hence suppose that $P_i \cap A_i$ is non-empty and let $v$ be the first vertex on $P_i \cap A_i$. We distinguish two cases: First suppose $v \in \probSpace_1$. Note that for $\varepsilon_1$ small enough it holds
$$\gamma(\varepsilon_1)-1 =\frac{3-\beta-\varepsilon_1}{\beta-2} \ge \frac{3-\beta-\varepsilon_1}{2-2\varepsilon} = \frac{1-1/\gamma(\varepsilon_1)}{2}.$$ 
Then by Lemma~\ref{lem:firstphase}~(i) and a Chernoff bound, with probability at least $1-\exp(-\Omega(\wmin^{\Omega(1)}))$ there exists a neighbor $v'$ of $v$ with objective
$$\phi(v') \ge \phi(v) \w v^{\gamma(\varepsilon_1)-1} \ge \phi(v) \wmin^{\gamma(\varepsilon_1)-1}\ge \phi(v) \wmin^{(1-1/\gamma(\varepsilon_1))/2}\ge \psi_i.$$ 

It remains the case $v \in \probSpace_2$. Since $\w t \ge \wmin$ and since we are assuming that $\wmin$ is sufficiently large, it holds $\phi(v) \le (c_1 \w t)^{-1} \le (c_1 \wmin)^{-1} \le 1$. This allows us to apply Lemma~\ref{lem:secondphase}~(i), and we see that the expected number of neighbors of $v$ with objective at least $\phi(v)^{1/\gamma(\varepsilon_1)}$ is at least
$\Omega(\wmin^{\beta-2}\phi(v)^{-\Omega(\varepsilon_1)}) = \Omega(\wmin^{\Omega(\varepsilon_1)})$. By a Chernoff bound, the probability that such a neighbor $v'$ exists is at least $1-\exp(-\Omega(\wmin^{\Omega(1)}))$. By definition of $V^+(v,\varepsilon_1)$ we have at least $\phi(v') \ge \phi(v)^{1/\gamma(\varepsilon_1)}$. In addition, we can use that $\phi(v) \ge \phi_1$ and $\wmin$ is a sufficiently large constant. Then it follows
\begin{align*}
\phi(v') & \ge \phi(v)^{1/\gamma(\varepsilon_1)} = \phi(v)^{1 - (3-\beta-\varepsilon)/(1-\varepsilon)} \ge \phi(v)(c_1 \wmin)^{(3-\beta-\varepsilon)/(1-\varepsilon)}\\ &\ge \phi(v) \wmin^{(3-\beta-\varepsilon)/(2-2\varepsilon)}=\phi(v) \wmin^{(1-1/\gamma(\varepsilon_1))/2} \ge \psi_i.
\end{align*}

It follows that $\Pr[\mathcal{E}_i] \ge 1-O(e^{-\wmin^{\Omega(1)}})$ holds for all $1 \le i \le j$. Note that by construction we have
$\phi_0 \wmin^{j (1-1/\gamma(\varepsilon_1))/2}=\phi_1$, which implies $j = O(\frac{|\log \phi_0|}{\log \wmin}) = O(\frac{\wmin^{O(1)}}{\log \wmin})$. By a union bound it follows
$$\Pr\left[\wedge_{i=1}^j\mathcal{E}_i\right]= 1-O(\wmin^{O(1)}e^{-\wmin^{\Omega(1)}}) = 1-O(e^{-\wmin^{\Omega(1)}}).$$

Therefore we are allowed to assume that all good events occur and finish the proof as follows: Let $P$ be the final greedy path and let $v \in P$ be the first vertex contained in a layer $A_i$. Then $v \in P_i \cap A_i$, and the event $\mathcal{E}_i$ implies that $v$ has a neighbor $v'$ with higher objective than $\psi_i$. Then either $\phi(v') \ge \phi_1$ or $v' \in A_{i'}$ for some $i' > i$. We can repeat the argument until we reach the last layer $A_j$, and it follows that indeed $P$ must contain a vertex of objective at least $\phi_1$.
\end{proof}

\begin{proof}[Proof of Theorem~\ref{thm:greedysuccess2}~(i)]
We prove the statement by combining the two previous lemmas with Lemma~\ref{lem:main}. Let $c'>0$ be a sufficiently large constant and let $P$ be the resulting greedy path. We put $w_0 := \exp(\wmin)$ and $\phi_0 := w_0^{-c'}$. First suppose that $\phi(s) \le \phi_0$ and $s  \notin \overline{\probSpace}(w_0,\phi_0)$. In this case, we apply Lemma~\ref{lem:expwminfirststeps} and see that if $c'$ is large enough, then with probability
$$1-O(e^{-\wmin^{\Omega(1)}})$$
there exists a vertex $u \in P$ which is either contained in the large set $\overline{\probSpace}(w_0,\phi_0)$ or satisfies $\phi(u)> \phi_0$.

Next suppose that there exists a first vertex $u_1 \in P \cap  \overline{\probSpace}(w_0,\phi_0)$ (in case $s \in \overline{\probSpace}(w_0,\phi_0)$, we would put $u_1 = s$). If there exists such a vertex, by Lemma~\ref{lem:main} with probability
$$1-O(\wmin^{\beta-2} \min\{w_0,\phi_0^{-1}\}^{-\Omega(1)}) = 1-O(e^{-\wmin^{\Omega(1)}})$$
there exists a vertex $u_{\ell}$ on the greedy path induced by $V_{\le \phi_0}$ such that $u_{\ell} \in \overline{\probSpace}(w_0,\phi_0)$ and 
$$\Ex_{> \phi_0}[|\Gamma(u_{\ell}) \cap V_{> \phi_0}|]=\Omega\left(\wmin^{\beta-2} \min\{w_0,\phi_0^{-1}\}^{\Omega(1)}\right)=\Omega\left(e^{\wmin^{\Omega(1)}}\right).$$

So far, we uncovered the graph $G[V_{\le \phi_0}]$. We continue by tossing the coins for all remaining vertices and edges. By a Chernoff bound, with probability $1-\exp(-\Omega(\exp(\wmin^{\Omega(1)})))$ the set $\Gamma(u_{\ell}) \cap V_{>\phi_0}$ is non-empty, and $u_{\ell}$ has a neighbor $u'$ of objective at least $\phi_0$. Then there are two possible situations: If $P$ visits $u_{\ell}$, then the vertex following after $u_{\ell}$ in $P$ has objective at least $\phi_0$. If $P$ does not visit $u_{\ell}$, then there must exist an earlier vertex $u_i \in P$ from which $P$ jumps to $G[V_{\ge \phi_0}]$. In both cases, the routing arrives at vertex $u'$ satisfying $\phi(u') \ge \phi_0$.

Finally suppose that the routing visits a vertex $u'$ of objective at least $\phi_0$ or even starts at such a vertex. Then by Lemma~\ref{lem:expwminlaststeps}, with probability $1-O(e^{-\wmin^{\Omega(1)}})$ it also visits a vertex $u''$ such that $\phi(u'') \ge (c_1 \w t)^{-1}$.

By a union bound over all error probabilities it follows that with probability $1-O(e^{-\wmin^{\Omega(1)}})$ the greedy path $P$ arrives at a vertex $u''$ fufillling $\phi(u'') \ge (c_1 \w t)^{-1}$. If $u''=t$, we are done, otherwise we have
$$\|\x {u''} - \x t\|^d = \frac{\w {u''}}{\wmin n \phi(u'')} \le \frac{c_1 \w {u''} \w t}{\wmin n}$$
and by our assumption on the model it follows that $\{u'',t\} \in E$ and therefore $t \in P$.
\end{proof}

\subsection{Proof of Theorem~\ref{thm:length} (Length of Greedy Path)}
\label{subsec:lengthofrouting}

Before bounding the length of the resulting greedy path, we look at the very first steps of the routing process. With the following lemma, we show for both the basic routing algorithm and its patching variants that they can not visit too many vertices of low weight before reaching a vertex of growing weight. We will use the lemma both for proving Theorem~\ref{thm:length} and later in Chapter~\ref{sec:proofpatching} for analyzing patching protocols.

\begin{lemma} \label{lem:patchingstart}
Let $w_0=w_0(n)=\omega(1)$ be a function growing in $n$, and suppose that the starting vertex $s$ satisfies $\w s \le w_0$ and $\phi(s) \le e^{-w_0}$. Let $A$ be a routing protocol which satisfies (P1), i.e., it makes greedy choices. Then a.a.s., either $A$ visits at most $O(w_0^4)$ different vertices in total, or after visiting at most $O(w_0^4)$ different vertices, $A$ reaches a vertex of weight at least $w_0$.
\end{lemma}
\begin{proof}
We denote by $G'$ the subgraph induced by all vertices of weight less than $w_0$ and by $G''$ the subgraph induced by all vertices of weight less than $w_0^3$. Let $r=(w_0^c/n)^{1/d}$, where we take $c>0$ sufficiently large. 

We first prove that a.a.s., during the first few hops the algorithm stays close to the starting position $\x s$, unless it finds a large-weight vertex very soon. Suppose that we first uncover all vertices of weight less than $w_0$ and see only $G'$. Let $s'=v_0, v_1, v_2, \ldots$ be the sequence of different vertices visited by algorithm $A$ when we run it on $G'$ (so if a vertex is visited several times, we only list it once). By Lemma~\ref{lem:nojumps}~(ii), with probability $1-O(w_0^{-5})$ there exists no neighbor of $s$ in $G'$ with geometric distance at least $r$ to $\x s$, i.e., there is no \emph{jump} outgoing from $s$. For the following vertices $v_1, v_2, \ldots$, we argue similarly: If we condition on that there was no outgoing jump on the previous vertices $s, v_1, \ldots, v_i$, then by Lemma~\ref{lem:correlatedprobs}, the probability that there is an outgoing jump from $v_{i+1}$ is even smaller than without conditioning, and by Lemma~\ref{lem:nojumps}~(ii) for $v_{i+1}$ this probability is at most $1-O(w_0^{-5})$, regardless whether we condition on the same event for the previous vertices or not. We apply a union bound over the first $2w_0^4$ vertices that the routing algorithm finds in $G'$ and see that with probability $1-O(w_0^{-1})$, we do not see such an outgoing jump during the first $2w_0^4$ steps. Note that if the algorithm visits less vertices in $G'$, we can apply union bound just over all visited vertices and obtain the same result.

As a next step, let us consider first a $w_0^3$-grid as given by Definition~\ref{def:grid} and second the ball $B(2w_0^4 r,s)$ of radius $2w_0^4 r$ around $\x s$. This ball has volume $4^d w_0^{4d} r^d = 4^d w_0^{4d+c} n^{-1}$ and intersects $\Theta(w_0^{4d+c-3})=w_0^{O(1)}$ cells of the $w_0^3$-grid. Let $\mathcal{C}$ be the collection of all these cells contained in our ball $B(2w_0^4 r,s)$. The expected number of vertices of weight at most $w_0$ contained in a single cell $C_i \in \mathcal{C}$ is at most $w_0^3$, and by a Chernoff bound, with probability $1-\exp(-\Omega(w_0^3))$ there are not more than $2w_0^3$ vertices in $C_i$. The same is true for all $C_i \in \mathcal{C}$, and by a union bound we deduce that with probability $1-O(w_0^{O(1)}\exp(-\Omega(w_0^3)))=1-o(1)$ every cell $C_i \in \mathcal{C}$ contains at most $2w_0^3$ vertices.

It follows that after visiting $2w_0^4$ different vertices, the algorithm either has escaped from the ball $B(2w_0^4 r,s)$ or it has visited at least $w_0^3$ different cells. However, from the above calculation we know that a.a.s.\ none of the first $2w_0^4$ visited vertices has a neighbor inside $G'$ with geometric distance at least $r$. Thus for all $i \le 2w_0^4$ it holds 
$$\|\x s - \x {v_i}\| \le 2w_0^4 r.$$
Therefore a.a.s.\ the routing does not escape the ball $B(2w_0^4 r,s)$ during the first $2w_0^4$ hops, and in this case, the algorithm running on $G'$ either stops before visiting $2w_0^4$ vertices or it visited vertices in at least $w_0^3$ different cells. 

In particular, we can lower-bound the geometric distance to $t$ for all these first vertices. Observe that $\|\x s - \x t\|^d = \frac{\w s}{\phi(s)\wmin n} \ge \frac{e^{w_0}}{n}$ as we are assuming $\phi(s) \le e^{-w_0}$. On the other hand it holds
$$(2w_0^4 r)^d = 2^d w_0^{O(1)}n^{-1}= o\left(e^{w_0} n^{-1}\right),$$
and by the triangle inequality it follows that for all $i \le 2w_0^4$ we have
$$\|\x {v_i} - \x t\|^d \ge \frac{e^{w_0}}{2 n}.$$

Now we add the vertices of weights between $w_0$ and $w_0^3$ to the graph in order to obtain $G''$. Notice that due to the Poisson point process, the distribution of the new vertices is completely independent from what we have observed in $G'$. The vertices with even higher weight will be uncovered later. Then we rerun the algorithm $A$ on $G''$, yielding a sequence $s''$ of different visited vertices. There are three cases. First, if $|s''| < 2w_0^4$ then we are done, since then on $G$ the algorithm either visits less than $2w_0^4$ different vertices, or it visits a vertex that is not in $G''$. Similarly, if $|s''| \geq 2w_0^4$, and among the first $2w_0^4$ vertices in $s''$ there is at least one with weight at least $w_0$, then we are also done: on $G$ the algorithm must either visit a vertex that is not in $G''$, or it visits the same first $2w_0^4$ different vertices as in $s''$, thus visiting a vertex of weight at least $w_0$.

%on $G''$, one of the first $2w_0^4$ visited vertices has weight at least $w_0$, we are done, since by uncovering the remaining vertices of $G$ (i.e., vertices of larger weight). If the algorithm on $G''$ visits less than $2w_0^4$ vertices, either it will not visit more vertices on the whole graph $G$, or when adding the heavy vertices later on the protocol $A$ will visit such a heavy neighbor. In this case, one of the first $2w_0^4$ vertices visited by the algorithm is a vertex of heavy weight as desired.

It remains the case where $|s''| \geq 2w_0^4$, and the first $2w_0^4$ vertices in $s''$ all have weight at most $w_0$. We have seen that in this case, it visits vertices in at least $w_0^3$ different cells of the $w_0^3$-grid. By the Bulk Lemma~\ref{lem:bulk}, we deduce that when adding the heavy vertices of weight at least $w_0^3$ to the graph, a.a.s.\ there exists a smallest index $i<2w_0^4$ such that $v_i$ has a neighbor $u$ with $\w u \ge w_0^3$ which is located in the same cell. We know that a.a.s.\
$\|\x {v_i} - \x t\|^d \ge \frac{e^{w_0}}{2 n}$. Since $u$ is in the same cell as $v_i$, it holds $\|\x {u} - \x t\|^d \le (1+o(1))\|\x {v_i} - \x t\|^d$ and therefore 
$$\phi(u) \ge (1-o(1))w_0^2 \phi(v_i).$$ 

On the other hand, we know that a.a.s.\ $v_i$ has no neighbor $u'$ satisfying both $\w u' < w_0$ and $\|\x {v_i} - \x u'\|\ge r$, therefore every neighbor $u''$ of $v$ with weight less than $w_0$ satisfies 
$$\|\x {v_i} - \x u''\| < r = o(\|\x {v_i} - \x t\|).$$ Again by the triangle inequality it follows that such a neighbor $u''$ satisfies 
$$\|\x {u''} - \x t\|=(1-o(1))\|\x {v_i} - \x t\|.$$ 
We conclude that in this case we have $$\phi(u'') \le (1+o(1)) w_0 \phi(v_i)$$
and indeed the objective of the vertex $u$ is large enough such that by (P1) the algorithm will prefer $u$ to all neighbors of $v_i$ of weight at most $w_0$. Hence, the neighbor of $v_i$ with maximal objective has weight at least $w_0$ and at the latest at this point, $A$ visits a vertex of weight at least $w_0$.
\end{proof}

\begin{proof}[Proof of Theorem~\ref{thm:length}]
Let $f_0=f_0(n)=\omega(1)$ be a growing function. Furthermore let $w_0=w_0(n) := \max\{\log f_0(n),\w s\}$ and $\phi_0 := \min\{(c_1 \w t)^{-1},f_0^{-1}\}$. 
We first assume $\phi(s) \le \phi_0$. Suppose $\w s \le w_0$. Then by Lemma~\ref{lem:patchingstart}, a.a.s.\ either greedy routing ends in a dead end after visiting $O(w_0^4)=o(f_0)$ vertices, or one of the first $O(w_0^4)$ vertices has weight at least $w_0$. If $\w s \ge w_0$, we already start at such a vertex. Let $u_1$ be the first visited vertex of weight at least $w_0$. Clearly $\phi(u_1)\ge\phi(s)$. Suppose $\phi(u_1)\le \phi_0$. Then we apply Lemma~\ref{lem:main} with $w_0$ and $\phi_0$. By Statements (v) and (vi), a.a.s.\ after at most 
$$L = \frac{1+o(1)}{|\log\beta-2|}\left(\log\log_{w_0} (\phi(s)^{-1})+\log\log_{\phi_0} (\phi(s)^{-1})\right)+O\left(f_0(n)\right)$$
steps, the routing process arrives at a vertex of objective at least $\phi_0$. Notice that by construction, $w_0\ge \w t$ and $\phi_0^{-1}\ge \w t$, thus $L$ is not larger than the claimed length of the greedy path. 

Hence we can assume that the routing reaches a vertex $u'$ such that $\phi(u')\ge\phi_0$. Note that in the case $\phi(s)\ge\phi_0$, we would use $s=u'$. If $\phi_0 = f_0^{-1}$, then we apply Lemma~\ref{lem:verticesofhighobjective} to see that a.a.s., there exist at most $O(f_0)$ vertices with objective at least $\phi_0$. Even if the algorithm visits every single vertex with this property, this needs only $O(f_0(n))$ steps. It remains to treat the case $\phi_0^{-1}=c_1\w t$. Then
$$c_1 \frac{\w {u'} \w t}{\wmin n \|\x {u'}-\x t\|^d} = c_1 \phi(u')\w t \ge 1,$$
and we observe that every vertex $u$ that the algorithm visits after $u'$ has probability $p_{ut}=\Omega(1)$ to directly connect to the target $t$. Suppose for contradiction that from now on, the routing process visits more than $f_0(n)=\omega(1)$ additional vertices. Since each of these vertices independently connects to $t$ with constant probability, it follows that with probability $1-o(1)$ at least one of the $f_0$ additional vertices connects to $t$. However, in this case the algorithm indeed goes to $t$ and stops. We conclude that after reaching a vertex of objective at least $\phi_0$, a.a.s.\ the routing stops before visiting $f_0(n)$ additional vertices. Adding up all steps then proves the second statement of the Theorem. For the first, general bound of $\frac{2+o(1)}{|\log \beta-2|} \log \log n$, we pick $f_0(n)=o(\log \log n)$. Then this statement follows immediately since $\phi(s) = \Omega(n^{-1})$ and $\w s, \w t = \Omega(1)$.

\end{proof}

\section{Proof: Patching}\label{sec:proofpatching}

%\subsection{Proof of Theorem~\ref{thm:patching} (Runtime of Patching Algorithm)} \label{subsec:proofpatching}

Similar to our analysis of the basic routing process, we prove Theorem~\ref{thm:patching} in three steps. We first use Lemma~\ref{lem:patchingstart} in order to show that the algorithm soon finds vertices of high weight. Next we can apply Lemma~\ref{lem:main} for analyzing the process until an objective $\phi_0$ is reached. For the end of the routing process, we need a preparatory result. We do this with the following lemma. It is analogue to Lemma~\ref{lem:patchingstart} but considers the very last phase before the algorithm hits the target vertex $t$.

\begin{lemma} \label{lem:patchingend}
There is a constant $c>0$ such that the following holds. Let $w_0=w_0(n)=\omega(1)$ be a function growing in $n$ such that $w_0(n) = O(\log\log n)$, and suppose that the target vertex $t$ satisfies $\w t \le w_0$ and that the start vertex satisfies $\phi(s) \leq w_0^{-c}$. Then a.a.s., either $s$ and $t$ are not in the same component, or there exists a path $P$ of length $O(w_0^4)$ from $t$ to a vertex $u$ with $\w u \ge w_0$ such that every vertex $v \in P$ (including $u$) fufills $\phi(v)\ge w_0^{-O(1)}$.
\end{lemma}
\begin{proof}
Let $G'$ be the subgraph induced by vertices of weight at most $w_0$, and suppose that we first uncover this subgraph. Furthermore let $C$ be the connected component of $G'$ which contains $t$, and let $r:= (w_0^c/n)^{1/d}$, where $c>0$ is a sufficiently large constant. By Lemma~\ref{lem:nojumps}~(ii), with probability $1-O(w_0^{-5})$ the target vertex $t$ has no neighbor in $G'$ with geometric distance at least $r$ to $\x t$. Next we explore $C$ by depth first search (choosing vertices in arbitrary order) until  either $2w_0^4$ vertices have been visited or $C$ has been completely traversed. Let $v$ be any vertex which has been explored in this process. By Lemma~\ref{lem:correlatedprobs} and Lemma~\ref{lem:nojumps}~(ii), with probability $1-O(w_0^{-5})$ the vertex $v$ has no neighbor in $G'$ with geometric distance at least $r$ to $\x v$. As we explored at most $2w_0^4$ vertices, by a union bound it follows that with probability $1-O(w_0^{-1})=1-o(1)$, we didn't explore any edge of geometric length at least $r$. Therefore, every explored vertex $v$ satisfies $\|\x v - \x t\|^d \le 2w_0^4 r^d= w_0^{O(1)}/n$ and thus $\phi(v) \ge w_0^{-O(1)}$.

We proceed by considering a $w_0^3$-grid as given by Definition~\ref{def:grid}, and uncovering the vertices of weight at least $w_0$. Similarly as in the proof of Lemma~\ref{lem:patchingstart}, we deduce that there are $w_0^{O(1)}$ cells of the grid which are in distance at most $2w_0^4 r$ to $\x t$, and a.a.s.\ all of these cells contain at most $2w_0^3$ vertices. It follows that if we explored $2w_0^4$ vertices, then they cover at least $w_0$ different cells. In this case, by Lemma~\ref{lem:bulk} there exists a vertex $u$ and a vertex $v$ such that (i) $\w u \ge w_0$, (ii) $u$ and $v$ are in the same cell, (iii) $v$ is among the explored vertices, and (iv) $u$ and $v$ are connected with an edge. From (ii) it follows in particular that $\phi(u) \ge w_0^{-O(1)}$, and clearly there exists a $v$-$t$-path among the explored vertices. Hence we find a $u$-$v$-$t$-path with all desired properties.

Suppose that the DFS stops after exploring less than $2w_0^4$ vertices, and let $v$ be an explored vertex. Recall that $\|\x v - \x t\|^d \le 2w_0^4 r^d= w_0^{O(1)}/n$. Then by Lemma~\ref{lem:nojumps}~(iii), with probability $1-O(w_0^{-5})$ there exists no neighbor $u$ of $v$ such that $\phi(u) \le w_0^{-c}$ and $\w u \ge w_0$, for a suitable choice of $c$. By a union bound we deduce that a.a.s.\ this holds for all explored vertices. However, this is only possible if either $C$ is itself the whole connected component of $t$ (and in this case $s$ is not in the same coomponent as $t$) or there exists a vertex $u$ with $\w u \ge w_0$ and $\phi(u) \ge w_0^{-O(1)}$ which has an \emph{explored} neighbor $v$. This yields again an $u$-$v$-$t$ path with all desired properties.
\end{proof}

\begin{proof}[Proof of Theorem~\ref{thm:patching}]~
We first treat the main case $\phi(s) = o(1)$, and only come to the exceptional case $\phi(s) = \Omega(1)$ at the very end.
\paragraph{The Case \boldmath{$\phi(s) = o(1)$}: finding a vertex $u_1$ of large weight.} %Let $f_0=f_0(n)=\omega(1)$ be a growing function, % such that $f_0=o(\log\log n)$. 
%and let $w_0 := \log\log (f_0)$. %and $w_0' := e^{-w_0}=(\log f_0)^{-1}$.
Let $w_0 = O(\log\log\log n)$. We assume for now that the starting vertex $s$ satisfies $\phi(s) =o(1)$, which happens a.a.s.\ if $s$ is a random vertex. Note that we may choose $w_0$ such that $\phi(s) \le e^{-w_0}$. Furthermore let $A$ be an algorithm satisfying $(P1)-(P3)$. We first show that either $s$ is not in the same connected component as $t$, or $A$ finds $t$ quickly, or $A$ visits soon a vertex of weight at least $w_0$. If $\w s \ge w_0$, this is trivial. Otherwise, either $A$ finds $t$ among the first $O(w_0^4)$ (different) visited vertices, or the connected component of $s$ has only $O(w_0^4)$ vertices, or the algorithm visits $O(w_0^4)$ without finding $t$. Note that by condition $(P2)$, the algorithm $A$ spends at most $w_0^{O(1)}=o(\log\log n)$ steps in this phase, which settles the first two cases. In the last case, $A$ visits $O(w_0^4)$ vertices without finding $t$, and by Lemma~\ref{lem:patchingstart} a.a.s.\ at least one of these vertices has weight at least $w_0$. Let us call the first such vertex $u_1$. 

\paragraph{The length of middle part: application of the main lemma.} As a next step, we want to apply Lemma~\ref{lem:main} to bound the number of steps in the middle phase. %Let $\delta$ be a sufficiently small constant, to be fixed later, 
Let $f_1(n)=\omega(1)$ be a function which grows sufficiently small in $n$, and let $w_1 := \log w_0$, $w_2 := \log  w_1=\log\log w_0$,
% \log (w_0),w_t^{\delta}\}$ \rk{currently, $w_1=(\log w_0)$} 
and $\phi_0 := (w_2)^{-1/f_1}=o(1)$. We consider the induced graph $G[V_{\leq \phi_0}]$ of vertices of objective at most $\phi_0$, and denote by $P=\{s, \ldots, u_1, \ldots\}$ the path obtained from running $A$ on $G[V_{\leq \phi_0}]$. We apply Lemma~\ref{lem:main} for $w_0$ and $\phi_0$. It follows that a.a.s., after $u_1$ the patching algorithm always finds an objective-improving neighbor until it reaches a vertex $u_{\ell}$ which has a.a.s.\! $\omega(1)$ neighbors in $V^+(u_{\ell},\varepsilon_1)$ of objective at least $\phi_0$. Let $P'$ be the subpath $\{u_1, \ldots, u_{\ell}\}$. By Statement~(v) of Lemma~\ref{lem:main}, and by our choice of $w_0$ and $\phi_0$, a.a.s.\ it holds
$$|P'|  \le \frac{1+o(1)}{|\log(\beta-2)|} (\log \log_{\w 0} \phi(u_1)^{-1} +  \log \log_{\phi_0^{-1}} \phi(u_1)^{-1}) +o(\log\log n) \le \frac{2+o(1)}{|\log \beta-2|} \log\log n.$$

\paragraph{Existence of $\mathbf{u_j,u_i}$: vertices of objective $\mathbf{\phi(u_j) \approx w_1^{-1}, \phi(u_i) \approx w_2^{-1}}$.} Next, we claim that a.a.s.\! there exist a vertex $u_j \in P'$ such that 
\begin{equation}\label{eq:propuprime}
w_1^{-\Omega(1)} \ge \phi(u_{j}) \ge w_1^{-1-\Omega(1)} ,
\end{equation}
and such that $A$ also visits $u_j$ if we let it run on $G \setminus \{t\}$ instead of $G[V_{\leq \phi_0}]$. 

By Lemma~\ref{lem:nojumps}~(i), there exists a constant $\varepsilon'>0$ such that a.a.s. the graph has no vertices of weight at least $w_0$ and objective at least $w_0^{-\varepsilon'}$ (except potentially $t$, if we fix $t$ and let $\w t$ be large). Moreover, we already know that the vertices which $A$ visited before $u_1$ have no neighbor with distance at least $r$ and weight less than $w_0$, thus for $\varepsilon'$ small enough, they have no neighbor of objective at least $w_0^{-\varepsilon'} < w_1^{-1-\Omega(1)}$. In particular it follows that $\phi(u_1) \le w_1^{-1-\Omega(1)}$. 

Now let $u \in P'$ be a vertex such that $\phi(u) \le w_1^{-1-\Omega(1)}$. Then by Lemma~\ref{lem:nojumps}~(iv), with probability $1-O(\min\{\w {u},\phi(u)^{-1}\}^{-\Omega(1)})$ the vertex $u$ has no neighbor of objective at least $w_1^{-\Omega(1)}$ and weight at most $w_1$. Moreover, by Lemma~\ref{lem:main}~(iii) the weights of the vertices in $P'$ increase doubly exponentially during the first phase, and by Lemma~\ref{lem:main}~(iv) the objectives increase doubly exponentially during the second phase. This allows us to apply union bound over all vertices on $P'$ with objective at most $w_1^{-1-\Omega(1)}$. We deduce that a.a.s., no vertex $u \in P'$ with $\phi(u)\le w_1^{-1-\Omega(1)}$ has a neighbor of objective at least $w_1^{-\Omega(1)}$ and weight at most $w_1$. By Lemma~\ref{lem:nojumps}~(i), we know that a.a.s.\ there are no vertices except $t$ with objective at least $w_1^{-\Omega(1)}$ and weight at least $w_1$. We finally see that indeed, a.a.s.\ the path $P'$ contains at least one vertex $u_j$ satisfying (\ref{eq:propuprime}). Moreover, for sufficiently large $n$ we have $w_1^{-\Omega(1)} \leq \phi_0$. Then, since none of the preceding vertices has a neighbor in $V_{> \phi_0}$, this vertex $u_j$ is also visited if $A$ runs on $G$. This proves that $u_j$ exists.
%Now we distinguish two cases. If $w_1 = w_t^\delta$ then~\eqref{eq:propuprime} implies $\phi(u_{i}) \ge w_t^{-1/2}$ by a suitable choice of $\delta$. In this case, $u_i$ has probability $\Omega(1)$ to connect directly to $t$, and so has every other vertex of larger objective. Moreover, all vertices in the subpath $(u_i,\ldots,u_{\ell})$ have objective at least $\phi(u_i)$, and $u_{\ell}$ has $\omega(1)$ neighbors of even higher objective. Therefore, in the graph $G[V_{\geq \phi(u_i)}]$ there are $\omega(1)$ vertices in the connected component of $u_i$, and all of them have probability $\Omega(1)$ to connect to $t$. By condition (P3), the algorithm will explore $k = \log f_0$ of them in the next $k^{O(1)} = o(f_0)$ steps, and a.a.s.\! at least one of them is connected to $t$. This concludes the case $w_1 = w_t^\delta$.

Let $w_2=\log(w_1)$. We claim that analogously to~\eqref{eq:propuprime}, a.a.s.\! there exist a vertex $u_i \in P'$ such that 
\begin{equation}\label{eq:propuprime2}
w_2^{-\Omega(1)} \ge \phi(u_{i}) \ge w_2^{-1-\Omega(1)},
\end{equation}
and such that $A$ also visits $u_i$ if we let it run on $G \setminus \{t\}$ instead of $G[V_{\leq \phi_0}]$. Indeed, the proof is completely analogous to~\eqref{eq:propuprime}.

\paragraph{End phase: finding a path from $\mathbf{u_i}$ to $\mathbf{t}$ in $\mathbf{G[V_{> \phi(u_j)}]}$.} To finish the argument, we consider the graph $G[V_{> \phi(u_j)}]$. Clearly, if $t$ is already connected to a vertex on the path $P$ which is located before $u_i$ on $P$, then the algorithm succeeds. Otherwise, we claim that it is sufficient to show that either $t$ is not in same connected component as $s$, or that among all vertices of objective at least $\phi(u_j)$ there exists a path from $u_j$ to $t$. Suppose we find such a path. Then $u_j$ and $t$ are in the same connected component $S$ in $G[V_{> \phi(u_{j})}]$. In particular, by (P3) the algorithm $A$ visits all vertices of $S$ (including $t$) during the next $|S|^{O(1)}$ steps. By Lemma~\ref{lem:verticesofhighobjective}, a.a.s.\ there are at most $O(\phi(u_j)^{-1}) = w_0^{O(1)}$ vertices of objective at least $\phi(u_j)$, therefore $|S|^{O(1)} = w_0^{O(1)}=o(\log\log n)$ as desired. Hence it remains to find a $u_j$-$t$-path in $G[V_{> \phi(u_{j})}]$. Since we already know that $u_i$ and $u_j$ are connected by $P'$, it also suffices to find a path from $u_i$ to $t$ in $G[V_{> \phi(u_{j})}]$. (Mind the two different indices $i, j$ here.)

We start by finding a path from $t$ to a vertex $u$ of weight at least $\w u \ge w_2^{\eps}$, where $\eps>0$ is a sufficiently small constant. If $\w t \ge w_2^{\eps}$, this is trivial by setting $u=t$, otherwise we apply Lemma~\ref{lem:patchingend} with weight $w_2^{\eps}$. Note that the condition $\phi(s)< w_2^{-c}$ is satisfied for sufficiently large $n$. Then either $s$ and $t$ are in different components, or there is a path from $t$ to a vertex $u$ with $\w u \geq w_2^{\eps}$ such that every vertex in the path has objective at least $w_2^{-O(1)} \geq \phi(u_j)$, where the latter inequality holds for sufficiently large $n$. So it only remains to show that $u$ and $u_i$ are in the same connected component of $G[V_{> \phi(u_{j})}]$. For this last step, fix a sufficiently large constant $C>0$ and consider the set $B := \{v \in \probSpace \mid \|\x v - \x t\|^d \le w_2^C/n\}$. Clearly, $B$ contains $t$. If $u \neq t$, then by Lemma~\ref{lem:nojumps}~(i), a.a.s.\ we have $\phi(u) \leq \w u^{-\Omega(1)}$, and thus $w_2^{-O(1)} \leq \phi(u) \leq \w u^{-\Omega(1)} \leq w_2^{-\Omega(1)}$. This bounds both $\w u$ and $\phi(u)^{-1}$, so we get $\|\x u - \x t\|^d = \phi(u)^{-1}\w u/n \leq w_2^{O(1)}/n$. Hence, if $C$ is large enough then $u \in B$. Similarly, using~\eqref{eq:propuprime2} we find that $u_i \in B$. The set $B$ contains in expectation $n' = w_2^C$ vertices, and by Corollary~\ref{cor:subcore} a.a.s.\! the induced graph $G[B]$ has a giant component which contains all vertices of weight at least $(\log n')^{O(1)}$. In particular, $u$ and $u_i$ are in the giant component of $G[B]$, so there is a path from $u$ to $u_i$ in $G[B]$. Moreover, every vertex $v$ in $B$ has objective at least $\phi(v) \geq \w u/(n\|\x v - \x t\|^d) \geq \wmin w_2^{-C} > \phi(u_j) $ (if $n$ is sufficiently large), so $B \subseteq V_{> \phi(u_j)}$. Altogether, we have found a path from $u_i$ to $u$ in $V_{> \phi(u_j)}$, and thus a path from $u_j$ to $t$ in $V_{> \phi(u_j)}$. This concludes the proof under the initial assumption
 $\phi(s) = o(1)$. 
 
%Suppose $\w t=\omega(1)$ and $\phi(s) = \omega(1)$. Clearly, the larger $\w t$ is, the sooner the algorithm $A$ reaches $t$. For this case, we can modify the previous analysis as follows: Again we find vertices $u_i,u_j$ on the path induced by $G[V_{\le \phi_0}]$. Thereby, $u_i$ lies in the giant component of the induced subgraph $G[B]$ which contains $\Omega(n')=\omega(1)$ vertices. Again, we have $B \subseteq V_{\ge \phi(u_j)}$. If we run the algorithm $A$ on $G \setminus \{t\}$, then by $(P3)$ it will visit all $\Omega(n')=\omega(1)$ before returning back to $u_j$. However, if we choose $f_0$ small enough compared to $\w t$ (and thus $\phi_0$ large enough compared to $\w t$), then for every vertex $u \in B$ it holds $p_{ut}=\Omega(1)$. When adding the edges incident to $t$ to the graph, it follows that a.a.s.\ there exists an $u \in B$ visited by $A$ such that $p_{ut}\in E$. It follows that a.a.s.\ the routing reaches the target, either via $u$ or via an earlier vertex. In particular, for $f_0(n)=o(\log\log n)$ the number of steps was at most $\tfrac{2+o(1)}{\log|\beta-2|}\log\log n$.

\paragraph{The Case \boldmath{$\phi(s) = \Omega(1)$}.} In the remaining case $\phi(s) = \Omega(1)$, condition (P3) ensures that there exists an objective $\phi_1 = o(1)$, $\phi_1 \geq (\log \log n)^{o(1)}$, such that we explore the connected component $S$ of $s$ in $G[V_{> \phi_1}]$ in $|S|^{O(1)}$ steps. We note that by Lemma~\ref{lem:verticesofhighobjective}, a.a.s.\! $|V_{> \phi_1}| =O(\phi_1^{-1})$, so in particular we need at most $|S|^{O(1)} = o(\log \log n)$ steps to explore $S$. It only remains to show that a.a.s.\! either $s$ and $t$ are in different connected components in $G$, or they are in the same component in $G[V_{> \phi_1}]$. This argument resembles closely the argument in the previous paragraph, and we only repeat it briefly. By Lemma~\ref{lem:patchingend} either $s$ and $t$ are in different components in $G$, or there is path from $s$ to $t$ among vertices of objective at least $\phi_1$ and of weight at most $\phi_1^{-\eps}$, or in $G[V_{> \phi_1}]$ there are paths from $s$ and from $t$ to vertices $u_s$ and $u_t$ of weights at least $\phi_1^{-\eps}$, where $u_s=s$ or $u_t=t$ are possible if $\w s$ or $\w t$ is large. Then by the same argument as we used for $u$ and $u_j$ above, a.a.s.\! $u_s$ and $u_t$ are in the same connected component of $G[V_{> \phi_1}]$. Summarizing, either there is no path from $s$ to $t$ in $G$, or there is a path from $s$ to $t$ in $G[V_{> \phi_1}]$, and in the latter case the path is explored in $o(\log \log n)$ steps, as required.

\end{proof}

\section{Proof: Relaxations} \label{sec:relaxations}

In our analysis of the basic greedy algorithm and the patching algorithm we always assumed that a node is aware of the exact objectives of its neighbors and packets are forwarded always to the neighbor with largest objective. In practice the exact objective of a neighbor might not be available. In the following we show that all our theorems also hold if objective values are only approximately known. 
In particular, we generalize from our original objective function $\phi$ to a class of objective functions, and prove Theorem~\ref{thm:relaxations}. Before we come to the proof, we remark without formal proof that Theorem~\ref{thm:relaxations} is essentially best possible.

\begin{remark}\label{rem:patching}
The relaxations in Theorem~\ref{thm:relaxations} are best possible. In Section~\ref{sec:proofsketch} we have described typical trajectories of greedy routing. In particular, it turned out that there are two phases in which the weight and the objective improve by an exponent of $1/(\beta-2)$, respectively. If we would replace the function $f_0(n)$ in Theorem~\ref{thm:relaxations} by a constant, then in each step the algorithm could pick vertices which would increase the weight (in the first phase) and the objective (in the second phase) only by an exponent $c < 1/(\beta-2)$. This would increase the number of steps to at least $\tfrac{2+o(1)}{\log c}\log \log n$, thus making the algorithm considerably slower. 
\end{remark}

Below we show that our analysis of the greedy routing w.r.t.\ $\phi$ applies as well for such a relaxed objective function $\tilde{\phi}$ which fufills the preconditions of Theorem~\ref{thm:relaxations}. Note that all technical definitions of sets such as $\probSpace_1$, $\probSpace_2$, $V^+(v,\varepsilon)$, $V^-(v,\varepsilon)$, \ldots are still defined via $\phi$ and not via $\tilde{\phi}$.

\begin{proof}[Proof of Theorem~\ref{thm:relaxations}]
We first study how Theorems~\ref{thm:greedysuccess1}, \ref{thm:greedysuccess2}, and~\ref{thm:length} for the basic routing process can be transferred for relaxed objective functions, and defer the corresponding analysis for patching algorithms to the end of this section. Let $f_0(n)=\omega(1)$ be a growing function, let $\delta>0$, and recall from our basic analysis that $\varepsilon_1$ is a constant chosen sufficiently small. Moreover we put $\varepsilon_2 = (\log\log f_0(n))^{-1}$. We want to prove for every $g_0(n)=o(1)$ that the results about greedy routing w.r.t.\ $\phi$ transfer to every relaxed objective function $\tilde{\phi}$ which is maximized at the target $t$ and satisfies
$$\tilde{\phi}(v) = \Theta\big( \phi(v) \cdot \min\{\w v, \phi(v)^{-1}\}^{\pm g(n)}\big)$$
for all $v$ with $\phi(v)  \le \overline{c}\w t^{-1+\delta}$, where $\overline{c}>0$ is a constant chosen sufficiently large later. Furthermore we assume that all other vertices satisfy $\tilde{\phi}(v) = \Omega(\w t^{-1+\delta/2})$. We observe that by choosing $f_0(n)$ sufficiently small compared to $g(n)$, we can assume as well that every vertex $v$ with $\phi(v) \le \overline{c}\w t^{-1+\delta}$ satisfies
\begin{equation}\label{eq:newrelax}
\tilde{\phi}(v) = \Theta\big( \phi(v) \cdot \min\{\w v^{\gamma(\varepsilon_1)},\phi(v)^{-1}\}^{\pm \varepsilon_2/c)}\big)
\end{equation}
for $c=4\gamma(\varepsilon_1)$,
which is a weaker relaxation and thus gives a stronger statement. Notice that for vertices $v \in \probSpace_1$, the minimum will be attained by $\w v^{\gamma(\varepsilon_1)}$, and vice-versa for vertices $v \in \probSpace_2$, the maximum will be attained by $\phi(v)^{-1}$. 

Let $\tilde{A}$ be a greedy routing algorithm which performs routing w.r.t.\ $\tilde{\phi}$. In the following, we show that we can modify our analysis of the routing process such that whenever $\tilde{A}$ proceeds to a new vertex $v$, this vertex still has all desired properties and the routing trajectory is essentially the same. We will describe how the proofs of Theorem~\ref{thm:greedysuccess1}, \ref{thm:greedysuccess2}, and \ref{thm:length} can be adapted for the relaxed routing algorithm $\tilde{A}$, but we will not prove this in detail. Recall that all these proofs basically contain three parts: First we studied the start of the routing, which contains the very first steps until a certain weight $w_0$ is reached. Then in the second step we applied Lemma~\ref{lem:main} in order to reach a vertex of objective at least $\phi_0$, and in the final part we studied how the process finds $t$. When describing how our results can be transferred to the relaxed algorithm $\tilde{A}$, we therefore handle these three parts separately. Thereby, we assume that $\phi_0 \le \underline{c} \w t^{-1+\delta/2}$, where $\underline{c}>0$ is a constant chosen sufficiently small later. Below we will see that never a larger $\phi_0$ for transferring the main theorems to $\tilde{A}$ is needed. We start by studying the main part of the routing process.

\para{Main part of routing process:} In all proofs we applied Lemma~\ref{lem:main} for analyzing the main part of the routing process. Hence we indicate why this main lemma still applies for the relaxed algorithm $\tilde{A}$. In fact, we will slightly strengthen part (vi) of the lemma such that for the end phase we can still guarantee that $\tilde{A}$ reaches a vertex of real objective at least $\phi_0$. Recall that we put $\phi_0 \le \underline{c} \w t^{-1+\delta/2}$. We require that $\underline{c}$ is small enough such that for every vertex $v \in V_{\le \phi_0}$ and every vertex $u$ with $\phi(u)\ge \overline{c} \w t^{-1+\delta}$, it holds $\tilde{\phi}(v) < \tilde{\phi}(u)$. Notice that by the weak relaxation \eqref{eq:weakrelaxation}, indeed this is possible. Then we can assume that whenever we use this lemma, all vertices satisfy (\ref{eq:newrelax}). 

Lemma~\ref{lem:main} treats two phases of the routing. In the first phase of the routing process, our goal was to prove for a given vertex $v \in \probSpace_1$ that its best neighbor w.r.t.\ $\phi$ has weight at least $\w v^{\gamma(\zeta\varepsilon)}$, where $\varepsilon \in \{\varepsilon_1,\varepsilon_2\}$ and this choice was depending on the weight $\w v$. Recall that our analysis assumed that all good events occur. Then in particular $v$ has no neighbor in the set $V^-(v,\varepsilon)$, and every neighbor $u$ with weight at most $\w v^{\gamma(\zeta\varepsilon)}$ can have objective at most $\phi(v) \w v^{\gamma(\varepsilon)-1}$. Let $u$ be such a neighbor. We claim that $\tilde{\phi}(u)$ is small enough such that $u$ is not the best neighbor of $v$ w.r.t.\ $\tilde{\phi}$. Due to the relaxation (\ref{eq:newrelax}), the vertex $u$ satisfies
$$\tilde{\phi}(u) \le \phi(u) \w u^{\gamma(\varepsilon_1)\varepsilon_2/c} \le O\big(\phi(v)\w v^{\gamma(\varepsilon)-1+\gamma(\zeta\varepsilon)\gamma(\varepsilon_1)\varepsilon_2/c}\big).$$

In order to guarantee that the routing does not proceed with such a vertex, we need to find a neighbor of weight at least $\w v^{\gamma(\zeta\varepsilon)}$ and higher relaxed objective $\tilde{\phi}$. Before, we found such good neighbors in the set $V^+(v,\varepsilon)$. Instead, we could as well consider the set $V^+(v,\varepsilon/3)$. We apply Lemma~\ref{lem:firstphase} to see that
$$\Ex[|\Gamma(v) \cap V^+(v,\varepsilon/3)|] = \Omega\big(\wmin^{\beta-2} \w v^{\Omega(\varepsilon)}\big)$$
which is fine for all error bounds.
Let $\gamma' := \frac{1+\varepsilon/3}{\beta-2}$. By Lemma~\ref{lem:weighttoohigh}, the expected number of neighbors with weight larger than $\w v^{\gamma'}$ is $O(\wmin^{\beta-2}\w v^{-\Omega(\varepsilon)})$. Hence we expect $\Omega(\wmin^{\beta-2} \w v^{\Omega(\varepsilon)})$ neighbors in $V^+(v,\varepsilon/3)$ with weight at most $\w v^{(1+\varepsilon/3)/(\beta-2)}$. By our relaxations, a vertex $u'$ in this set satisfies
$$\tilde{\phi}(u') \ge \phi(u') \w u'^{-\gamma(\varepsilon_1)\varepsilon_2/c} \ge \Omega\big(\phi(v) \w v^{\gamma(\varepsilon_i/3)-1-\gamma'\gamma(\varepsilon_1)\varepsilon_2/c}\big).$$
Since we assume that $\varepsilon_1$ is small enough and $\varepsilon \le \varepsilon_1$, we observe that
$$\gamma(\varepsilon/3)-\gamma(\varepsilon) = \frac{2\varepsilon}{3(\beta-2)} > \frac{(2+\varepsilon-\zeta\varepsilon)\varepsilon_2}{4(\beta-2)} = (\gamma(\zeta\varepsilon)+\gamma')\frac{\gamma(\varepsilon_1)\varepsilon_2}{c}$$
and indeed, $\tilde{\phi}(u')$ is larger than the relaxed objective of $v$ itself and of every neighbor of $v$ with weight less than $\w v^{\gamma(\zeta\varepsilon)}$. It follows that our analysis of the first phase is valid for routing w.r.t.\ $\tilde{\phi}$ as well and thus applies to $\tilde{A}$.

For the second phase in $\probSpace_2$, we can apply the same arguments in order to generalize our analysis. For $v \in \probSpace_2$, we consider again $V^+(v,\varepsilon/3)$ instead of $V^+(v,\varepsilon)$. Given all good events, the set $V^-(v,\varepsilon)$ is again empty and every vertex $u \in \Gamma(v) \cap \probSpace_1$ has relaxed objective at most
$\tilde{\phi}(u) = O(\phi(v)^{(1-\varepsilon_2/c)/\gamma(\varepsilon)})$. However, every vertex $u' \in \Gamma(v) \cap V^+(v,\varepsilon/3)$ with $\phi(u') \le \overline{c}\w t^{1-\delta}$ satisfies $\tilde{\phi}(u') = \Omega( \phi(v)^{(1+\varepsilon_2/c)/\gamma(\varepsilon/3)})$ by the first condition on the relaxation $\tilde{\phi}$. Note that since $c>4$ and $\varepsilon_2 \le \varepsilon$, we have $\varepsilon_2/c \le \varepsilon/3$. With a short calculation we deduce
$(1+\varepsilon_2/c)(1-\varepsilon) < (1-\varepsilon_2/c)(1-\varepsilon/3)$, which is equivalent to $$\frac{1+\varepsilon_2/c}{\gamma(\varepsilon/3)} < \frac{1-\varepsilon_2/c}{\gamma(\varepsilon)}.$$
Indeed, $\tilde{\phi}(u') > \tilde{\phi}(u)$ and the routing algorithm $\tilde{A}$ favours vertices in $V^+(v,\varepsilon)$ as desired. 

Moreover, regarding statement~(vi) of the main lemma, the above arguments imply that every vertex $v \in \Gamma(u_{\ell}) \cap V^+(u_{\ell},\varepsilon)$ with $\phi(v) \le \overline{c} \w t^{-1+\delta}$ satisfies $\tilde{\phi}(v) > \tilde{\phi}(u_{\ell})$. Furthermore, for every vertex $v \in \Gamma(u_{\ell}) \cap V^+(u_{\ell},\varepsilon)$ with $\phi(v) > \overline{c} \w t^{-1+\delta}$ the same is true by our choice of $\phi_0$, since $u_{\ell} \in V_{\le \phi_0}$. It follows
$$\{v \in \Gamma(u_\ell) \cap V^+(u_\ell,\varepsilon) \cap V_{> \phi_0} \mid \tilde\phi(v) > \tilde \phi(u_\ell)\} = \Gamma(u_\ell) \cap V^+(u_\ell,\varepsilon) \cap V_{> \phi_0}, $$
and using $M := \min\{w_0,\phi_0^{-1}\}$ we obtain
\begin{equation}\label{eq:relaxmainlemma}
\Ex_{> \phi_0}\big[\big|\{v \in \Gamma(u_\ell) \cap V^+(u_\ell,\varepsilon) \cap V_{> \phi_0} \mid \tilde\phi(v) > \tilde \phi(u_\ell)\}\big|\big]= \Omega\left(\wmin^{\beta-2}M^{\Omega(1)}\right)
\end{equation}
for a suitable choice of $\varepsilon \in \{\varepsilon_1,\varepsilon_2\}$. Note that the only difference to part (vi) of Lemma~\ref{lem:main} is the additional condition $\tilde\phi(v) > \tilde \phi(u_\ell)$. In particular, if $\tilde A$ visits $u_\ell$ and the set described in \eqref{eq:relaxmainlemma} is non-empty, then $\tilde A$ proceeds to a vertex in $V_{> \phi_0}$ which has in particular higher \emph{relaxed} objective than $u_{\ell}$. Notice that it is not guaranteed that this next vertex lies in $V^+(u_{\ell},\varepsilon)$, because vertices with real objective at least $\overline{c} \w t^{-1+\delta}$ can have arbitarily large relaxed objective, thus vertices in $V^+(u_{\ell},\varepsilon)$ don't need to be the best neighbors of $u_{\ell}$ w.r.t.\ $\tilde{\phi}$.

\para{Start of routing process:} We can assume that all considered vertices are contained in the set $V_{\le \phi_0}$, as the analysis of the end of the routing will handle all other vertices. Let us start by considering the proof of Theorem~\ref{thm:greedysuccess1}. There, we used that if $\w s < w_1(\varepsilon_1)$, then with constant probability $s$ has a neighbor in a superior set $A_s$. More precisely, $A_s$ contains vertices of weight at least $w_1(\varepsilon_1)$, and the proof required that with constant probability, $s$ has at least one neighbor $u_1 \in A_s$ and no neighbor with smaller weight than $w_1$ but higher objective than $u_1$. We observe that this argument can be transferred to the analysis $\tilde{A}$ as long as $\tilde{\phi}(u_1)$ is large enough and in particular larger than $\tilde{\phi}(s)$. However, due to (\ref{eq:newrelax}) it holds $\tilde{\phi}(u_1)=\Omega(\phi(u_1) \w {u_1}^{-o(1)})=\Omega(\phi(s))$. If we make the set $A_s$ slightly smaller and require that the weight of its vertices is by a constant factor larger than $w_1(\varepsilon_1)$, then still with constant probability there exists a vertex $u_1 \in \Gamma(s) \cap A_s$, and $\tilde{\phi}(u_1)$ will be large enough. Then this starting argument applies as well for $\tilde{\phi}$.

The proof of Theorem~\ref{thm:greedysuccess2} used Lemma~\ref{lem:expwminfirststeps} for the very first steps. There we constructed a set of layers $A_i$, defined via weights, and proved that with sufficiently high probability, the routing process traverses these layers and thereby the weight of the current vertex increases as desired. By the same arguments as seen above for the main lemma, we can apply  Lemma~\ref{lem:expwminfirststeps} for the relaxations.

When proving Theorem~\ref{thm:length}, we used Lemma~\ref{lem:patchingstart} for the starting phase.
The crucial argument behind Lemma~\ref{lem:patchingstart} was that after visiting many vertices of weight at most $w_0(n)=\omega(1)$, the routing will arrive at a vertex $v$ which has a neighbor $u$ in the same cell of a $w_0$-grid such that $\w u \ge w_0^3$. We observe that due to (\ref{eq:newrelax}), $u$ still has the property $\tilde{\phi}(u) \ge (1-o(1)) w_0^{2-o(1)} \phi(v)$. On the other hand, every neighbor $v'$ of $v$ which has weight less than $w_0$ has relaxed objective $\tilde{\phi}(v') \le (1+o(1)) w_0^{1+o(1)} \phi(v)$. We see that the since $w_0=\omega(1)$, the algorithm $\tilde{A}$ still prefers $u$ to $v'$. Therefore Lemma~\ref{lem:patchingstart} can be applied as well for studying the start of the routing algorithm $\tilde{A}$.

\para{End of routing process:} 
We have seen that as long as the visited vertices have real objective at most $\phi_0 \le \underline{c}\w t^{-1+\delta/2}$, the routing w.r.t.\ $\tilde{\phi}$ has the same properties as the original routing. It remains to study the end phase where the vertices have higher objective. We start by considering the proof of Theorem~\ref{thm:greedysuccess2}~(ii). There, we used $\phi_0 = \frac{1}{c_1 \w t}$. Since $\w t=\omega(1)$ by assumption, we have $\w t^{-1}=o(\w t^{-1+\delta})$. Hence, there is nothing to change. For adapting the proof of Theorem~\ref{thm:greedysuccess2}~(i), we recall that we are assuming that $\wmin$ is sufficiently large and thus $\w t$ as well. In particular, we can assume that $(c_1 \w t)^{-1} \le \underline{c} \w t^{-1+\delta/2}$. Then, we slightly change this proof by setting
$$\phi_0:=\min\{(c_1 \w t)^{-1}, \exp(-\wmin^{c'})\},$$
where $c'>0$ is the same large constant used previously in the proof of Theorem~\ref{thm:greedysuccess2}~(i).  
We distinguish two cases: If $\phi_0 = (c_1 \w t)^{-1}$, then our choice of $\phi_0$ is valid for applying the main lemma as well with relaxations. We also know by~\eqref{eq:relaxmainlemma} that in this case the algorithm $\tilde{A}$ reaches a vertex $u$ satisfying $\phi(u) \ge (c_1 \w t)^{-1}$. In this case, $p_{ut}=1$, and the routing finds $u$. On the other hand, if $\phi_0 = \exp(-\wmin^{c'})$, then it follows that $\w t = O(1)$. In this case, we need the sequence of additional layers defined via $\phi$, as given in the proof of the theorem, until we reach real objective $(c_1 \w t)^{-1}$. Notice that for all these vertices in-between the relaxation (\ref{eq:newrelax}) holds. That is, every vertex $v$ with $\phi_0 \le \phi(v) \le \frac{1}{c_1 \w t}$ satisfies $\tilde{\phi}(v) = \Theta(\phi(v))$. As the two objective functions are very similar in this additional set of layers, all arguments can be transferred to $\tilde{\phi}$, and then Theorem~\ref{thm:greedysuccess2} holds as well for the algorithm $\tilde{A}$.

For adapting the proof of Theorem~\ref{thm:greedysuccess1}, we assume again that $\underline{c}$ is chosen sufficiently small and $\overline{c}$ is chosen sufficiently large. Then we distinguish again two cases. If $c_1 \underline{c} \w t^{\delta/2} \ge 1$, we would like to choose $\phi_0 := (c_1 \w t)^{-1}$. First we have $\underline{c} \w t^{-1+\delta/2} \ge (c_1 \w t)^{-1}$, thus our choice is valid in terms of applying the main lemma with relaxations. Second, for $\underline{c}$ small enough it also follows that the error probability given by the main lemma is a sufficiently small constant. Hence we can take $\phi_0$ as claimed before, and by the same arguments as above we see that once the routing reaches a vertex $u$ of real objective at least $\phi_0$, it connects to $t$ with constant probability, which is already sufficient for the theorem. On the other hand, if $c_1 \underline{c} \w t^{\delta/2} < 1$, then $\w t = O(1)$. In this case, we choose $\phi_0=\Omega(1)$ small enough such that (a) $\phi_0$ is still a valid choice, i.e., $\phi_0 \le \underline{c} \w t^{-1+\delta/2}$, and such that (b) the error probability when applying Lemma~\ref{lem:main} is still small enough. Then we know that $\tilde{A}$ founds at least a vertex $u$ such that $\phi(u) \ge \phi_0=\Omega(1)$. However, since $c_1 \underline{c} \w t^{\delta/2} < 1$, it follows that all vertices with real objective at most $\overline{c} (c_1 \underline{c})^{2(1-\delta)/\delta}$ satisfy the relaxation (\ref{eq:newrelax}), which means that in this regime, $\phi$ and $\tilde{\phi}$ deviate by at most a bounded constant factor.  Recall that in the proof of Theorem~\ref{thm:greedysuccess1}, we constructed a superior set $A_t$. By picking $\overline{c}$ large enough, with constant probability $A_t$ contains vertices which still satisfy (\ref{eq:newrelax}). Then again, all arguments can be transferred to $\tilde{\phi}$.

For adapting the proof of Theorem~\ref{thm:length}, we slightly modify our choice of $\phi_0$ compared to the original proof and use $\phi_0 := \min\{(c_1 \w t)^{-1+\varepsilon_2/c},f_0^{-1+\varepsilon_2/c}\}=o(1)$. Note that since $\varepsilon_2=o(1)$, we satisfy $\phi_0 \le (\underline{c} \w t^{-1+\delta/2})$ and thus the choice is valid. If the minimum is attained with the first term, then as soon as $\tilde{A}$ reaches a vertex of real objective at least $\phi_0$, it will never visit a vertex of real objective less than $(c_1 \w t)^{-1}$ later on, and then the argument from the proof can be adapted. On the other hand, if the minimum is attained with the second term, then again $\tilde{A}$ will stay at vertices of real objective at least $f_0^{-1}$, once it reached a vertex of real objective at least $\phi_0$. And then again the proof of Theorem~\ref{thm:length} can be adapted. 

\para{Patching algorithms:} 
We first want to explain why the situation for patching algorithms is slightly different than for the basic routing algorithm, and why we can not transfer the proof of Theorem~\ref{thm:patching} in the case where $\tilde{\phi}$ satisfies only the weaker relaxation and $t$ can be an arbitrary vertex. Above, we have seen that if we only assume the weak relaxation \eqref{eq:weakrelaxation}, then we need to take $\phi_0 \le \underline{c} \w t^{-1+\delta/2}$ when applying the main lemma. However, our proof of Theorem~\ref{thm:patching} requires that $\phi_0=o(1)$ falls arbitrarily slow in $n$. We see that in a model where $\w t$ can be chosen arbitrarily large, our choice of $\phi_0$ is no longer valid when considering relaxations. Vice-versa, if $t$ is a random vertex, then a.a.s.\ the weight $\w t$ is small enough such that this conflict does not appear. Similarly, if we assume the stronger relaxation \eqref{eq:relaxation} for  \emph{all} vertices, then it is not difficult to observe that we get rid of the restriction $\phi_0 \le \underline{c} \w t^{-1+\delta/2}$ and our choice of $\phi_0$ is also valid when considering the relaxations.

For the start, we analyzed patching algorithms by applying Lemma~\ref{lem:patchingstart}. Above, we have already seen that this Lemma can be transferred for relaxed objectives $\tilde{\phi}$. For the main part of the routing, we argued before that under the additional assumptions, we still can apply the main lemma. Thus it remains to study the end of the routing process. There, we proved the existence of two specific vertices $u_j$ and $u_i$ on the Greedy path. We required that $\phi(u_j)$ and $\phi(u_i)$ lie in a specific range. We observe that since $u_j,u_i \in V_{\le \phi_0}$, the stronger relaxation \eqref{eq:relaxation} holds in the considered ranges. Therefore, $\tilde{\phi}$ deviates only by an $o(1)$-exponent from $\phi$, and we don't need to modify the arguments. Finally, we concluded the proof by finding a path from $u_i$ to $t$ inside a certain subset $B$. Clearly, this path is still present in the graph, and again the relaxations ensure that $B \subset V_{> \phi(u_j)}$.  This allows us to transfer the analysis of patching algorithms to relaxed objective functions $\tilde{\phi}$.
\end{proof}

\section{Proof: Greedy Routing on Hyperbolic Random Graphs} \label{sec:hyperbolic}
Hyperbolic random graphs were introduced by Krioukov et al.\ \cite{KrioukovPKVB} and have attracted a lot of attention during the last years, both in theoretical and experimental studies. Let us give a short formal explanation of the model, following the notation introduced by Gugelmann et al.\ \cite{gugelmann2012random}. It can be described by a disk $H$ of radius $R$ around the origin~$0$, where the position of every point $x$ is given by its polar coordinates $(r_x,\nu_x)$. The model is isotropic around the origin. Then the hyperbolic distance between two points $x$ and $y$ is given by the non-negative solution $d_H=d_H(x,y)$ of the equation
$$
\cosh(d_H) = \cosh(r_x)\cosh(r_y)-\sinh(r_x)\sinh(r_y)\cos(\nu_x-\nu_y).
$$
Then hyperbolic random graphs are defined as follows.
\begin{definition} \label{def:hyperrand}
Let $\alpha_H,T_H>0, C_H\in \R,n\in \N$, and set $R:=2\log n+C_H$. The hyperbolic random graph $G_{\alpha_H,C_H,T_H}(n)$ is a distribution on graphs with vertex set $V=[n]$ s.t.

\begin{itemize}
\item Every vertex $v \in [n]$ independently draws random coordinates $(r_v,\nu_v)$, where the angle $\nu_v$ is chosen uniformly at random in $[0,2\pi)$ and the radius $r_v \in [0,R]$ is random with density $f(r) := \frac{\alpha_H\sinh(\alpha_H r)}{\cosh(\alpha_H R)-1}$.
\item Every potential edge $e=\{u,v\}$, $u,v \in [n]$, is independently present with probability
$$p_H(d_H(u,v)) = \left(1+e^{\frac{1}{2T_H}(d_H(u,v)-R)}\right)^{-1}.$$
\end{itemize}

In the limit $T_H \rightarrow 0$, we obtain the threshold hyperbolic random graph $G_{\alpha_H,C_H}(n)$, where every edge $e=\{u,v\}$ is present if and only if $d_H(u,v) \le R$.
\end{definition}

In~\cite[Theorem~6.3]{bringmann2015euclideanGIRG} it was shown that such graphs are a special case of GIRGs. However, the mapping is a bit counterintuitive, so we describe it briefly here. In particular, although hyperbolic random graphs are defined on a two-dimensional hyperbolic disc, they give rise to a one-dimensional GIRG. The reason is that in GIRGs the weights provide one additional dimension. Notice that a single point on the hyperbolic disk has measure zero, so we can assume that no vertex has radius $r_v=0$. For the parameters, we put 
$$d:=1,\quad\beta := 2\alpha_H + 1,\quad\alpha := 1/T_H,\quad \wmin := e^{-C_H/2}.$$ 
Then the mapping that transforms hyperbolic random graphs into the GIRG framework is given by
$$\w{v} := ne^{-r_v/2} \quad\text{and}\quad \x{v} := \frac{\nu_v}{2\pi}.$$
Since this is a bijection between $H \setminus \{0\}$ and $[1,e^{R/2}) \times \mathds{T}^1$, there exists an inverse function $g(\w{u},\x{u})=(r_u,\nu_u)$. Finally for any two vertices $u \neq v$ on the torus, we set 
$$p_{uv} := p_H(d(g(u),g(v))) = p_H(d_H(g(\w{u},\x{u}),g(\w{v},\x{v}))).$$ 
This finishes our embedding. In \cite{bringmann2015euclideanGIRG} we proved that it is a special case of GIRGs, in a rigorous sense. Compared to~\cite{bringmann2015euclideanGIRG}, we have slightly changed the embedding such that we do not loose the freedom of choosing $\wmin$.

Greedy routing on hyperbolic random graphs is completely geometric: From a vertex $v$, the packet is sent to the neighbor $u \in \Gamma(v)$ which minimizes the (hyperbolic) distance $d_H(u,t)$. Notice that minimizing $d_H(u,t)$ is equivalent to maximizing $(\cosh(d_H(u,t)))^{-1/2}$, as the function $\cosh$ is monotone increasing on positive values, and we are free to multiply this expression by any terms which are independent of $u$. 

Let us assume that $d=1$ and that a target vertex $t$ is given. Then on such a GIRG, we can define the objective function \label{defn:phi_H}
$$\phi_H(v) := \frac{n}{\w t \wmin \sqrt{\cosh(d_H(g(v),g(t)))}},$$
where $g(v)$ and $g(t)$ are the vertices on the hyperbolic disk that correspond to $u$ and $t$. Suppose that the one-dimensional GIRG was obtained by embedding a hyperbolic random graph, as described above. Then, maximizing $\phi_H$ on the GIRG is the same as minimizing the hyperbolic distance to the target on the original hyperbolic graph. Thus routing w.r.t.\ $\phi_H$ on the GIRG is equivalent to geometric greedy routing on hyperbolic random graphs. The following lemma states that for $d=1$, with high probability, for most vertices the objective function $\phi_H$ differs by at most a constant factor from $\phi$.

\begin{lemma} \label{lem:hypobjective}
Let $d=1$ and let $\delta>0$ be a sufficiently small constant. Then with probability $1-n^{-\Omega(1)}$, for every vertex $v$ with $\phi(v) \le O(\w t^{-1+\delta})$ it holds
$$\phi_H(v) = \Theta(\phi(v)),$$
and for every vertex $v$ with $\phi(v) = \Omega(\w t^{-1+\delta})$ it holds
$$\phi_H(v) = \Omega(\w t^{-1+\delta}).$$
\end{lemma}

The statement of this lemma directly implies Corollary~\ref{cor:hyperbolic}.
\begin{proof}[Proof of Corollary~\ref{cor:hyperbolic}]
By Lemma~\ref{lem:hypobjective}, with probability $1-n^{-\Omega(1)}$ the objective function $\phi_H$ falls into the general class of objective functions considered in Theorem~\ref{thm:relaxations}. Note that the small error probability $n^{-\Omega(1)}$ coming from Lemma~\ref{lem:hypobjective} is negligible compared to the error probabilities of all previous theorems. Then the statement of the corollary follows directly from Theorem~\ref{thm:relaxations}, that is, our results transfer from GIRGs to hyperbolic random graphs. 
\end{proof}

\begin{proof}[Proof of Lemma~\ref{lem:hypobjective}]
Let $\delta>0$ and put $\phi_0 := O(\w t^{-1+\delta})$. For the proof of this lemma, we assume that there exists no vertex with weight larger than $n^{1-2\delta}$, which happens with probability $1-n^{-\Omega(1)}$ if $\delta$ is chosen small enough. Let $v \in \probSpace$ such that $\phi(v) \le \phi_0$, and denote by $(r_v, \nu_v) := g(\w v, \x v)$ and $(r_t, \nu_t) := g(\w t, \x t)$ the mappings of $v$ and $t$ onto the hyperbolic disk, and assume without loss of generality that $\nu_t=0$ and $\nu_v \le \pi$. We claim that on the hyperbolic disk, the assumption $\phi(v) \le \phi_0$ implies the geometric property 
\begin{equation} \label{eq:hypprop}
\nu_v \ge e^{-\min\{r_v,r_t\}}.
\end{equation}

We first show $\nu_v \ge e^{-r_v}$. By assumption $\phi(v)=O(1)$, and then for $n$ large enough our mapping implies
$$\nu_v = 2\pi \|\x v-\x t\|=\frac{2\pi\w v}{\wmin n \phi(v)} \ge \w v n^{-1-2\delta}.$$
Furthermore $\w v \le n^{1-2\delta}$ as there are no veritces of higher weight, and thus
$$\nu_v   \ge \w v n^{-1-2\delta} \ge \w v^2n^{-2}=e^{-r_v}.$$

Next, we observe that because there are no vertices of weight higher than $n^{1-2\delta}$, we have $\w t^{1+\delta} \le n^{(1+\delta)(1-2\delta)}=o(n)$ and thus $\w t^{1-\delta}n^{-1} = \omega(\w t^2 n^{-2})$. On the other hand we are assuming $\phi(v)=O(\w t^{-1+\delta})$. Then for $n$ large enough, we obtain
$$\nu_v = \frac{2\pi\w v}{\wmin n \phi(v)} \ge \frac{2\pi}{n\phi(v)} \ge \w t^2 n^{-2} = e^{-r_t}$$
which proves (\ref{eq:hypprop}).

In the following we estimate the hyperbolic cosine in terms of $\w v$, $\w t$ and $\|\x v - \x t\|$. We apply the definition of $d_h$ and the identity 
$$\cosh(x-y)=\cosh(x)\cosh(y)-\sinh(x)\sinh(y),$$ 
which yields
\begin{align*}
\cosh(d_H(g(v),g(t))) &= \cosh(r_v)\cosh(r_t)-\sinh(r_v)\sinh(r_t)\cos(\nu_v)\\
&= \cosh(r_v-r_t) + (1-\cos(\nu_v))\sinh(r_v)\sinh(r_t).
\end{align*}
Mapping our assumption $\w v \le n^{1-2\delta}$ on the hyperbolic geometry gives $r_v \ge 2\delta\log n$, therefore we get $\sinh(r_v)=(1-o(1))e^{r_v}/2$. Similarly it follows $\sinh(r_t)=(1-o(1))e^{r_t}/2$. Then
\begin{align*}
\cosh(d_H(g(v),g(t))) &= \cosh(r_v-r_t)+(1-\cos(\nu_v))(1-o(1))e^{r_v+r_t}/4\\&=\cosh(r_v-r_t)+\Theta(\nu_v^2 e^{r_v+r_t}),
\end{align*}
where we applied the Taylor approximation $1-\cos(\nu_v)=\nu_v^2/2-\nu_v^4/24+O(\nu_v^6)$. Finally, Equation~(\ref{eq:hypprop}) implies $\cosh(r_v-r_t) =O(\nu_v^2 e^{r_v+r_t})$. It follows
\begin{equation}\label{eq:hypcalc}
\phi_H(v) = \Theta\left(\frac{e^{r_t/2}}{\wmin \nu_v e^{(r_v+r_t)/2}}\right)=\Theta\left(\frac{n e^{-r_v/2}}{\wmin n \nu_v}\right)=\Theta\left(\frac{\w v}{\wmin n \|\x v - \x t\|^d}\right)=\Theta(\phi(v)).
\end{equation}
This proves the first part of the lemma. For the second part, let $v$ be a vertex such that $\phi(v) \ge \phi_0$. We try to add an artificial vertex $u$ of weight $\w u = \w v$ in the graph such that $\phi(u) = \phi_0$. If this is not possible because $\w v$ is too large, then we place $u$ such that $\|\x u - \x t\|=1/2$. In both cases $u$ has larger geometric distance to $t$ than $v$. If $\phi(u) = \phi_0$, we know by the first part of the proof that $\phi_H(u)=\Theta(\phi_0)$. In the second case, $\nu_u=\Omega(1)$ and therefore it holds
$$\cosh(d_H(g(u),g(t))) = \cosh(r_u - r_t) + \Theta(\nu_u^2 e^{r_u+r_t}) = \Theta(\nu_u^2 e^{r_u+r_t})= \Theta(e^{r_u+r_t}),$$
and by (\ref{eq:hypcalc}) again we obtain $\phi_H(u)=\Theta(\phi(u))=\Omega(\phi_0)$. Considering $g(v)$ and $g(u)$, we see that $\nu_u \ge \nu_v$. Since the hyperbolic distance in monotone decreasing in $\nu$, it follows $$d_H(g(v),g(t)) \le d_H(g(u),g(t))$$ and thus in both cases it holds $\phi_H(v) \ge \phi_H(u) = \Theta(\phi_0)$. 
\end{proof}

\clearpage
%\begin{footnotesize}
\bibliographystyle{plain}
\bibliography{../girg}
%\end{footnotesize}

\end{document}